%% file: main.tex
\documentclass[11pt]{amsart}
\usepackage[foot]{amsaddr}  

\usepackage{amssymb,amsmath,amsthm,amsfonts}
\usepackage{mathtools,mathrsfs}
\usepackage{setspace}  

\usepackage{textgreek}
\usepackage[hmargin=3cm,vmargin=3.5cm]{geometry}
\usepackage[dvipsnames,table,xcdraw]{xcolor}
\usepackage{url}
\usepackage[colorlinks = true,
            linkcolor = Fuchsia,
            urlcolor  = ForestGreen,
            citecolor = Plum,
            anchorcolor = blue]{hyperref}
\usepackage{graphicx} 
\usepackage{longtable, relsize}
\usepackage{float}
\usepackage{tikz-cd}
\usetikzlibrary{arrows,automata}
\usetikzlibrary{decorations.markings}

\usepackage{nomencl}  

\makenomenclature
\usepackage{etoolbox}

\renewcommand\nomgroup[1]{%
  \item[\bfseries
  \ifstrequal{#1}{D}{Complex Distributions}{%
  \ifstrequal{#1}{N}{Simple Distributions}{%
  \ifstrequal{#1}{O}{Symbols}{}}}%
]}

\input{main-command}

\toggletrue{fullpaper}

\iftoggle{fullpaper}{
\title{Generalization of the Alpha-Stable Distribution with the Degree of Freedom 
       }
}{
\title{Generalization of the Alpha-Stable Distribution with the Degree of Freedom 
       (SHORT VERSION)}
}

\author[S. H. Lihn]{Stephen H. Lihn}
\email{
\href{mailto:stevelihn@gmail.com}{stevelihn@gmail.com}, 
\href{mailto:sl0815@princeton.edu}{sl0815@princeton.edu}
}
\thanks{Disclaimer:
Mr. Lihn is a quant researcher at Atom Investors LP, Austin, TX 78746, USA.
This mathematical treatise is not related to the business of Atom Investors LP.
I attest there is no conflict of interest.
}


\date{May 12, 2024. v1.0.0512.
    Dedicated to my mother on the Mother's Day of 2024. 
    \iftoggle{fullpaper}{Full Version}{Main Results Only}
}

\providecommand{\keywords}[1]{\textbf{\textit{Key words and phrases.}} #1}


\begin{document}

\input{section-abstract}

\maketitle

\iftoggle{fullpaper}{
}{}

%
%

\input{section-intro}

\input{section-mixture}

\input{sub-symbols}

\input{section-main-result}

\input{section-main-result2}

\input{section-multivariate}

\iftoggle{fullpaper}{

\input{section-gsc}

\input{section-skew-kernel}

\input{section-fcm}

\input{section-gsas} 
\input{section-gexppow} 
\input{section-gas}

\input{section-gsc-rv}

\input{section-summary}


\appendix
\input{xapp-gsc-mgf}
\input{xapp-subord}

\input{xapp-wright}

\input{xapp-formula}

}


\bibliographystyle{plain} 
\bibliography{main} 

\end{document}

%% file: main-command.tex
\theoremstyle{definition}
\newtheorem{definition}{Definition}[section]
\newtheorem{lemma}{Lemma}[section]

\def\lt{<}
\def\gt{>}
\DeclareMathOperator{\sgn}{sgn}
\DeclareMathOperator{\diag}{diag}

\def\scN{\mathfrak{N}}  
\def\mcN{\mathcal{N}}  
\DeclareMathOperator\erf{erf}  
\def\pdfGG{f_{\text{\tiny GG}}}  
\def\alphahalf{\frac{\alpha}{2}}  
\def\GammaAlpha{\Gamma\left(\frac{1}{\alpha} + 1\right)}
\newcommand{\GammaAlphaN}[1]{\Gamma\left(\frac{#1}{\alpha} + 1\right)}
\def\R{\mathbb{R}}
\def\E{\mathbb{E}}
\def\Rn0{\mathbb{R}_{\ne 0}}
\newcommand{\FWrfn}[2]{{W}_{{-#1},0}\left(-{#2}\right)}  
\newcommand{\gscN}[1]{\scN_\alpha(#1; \sigma, d, p)}
\def\gscNx{\gscN{x}}
\def\gscNs{\gscN{s}}

\def\gsasDist{L_{\alpha,k}}
\def\gasDist{L_{\alpha,k}^\theta}
\def\afstableDist{L_{\alpha}^\theta}
\def\sasDist{P_{\alpha}}
\def\gskewDist{g_{\alpha}^\theta}  

\newcommand{\gsas}[1]{\gsasDist({#1})} 
\newcommand{\gas}[1]{\gasDist({#1})} 
\newcommand{\afstable}[1]{\afstableDist({#1})} 
\newcommand{\sas}[1]{\sasDist({#1})} 
\def\gskew{\gskewDist(x,s)} 

\def\chibar{\overline{\chi}} 
\def\chimeanDist{\chibar_{\alpha,k}}
\def\chimeanDistNegK{\chibar_{\alpha,-k}}
\newcommand{\chimean}[1]{\chimeanDist({#1})} 

\def\gepDist{\mathcal{E}_{\alpha,k}}

\def\fcmSigma{{\sigma}_{\alpha,k}}
\def\gsasSigma{\widetilde{\sigma}_{\alpha,k}}

\def\where{\,\, \text{where} \,\,}
\def\qedbar{\centering{\rule{0.5\linewidth}{0.3pt}}}  

\newtoggle{fullpaper}

%% file: section-abstract.tex
\begin{abstract} 
A Wright function based framework is proposed to combine and extend 
several distribution families.
The \(\alpha\)-stable distribution is generalized
by adding the \emph{degree of freedom} parameter.
The PDF of this two-sided super distribution family 
subsumes those of the original \(\alpha\)-stable,
Student's t distributions, as well as 
the exponential power distribution and 
the modified Bessel function of the second kind.
Its CDF leads to a fractional extension of the Gauss hypergeometric function.
The degree of freedom makes possible for valid variance, skewness, and kurtosis, 
just like Student's t. 
The original \(\alpha\)-stable distribution
is viewed as having one \emph{degree of freedom},
that explains why it lacks most of the moments.
A skew-Gaussian kernel is derived from the characteristic function of 
the \(\alpha\)-stable law,
which maximally preserves the law in the new framework.  
To facilitate such framework, the stable count distribution is generalized 
as the fractional extension of the generalized gamma distribution. 
It provides rich subordination capabilities, one of which is the fractional
\(\chi\) distribution that supplies the needed 'degree of freedom' parameter. 
Hence, the "new" \(\alpha\)-stable distribution is a "ratio distribution" of 
the skew-Gaussian kernel 
and the fractional \(\chi\) distribution.
Mathematically, it is a new form of higher transcendental function
under the Wright function family.
Last, the new univariate symmetric distribution is extended to 
the multivariate elliptical distribution
successfully.
\end{abstract}

%% file: section-intro.tex
\section{Introduction}
\label{section:intro}

The aim of this work is to propose a two-sided, super distribution family, 
called \emph{the generalized \(\alpha\)-stable distribution} (GAS), that subsumes three 
major distributions, plus one transcendental function:
\begin{itemize}
  \item the \(\alpha\)-stable distribution,
  \item Student's t distribution,
  \item the exponential power distribution, and
  \item the modified Bessel function of the second kind\footnote{It is the characteristic function of Student's t distribution}.
\end{itemize}
The literature review is as follows.

In 1908, William S. Gosset published a paper on measuring the distribution on samples of finite size 
under the pseudonym "Student"\cite{Student:1908}.
He pioneered a distribution later called \emph{Student's t distribution} 
or simply \emph{the t distribution}, that plays a crucial role in statistical inferences
and hypothesis testing.
The distribution has a major parameter called \emph{degrees of freedom} \(k\), 
hence the notation \(t_k\),
with the Cauchy distribution at \(k=1\),
and the standard normal distribution \(\mcN = \mcN(0,1)\) at \(k \to \infty\). 

Student's t distribution can be constructed as the ratio 
of a standard normal variable and the square root of a \(\chi^2\) distribution,
divided by degrees of freedom. That is, \(t_k \sim \mcN / \sqrt{\chi_k^2/k}\). 
The \(\chi^2\) distribution was first derived by F. R. Helmert in 1872
from his work on least squares adjustment and error propagation\cite{Helmert:1872}.
Later in 1900, Karl Pearson defined the \(\chi\) distribution as the square root of 
a \(\chi^2\) variable: \(\chi_{k} \sim \sqrt{\chi_k^2}\).

This history is summarised in our first and central guiding equation:
\begin{align}\label{eq:intro-t}
t_k
    &\sim \mcN  / \,\chibar_{k}, \quad \quad \where \chibar_{k} := \chi_k / \sqrt{k}. 
\end{align}

In 1925, Paul L\'evy published his seminal book on the \(\alpha\)-stable distribution\cite{Levy:1925}.
The distribution has a major parameter, among others, called \emph{the stability index} \(\alpha \in (0,2]\). 
In this work, we call the extension to the \(\alpha\) dimension \emph{fractional}.

The \(\alpha\)-stable distribution family also contains the Cauchy distribution at \(\alpha = 1\),
and the normal distribution at \(\alpha = 2\). 
Are the two commonalities between Student's t and \(\alpha\)-stable coincident? 

A main goal of this paper is to show that 
this is not a coincidence -- 
There is a super distribution family 
that subsumes both distributions seamlessly. 
However, its construction requires the fractional extension of the \(\chi\) distribution:
\(\chi_{k} \to \chi_{\alpha,k}\),
and a branch of mathematical functions called
\emph{higher transcendental function}.

In 1933, E. M. Wright published his work on the asymptotic theory of partitions\cite{Wright:1933, Wright:1935} that proposed
a particular type of higher transcendental function, with the notation of 
 \(W_{\lambda,\delta}(z)\) where \(\lambda,\delta\) are two shape parameters.
It was later collected into the Bateman manuscript in 1955,
but tucked under the Chapter of the Mittag-Leffler function
(Chapter 18.1 of Vol 3)\cite{Bateman:1955}.
In recent decades, this function is found to have an essential role
in the theory of fractional calculus and
the \(\alpha\)-stable distribution, hence is called \emph{the Wright function}.

In 1952, William Feller pioneered a new stable parameterization \((\alpha,\theta)\)
in his work on the Riesz-Feller fractional derivative\cite{Feller:1952},
linking the fractional calculus and
the \(\alpha\)-stable distribution together\cite{Mainardi:2007, Saxena:2014}.
It is called \emph{the Feller parameterization} in this work, in which
\(\theta\) is the skewness parameter in the form of a trigonometric angle.
Among many parameterizations in the \(\alpha\)-stable distribution 
(sometimes very confusing),
this is the primary parameterization used here
(See Definition 3.7 of \cite{Nolan:2020}).
Therefore, the \(\alpha\)-stable distribution takes the notation of \(L_{\alpha}^{\theta}\).

Later in Feller's 1971 seminal book on the probability theory\cite{Feller:1971}, 
the series representation for
the PDF of the \(\alpha\)-stable distribution was given. 
When \(\theta = \pm\alpha\), the distribution becomes one-sided (or called extremal).
Its PDF \(L_\alpha(x)\) \((x > 0)\) is based on one of the most important variants 
of the Wright function:
\(L_\alpha(x) = x^{-1} \FWrfn{\alpha}{x^{-\alpha}}\)
(See Table \ref{tab:frac-dist-mapping}).

In a nutshell, \(\chi_{\alpha,k}\) is a Wright function extension of \(\chi_{k}\).
We would properly standardize 
\(\chi_{\alpha,k}\) to \(\chibar_{\alpha, k}\), then extend \(t_k\) from \eqref{eq:intro-t} to 
\begin{align}\label{eq:intro-gsas}
L_{\alpha,k}
    &\sim \mcN  / \,\chibar_{\alpha, k}
\end{align}
after which the skewness \(\theta\) is added, and we obtain the enlarged distribution family 
\(L_{\alpha,k}^\theta\).
We make sure when \(k=1\), it conforms to the \(\alpha\)-stable distribution:
\(L_{\alpha,1}^\theta\) = \(L_{\alpha}^\theta\).

We can also look at it from a different angle. We take the \(\alpha\)-stable distribution
\(L_{\alpha}^\theta\), add \(k\) from \(t_k\) to it, and form \(L_{\alpha,k}^\theta\).
We would get to the same place.

\begin{center} \qedbar \end{center}

Since 1990's, Mainardi et al. have been exploring  extensively 
the relation between the Wright function, the fractional calculus, 
and the \(\alpha\)-stable distribution\cite{Mainardi:2010, Mathai:2017, Mainardi:2020}. 
I consider this work as a continuation of such exploration.

In 1962, E. W. Stacy proposed a super distribution family called 
\emph{the generalized gamma distribution} (GG). 
Many one-sided distributions commonly used for parametric models
are part of it, such as 
the exponential distribution and the exponential power distribution, 
the half-normal, Weibull, Rayleigh, gamma, \(\chi\), and \(\chi^2\) distributions
(See Table \ref{tab:classic-map-zero}). 

The Wright function and GG have a deep connection. The large-\(x\) asymptotic of 
a PDF based on the Wright-function becomes a GG-style PDF, apart from a constant. 
We can jokingly say \emph{they are joined at the tail}.

One contribution of this work is to have GG fractionally extended,
called \emph{the generalized stable count distribution} (GSC), 
which serves as the parent of the fractional \(\chi\) distribution.
Obviously, GSC is the largest one-sided distribution family in this scope.
The following is a brief review of my past effort.


In 2017, I discussed the \(\alpha\)-stable law with Professor Mulvey at Princeton University.
Later that year, I derived \emph{the stable count distribution}
\(\scN_\alpha(x)\) as the conjugate prior of the one-sided stable
distribution \(L_\alpha(x)\)\cite{Lihn:2017, Lihn:2020, Wikipedia_StableCountDistribution, PensonGorska:2010}.
\(\scN_\alpha(x)\) measures the distribution of the length \(N\) of a L\'evy sum
\(\sum_{i = 1}^{N} X_i\) where \(X_i\)'s are i.i.d. stable random variables.
It was found to be expressed elegantly in the Wright function:
\(\scN_\alpha(x) = {\Gamma(\frac{1}{\alpha}+1)}^{-1} \FWrfn{\alpha}{x^{\alpha}}\).

In 2022, I began to organize the subordination structure of the GG-related distributions 
via the stable count distribution,
which is documented in Appendix \ref{section:product-dist}.
The GG-style pattern in the PDF of GSC, \(x^{d-1} \FWrfn{\alpha}{x^p}\),
was discovered
(See Table \ref{tab:classic-map-zero}). 

It turned out such pattern is not a coincidence.
GSC becomes GG when \(\alpha \to 0\).
GSC is proven to be the fractional extension of GG, solidifying its theoretical position 
in the statistical distribution theory.
The proof is based on the nice properties of Mainardi's auxiliary Wright functions\cite{Mainardi:2010}.

In 2023, the subordination study was extended to the known symmetric two-sided distributions, 
as shown in Appendix \ref{section:ratio-dist}.
The main structure in the PDF of the fractional \(\chi\) distribution,
\(\chi_{\alpha,k}(x) \propto x^{k-2} \FWrfn{\frac{\alpha}{2}}{x^{\alpha}}\), 
became clear, 
when the \emph{continuous Gaussian mixture} \eqref{eq:intro-gsas} is focused.
This is elaborated in Sections \ref{section:mixture} and \ref{section:fcm}.
It combines the \(\alpha\) and \(k\) parameters into a single distribution.
This is the \textbf{cornerstone} that generalizes the 
\(\alpha\)-stable, Student's t, and exponential power distributions.
All three are combined to the expanded notation \(L_{\alpha,k}^\theta\).
These are the main contributions of this paper.

The symmetric case \((\theta = 0)\) is very solid and elegant: \(L_{\alpha,k} = L_{\alpha,k}^0\). 
Most of this paper aims to explain this case in great detail.
The skewness generalization is experimental. It is a much harder task that requires future work.
In a tongue-twisting manner, it goes like:
\begin{itemize}
  \item \(L_{\alpha,1}^\theta\): the original \(\alpha\)-stable distribution,
  \item \(L_{1,k} (k > 0)\): Student's t distribution,
  \item \(L_{\alpha,-1}\): the exponential power distribution.
\end{itemize}

I also proposed a framework for the generation of the random variables,
from the GSC to \(\chimeanDist\), to \(L_{\alpha,k}\). 
It is based on Feller's square-root process in 1951\cite{Feller:1951}.
Hence, this new distributional system is grounded on an elegant, well known stochastic process.
Amazingly, some of the solutions of the mean-reverting force
turn out to be simple polynomials.

Lastly, the univariate Gaussian mixture is extended nicely to the multivariate distribution.
Two kinds are proposed: The first kind is a straightforward copy from the elliptical distribution.
The second kind is an invention that allows each dimension to have its own shape.
The bivariate models from both kinds are applied to the VIX/SPX\footnote{
VIX is the CBOE volatility index; SPX is the S\&P 500 index.
} data set with success.
This is discussed in Section \ref{section:multivariate}.

Figure \ref{fig:gsc_hierarchy} illustrates the stack of distributions involved in this work. 
Seven new distributions are invented, and a handful of existing, classic distributions are subsumed into them.

\begin{figure}[htp]
    \centering
    \includegraphics[width=6in]{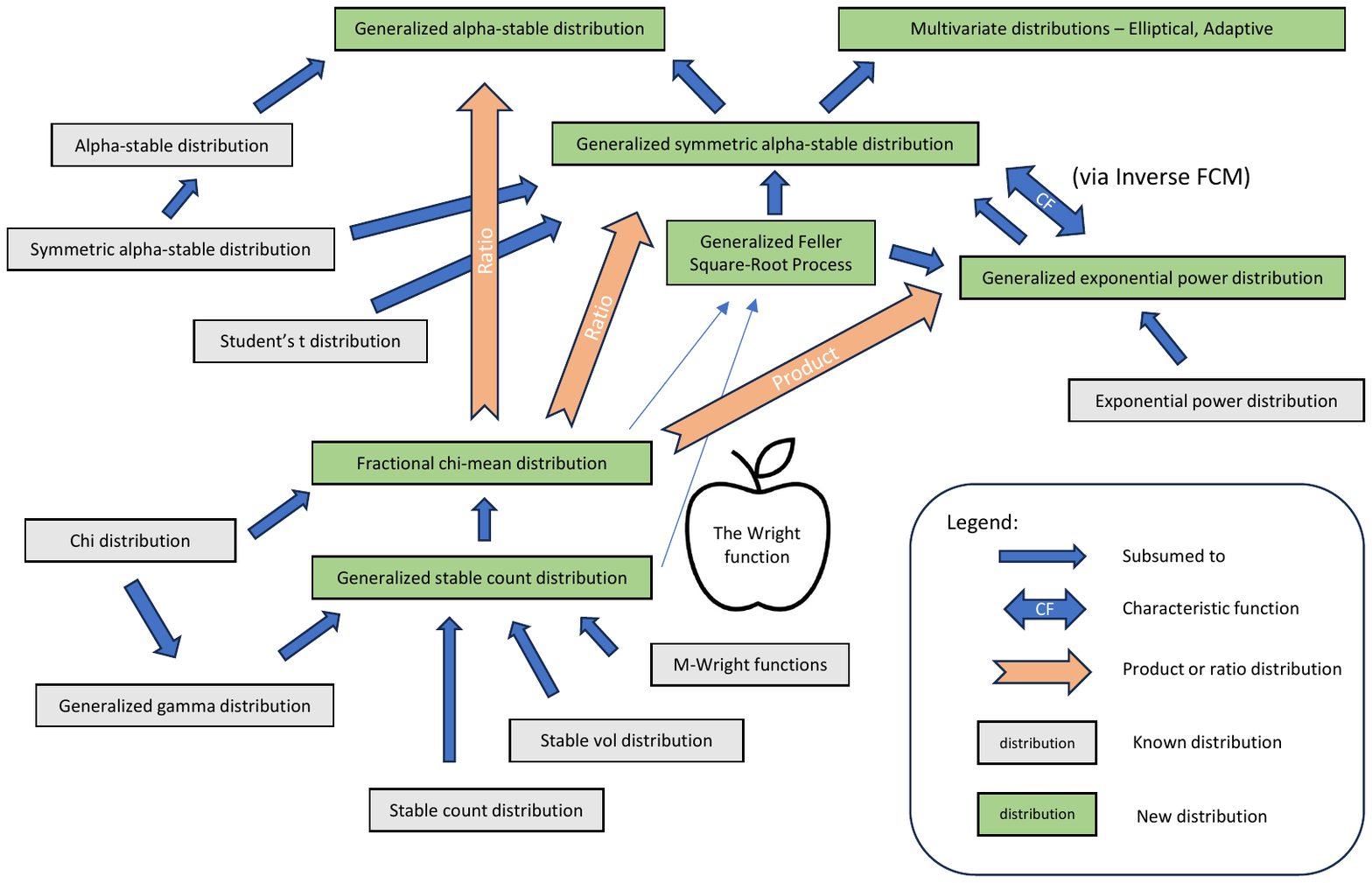}
    \caption{The hierarchy of the distributions involved in this work. 
        The Wright function plays a central role. 
        The green boxes are the contributions from this work.
    }
    \label{fig:gsc_hierarchy}
\end{figure}

\begin{center} \qedbar \end{center}

The benefits from multiple degrees of freedom, aka \(k > 1\), are two-fold.
First, the original \(\alpha\)-stable distribution
is viewed as having one degree of freedom.
That explains why it lacks most of the moments.
\(k > 1\) allows the distribution to 
have finite moments with increasing \(k\):
The \(n\)-th moment should exist when \(k > n\).
This addresses, to some extent, the decades-old issue of non-existing moments in the stable law. 

Secondly, for the generalized exponential power distribution,
it allows the slope of the PDF to be continuous everywhere.
This deserves some explanations. 

In 2017, I incorporated the exponential power distribution,
\(\mathcal{E}_\alpha(x) = e^{-{|x|}^\alpha}/2\GammaAlpha\),
into the Hidden Markov Model (HMM) for regime identification\cite{Lihn:2017hmm}.
It is especially important for the crash regime, such that
the model can capture the leptokurtic nature of the financial time series.

A main deficiency of its high-kurtosis distribution is that, 
when \(\alpha \le 1\), the slope of the PDF is discontinuous at \(x \to 0^{\pm}\).
A simple example is \(e^{-|x|}\).
This appears to violate the intuition that
the low-volatility returns should distribute like normal.
This deficiency is also attributed as having one degree of freedom.
It is addressed by the addition of the \(k\) dimension:
\(\mathcal{E}_\alpha \to \mathcal{E}_{\alpha,k}\). 
When \(k > 1\), the slope of the PDF is continuous everywhere.

\begin{center} \qedbar \end{center}


To showcase the new GAS distribution, it is used to fit the daily log returns 
of SPX from 1927 to 2020 in Figure \ref{fig:sp500}. 
This is a very difficult data set to fit for any existing two-sided distribution. 
In the \emph{pyro project}\cite{Bingham:2018pyro},
it is used as a failed example
that this data set can \textbf{not} be fitted by a single marginal distribution\footnote{
See \url{https://pyro.ai/examples/stable.html}}.

Nonetheless, GAS produced a good fit. It fits the shape of the histogram 
reasonably good in both the density and log-density scales, to the \(10^{-3}\) level,
which is about 6 standard deviations.
It produced the desired excess kurtosis (\(\sim\)19) and standardized peak density (\(\sim\)0.72), 
and added a small skewness. Most interestingly, it provides a set of interpretable parameters
that we are familiar with in the Student's t and \(\alpha\)-stable context: 
\(\alpha=0.813, k=3.292\) for the shape, and \(\theta=0.08\) for the skewness.
More discussion is in Section \ref{section:sp500-fit-explained}.

\begin{figure}[htp]
    \centering
    \includegraphics[width=6in]{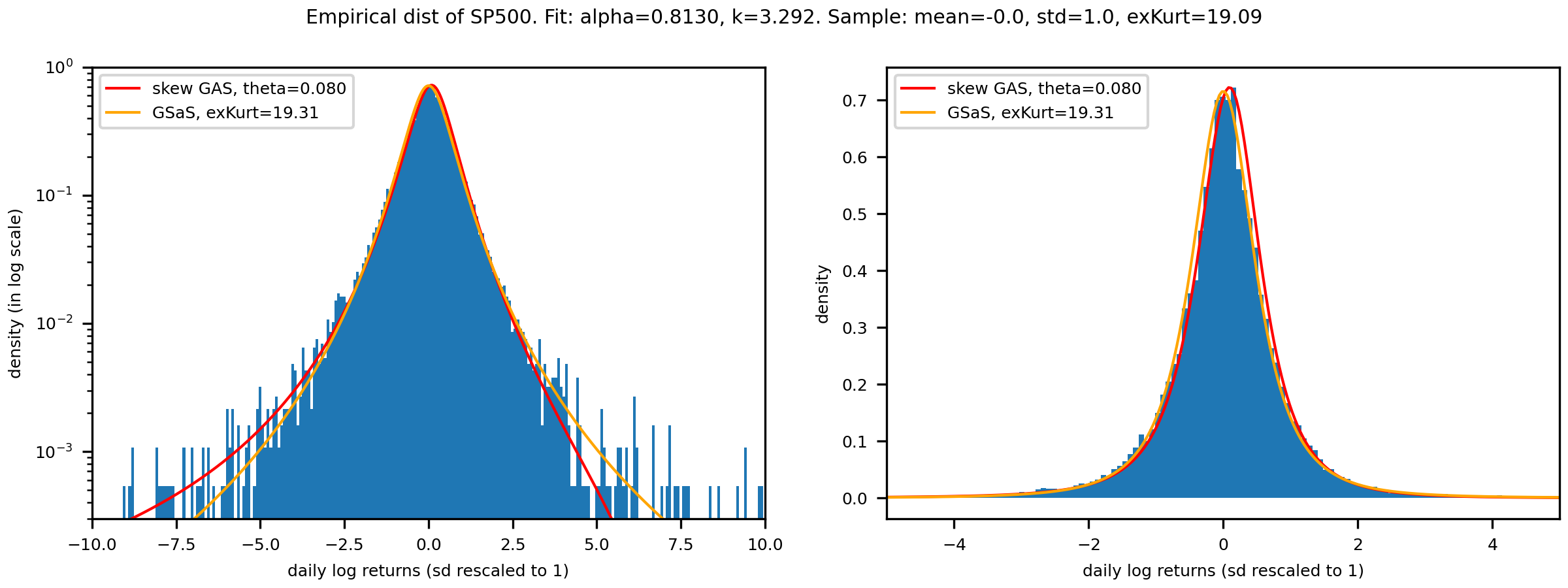}
    \caption{Empirical distribution of S\&P 500 daily log returns. 
    Data is standardized and centered for demonstration purpose. 
    The return distribution is normalized to one standard deviation with zero mean.
    Fit by GSaS (orange line) and GAS (red line).
    The fits are reasonably good up to 6 standard deviations.}
    \label{fig:sp500}
\end{figure}

%% file: section-mixture.tex
\section{Framework of Continuous Gaussian Mixture}
\label{section:mixture}

The bulk of this paper is focused on the construction of a symmetric two-sided distribution
in the form of a continuous Gaussian mixture. Both the ratio and product distribution methods are used.
In the case of the symmetric \(\alpha\)-stable distribution (SaS)\cite{Crisanto-Neto:2018}, 
the exponential power distribution comes from its characteristic function (CF).
We would like to present a unified framework and familiarize the reader with the notations,
which would be otherwise subtle and confusing.

Assume the PDF of a two-sided symmetric distribution is \(L(x)\) where \(x \in \mathbb{R}\).
It has zero mean, \(\E(X|L) = 0\). 
Assume the PDF of a one-sided distribution is \(\chi(x) \,(x > 0)\) such that 

\begin{align}\label{eq:mix-pdf-by-chi-ratio}
L(x)
    &:= \int_{0}^{\infty} s \, ds \,
      \mcN(xs) \, \chi(s)
\end{align}
This is nothing but the definition of a ratio distribution with a standard normal variable \(\mcN\).
This is the first form of the Gaussian mixture: \(L \sim \mcN / \chi\).

It also has the equivalent expression in terms of a product distribution
by way of \emph{the inverse distribution} \( \chi^\dagger \) such
that \(L \sim \mcN \chi^\dagger\). 
This is the second form of the Gaussian mixture.

\(\chi^\dagger\) is closer to our typical understanding of 
the marginal distribution of a volatility process. 
For example, the Brownian motion process \(dX_t = \sigma_t \, dW_t\),
measured in a particular time interval \(\Delta t\), 
we have \(\Delta X_t \sim L\) and \(\sigma_t \sim \chi^\dagger\).

However, \(\chi\) in the first form is more natural in the expression of 
the \(\alpha\)-stable distribution. 
So we are more inclined to use the ratio distribution.
The reader should keep this subtlety in mind.

\begin{lemma}\label{lemma-mix--ratio-vs-prod} (Inverse distribution)
The inverse distribution is defined as\cite{Wikipedia_InverseDistribution}
\begin{align} \label{eq:mix-inverse-chi}
     \chi^\dagger(s) &:= s^{-2} \, \chi\left(\frac{1}{s}\right)
\end{align}
such that
\begin{align}
\int_{0}^{\infty} s \, ds \, \mcN(xs) \, \chi(s)
    &= \int_{0}^{\infty} \, \frac{ds}{s} \,
      \mcN\left(\frac{x}{s}\right) \, \chi^\dagger(s)
    && (x \in \mathbb{R})
\\
\int_{0}^{\infty} s \, ds \, \mcN(xs) \, \chi^\dagger(s)
    &= \int_{0}^{\infty} \, \frac{ds}{s} \,
      \mcN\left(\frac{x}{s}\right) \, \chi(s)
\end{align}
The proof is straightforward by a change of variable \(t = 1/s\). 
You can move between LHS and RHS easily.

\qedbar
\end{lemma}

We use the notation \(\text{CF}\{g\}(t) = \E(e^{itX}|g)\) to represent 
the characteristic function transform of the PDF \(g(x)\).
Note that \(\mcN\) has a special property that 
its CF is still itself: \(\text{CF}\{\mcN\}(t) = \sqrt{2\pi} \, \mcN(t)\).

\begin{lemma} \label{lemma-mix-cf-dist} (Characteristic function transform of \(L\))
Let \(\phi(t)\) be the CF of \(L\) such that 
\( \phi(t) := \text{CF}\{L\}(t) =
    \int_{-\infty}^{\infty} dx  \, \exp(itx) \, L(x)\). 
\eqref{eq:mix-pdf-by-chi-ratio} is transformed to
\begin{align} \label{mix-cf}
\phi(t)
    &= \sqrt{2\pi} \, \int_{0}^{\infty} \, ds \,
      \mcN\left(\frac{t}{s}\right) \, \chi(s)
    && (t \in \mathbb{R})
\end{align}

This allows us to define a new distribution pair: \(L_{\phi}\) and \(\chi_{\phi}^\dagger\), 
in terms of a product distribution such that
\begin{align} \label{eq:mix-pdf-by-chi-dagger}
L_{\phi}(x)
    &:= \int_{0}^{\infty} \, \frac{ds}{s} \,
      \mcN\left(\frac{x}{s}\right) \, \chi_{\phi}^\dagger(s)
    && (x \in \mathbb{R})
\\ \label{eq:mix-chi-dagger}
\chi_{\phi}^\dagger(s) &:=
    \frac{ s \, \chi(s) }{\E(X|\chi)}
\end{align}
where \(\E(X|\chi)\) is the first moment of \(\chi\).
Here \(\chi_{\phi}^\dagger\) is the inverse distribution of \(\chi_{\phi}\),
which can be reverse-engineered according to \eqref{eq:mix-inverse-chi},
\begin{align} \label{eq:mix-chi-phi}
\chi_{\phi}(s) &:=
    \frac{ s^{-3} }{\E(X|\chi)} \,\chi\left(\frac{1}{s}\right) 
\end{align}

\qedbar
\end{lemma}

We are in an interesting place: We start with an one-sided distribution \(\chi\),
we derive two variants from it: \(\chi_{\phi}\) and \(\chi_{\phi}^\dagger\).
We also obtain two two-sided distributions: \(L\) and \(L_{\phi}\).

We shall call \(\chi_{\phi}\) \emph{the characteristic distribution} of \(\chi\) 
since it facilitates the following parallel relation:
\begin{align*}
L &\sim \mcN / \chi
\\
L_{\phi} &\sim \mcN / \chi_{\phi}
\end{align*}

\(\chi\) symbolizes the fractional \(\chi\) distribution we are about to introduce.
The \(\phi\) suffix will be replaced with the \emph{negation} (sign change) of the degree of freedom.


%% file: sub-symbols.tex
\subsection{List of Symbols}
\label{section:symbols}
The next tables describe the notations and symbols that will be used in this paper.
Note that GSC, FCM, GAS, GSaS are new notations. Others are known.

\vspace{6pt}

\textbf{The PDF of Simple Distribution}

\begingroup
\setlength{\tabcolsep}{10pt} 
\renewcommand{\arraystretch}{1.25} 
\begin{tabular}{ l l }
\(\Gamma(x;s)\)         & The gamma distribution (used only as a notation) \\
\({\chi}_{k}(x), {\chi}^2_{k}(x)\)  & The chi and chi-square distributions \\
\({L}(x)\)              & The two-sided Laplace distribution \\
\({\mcN}(x), \mcN(x;\mu, \sigma^2)\) & The normal distribution, 
    \({\mcN}(x) \sim {\mcN}(x; 0,1)\) \\
\(\mathcal{E}_\alpha(x)\) & The two-sided exponential power distribution
    (\(e^{-{|x|}^\alpha}/2\GammaAlpha\))\\
\(\scN_\alpha(\nu)\)    & The stable count distribution (SC) \\
\(V_{\alpha}(x)\)       & The stable vol distribution (SV) \\
\({L}_{\alpha}(x)\)     & The one-sided stable distribution, aka \(L_\alpha^\alpha(x)\) \\
\(\sas{x}\)             & The symmetric \(\alpha\)-stable distribution, aka \(L_\alpha^0(x)\) (SaS) \\
\({t}_{k}(x)\)          & Student's t distribution \\
\({\text{Wb}}(x;k)\)    & Weibull distribution \\
\({\text{IG}}(x;k)\)    & The inverse gamma distribution (See 
    \href{https://docs.scipy.org/doc/scipy/reference/generated/scipy.stats.invgamma.html}{\texttt{scipy.stats}}, but use \(a=k\)) \\
\({\text{IWb}}(x;k)\)   & The inverse Weibull distribution (See 
    \href{https://docs.scipy.org/doc/scipy/reference/generated/scipy.stats.invweibull.html}{\texttt{scipy.stats}}, but use \(c=k\)) \\\end{tabular}
\endgroup
\vspace{6pt}

\textbf{The PDF of Complex Distribution}

\begingroup
\setlength{\tabcolsep}{10pt} 
\renewcommand{\arraystretch}{1.25} 
\begin{tabular}{ l l }
\(\pdfGG(x; a, d, p)\)       & The generalized gamma distribution (GG) \\
\(\text{GenGamma}(x; s,c)\)  & The generalized gamma distribution too \((c=p, sc=d)\) \\
\(\afstable{x}\)  & The \(\alpha\)-stable distribution \\
\(\gscNx\)        & The generalized stable count distribution (GSC, new) \\
\(\gas{x}\)       & The generalized \(\alpha\)-stable distribution (GAS, new) \\
\(\gsas{x}\)      & The generalized symmetric \(\alpha\)-stable distribution (GSaS, new) \\
\(\chi_{\alpha,k}(x)\) & The fractional \(\chi\) distribution (new) \\
\(\chimean{x}\)        & The fractional \(\chi\)-mean distribution (FCM, new) \\
\(\mathcal{E}_{\alpha,k}(x)\) & The generalized exponential power distribution (GExpPow, new) \\
\end{tabular}
\endgroup
\vspace{6pt}

\textbf{Other Symbols}

\begingroup
\setlength{\tabcolsep}{10pt} 
\renewcommand{\arraystretch}{1.25} 
\begin{tabular}{ l l }
\(\alpha\)  & The shape parameter associated with the Le\'vy stability index \\
\(C\)       & The normalization constant for \(\gscNx\) \\
\(d\)       & The degree of freedom parameter in \(\gscNx\) \\
\(\theta\)  & The skewness parameter in \(\gas{x}\) and \(g_{\alpha}^{\theta}(x,s)\) \\
\(e^{-z^\alpha}\)        & The one-sided stretched exponential function \\
\(E_{\alpha}(z)\)        & The Mittag-Leffler function \\
\(k\)                    & The degree of freedom parameter in \(\gas{x}, {t}_{k}(x), \chimean{x}\) \\
\(\E(X^n|\text{Dist})\)   & The \(n\)-th moment of the distribution "Dist" \\
\(m_n\)                  & The \(n\)-th moment in a local context \\
\(p\)                    & The shape parameter in \(\gscNx\), unless mentioned otherwise \\
\(\Gamma(z)\)            & The gamma function \\
\(\Gamma(s,z)\)          & The upper incomplete gamma function \\
\(\gamma(s,z)\)          & The lower incomplete gamma function \\
\(g_{\alpha}^{\theta}(x,s)\)   & The skew-Gaussian kernel \\
\(W_{\lambda,\delta}(z)\)      & The Wright function (two-parameter)\\
\(    W \left[ 
        \begin{matrix}
        a,& b
        \\
        \lambda,& \mu
        \end{matrix}
    \right](z)
\)                             & The four-parameter Wright function \\
\end{tabular}
\endgroup
\vspace{6pt}

%% file: section-main-result.tex
\section{Main Results}
\label{section:main-result}

We will go over the main results in this section, and leave the explanations and proofs 
for \iftoggle{fullpaper}{later sections.}{later.} 
The new symmetric \(\alpha\)-stable distribution \(\gsasDist\) is constructed in three steps.
Then its CF is used to extend the exponential power distribution.
The last part is to add the skewness \(\theta\) experimentally 
to the new \(\alpha\)-stable distribution \(\gasDist\).

\subsection{Recap of the Wright Function}
\label{section:main-recap-wright}

The mathematical foundation of this work is the Wright function\cite{Wright:1933, Wright:1935}
as shown in the center of Figure \ref{fig:gsc_hierarchy}. 
Its series representation is

\begin{align}\label{eq:wright-fn}
W_{\lambda,\delta}(z) &:= \sum_{n=0}^\infty
    \frac{z^n}{n!\,\Gamma(\lambda n+\delta)} 
    && (\lambda \ge -1, z \in \mathbb{C})
\end{align}

Four \((\lambda,\delta)\) pairs are used extensively in this work.
The first group of the two are
\(F_\alpha(z) := W_{-\alpha,0}(-z)\) and 
\(M_\alpha(z) := W_{-\alpha,1-\alpha}(-z) = F_\alpha(z) /(\alpha z)\), 
where \(\alpha \in [0,1]\).
They are used in various definitions and proofs,
such as in \eqref{eq:main-gsc-pdf}, \eqref{eq:main-chimean-pdf}, \eqref{eq:m-wright-M}.

In particular, \(M_\alpha(z)\) is called \emph{the M-Wright function}
or simply \emph{the Mainardi function} 
(See Appendix \ref{section:formula-m-wright})\cite{Mainardi:2010, Mathai:2017, Mainardi:2020}. 
Conceptually, the \emph{fractional extension} of a classic, exponential-based distribution 
hinges on two important properties: 
\(M_0(z) = \exp(-z)\) and \(M_{\frac{1}{2}}(z) = \frac{1}{\sqrt{\pi}} \exp(-z^2/4)\),
that lead to \eqref{eq:main-gg-by-gsc}.

The second group of the two are \(W_{-\alpha,-1}(-z)\) and
\(-W_{-\alpha,1-2\alpha}(-z)\). 
Their usefulness is discovered by me 
for the generation of random variables,
such as in  \eqref{eq:mu-sol-adv} and \eqref{eq:m-wright-series-diff}.
They are associated with the derivatives of \(F_\alpha(z)\) and \(M_\alpha(z)\).
In some cases, they lead to beautiful polynomial solutions.
Section \ref{section-gsc-random} is dedicated to this subject.

Let's begin with the 2017 discovery of the stable count (SC) distribution \(\scN_\alpha(x)\)\cite{Lihn:2017}.
Its PDF was first formulated as the conjugate prior of the one-sided stable distribution \(L_\alpha(x)\).
Soon afterwards, it was linked to the Wright function:
\begin{align} \label{eq:main-sc-pdf}
\scN_\alpha(x) &:=
    \frac{1}{\Gamma(\frac{1}{\alpha}+1)} 
    \frac{1}{x} L_\alpha\left(\frac{1}{x}\right)
=
    \frac{1}{\Gamma(\frac{1}{\alpha}+1)} 
    W_{-\alpha,0}(-x^\alpha)
    & (0 < \alpha \leq 1)
\end{align}

In 2020, a variant of SC was proposed, called the stable vol (SV) distribution \(V_\alpha(x)\)\cite{Lihn:2020},
that works better with a normal variable:
\begin{align} \label{eq:main-sv-pdf}
V_{\alpha}(x) &:= 
    \frac{\sqrt{2 \pi} \,\Gamma(\frac{2}{\alpha}+1)}{\Gamma(\frac{1}{\alpha}+1)} \,
    \scN_{\frac{\alpha}{2}}(2 x^2)
    = 
        \frac{ \sqrt{2\pi} }{ \Gamma(\frac{1}{\alpha}+1) }
        \, W_{-\frac{\alpha}{2},0} \left( -{(\sqrt{2} x)}^\alpha \right)
    & (0 < \alpha \leq 2)
\end{align}

As the readers can observe, \(F_\alpha(z)\), \(M_\alpha(z)\), 
\(\scN_\alpha(x)\), \(V_{\alpha}(x)\), \(L_{\alpha}(x)\)
all point to a GG-style pattern that subsumes a large number of one-sided distributions.
The pattern is as follows.


\subsection{New One-Sided Distributions}
\label{section:one-sided-definitions}

\begin{definition}[Generalized stable count distribution (GSC)]\label{def:gsc}

GSC is a four-parameter one-sided distribution family, whose PDF is defined as

\begin{align}
\label{eq:main-gsc-pdf}
\gscNx &:= C \, {\left( \frac{x}{\sigma} \right)}^{d-1}
         \, \FWrfn{\alpha}{ {\left( \frac{x}{\sigma} \right)}^{p} } 
        && (x \ge 0)
\end{align}
where \(\alpha \in [0,1]\) controls the shape of the Wright function;
\(\sigma\) is the scale parameter;
\(p\) is also the shape parameter controlling the tail behavior \((p\ne 0, dp \ge 0)\);
\(d\) is the \emph{degree of freedom} parameter. 
When \(\alpha \to 1\), the PDF becomes a Dirac delta function: \(\delta(x-\sigma)\)
assuming \(\sigma\) is finite.
When \(d \ge 1\), all the moments of GSC exist and have closed forms.
\iftoggle{fullpaper}{See Section \ref{section:gsc} 
for more detail.}{} 

\iftoggle{fullpaper}{
The normalization constant \(C\) will be derived in Section \ref{section:gsc-const} 
and is shown below:
}{
The normalization constant \(C\) is:
}

\begin{align}\label{eq:main-gsc-const}
C &= \left\{
    \begin{array}{ll}
        \frac{|p|}{\sigma} 
        \frac{\Gamma(\frac{\alpha d}{p})}{\Gamma(\frac{d}{p})} 
            &, \,\, \text{for} \,\, \alpha \ne 0, d \ne 0.
    \\ 
        \frac{|p|}{\sigma \alpha}  &, \, \text{for} \,\, \alpha \ne 0, d = 0.
    \end{array}
    \right.
\end{align}

\qedbar
\end{definition}
The original stable count distribution \eqref{eq:main-sc-pdf}
is simply \(\scN_\alpha(x) = \scN_\alpha(x; \sigma=1,d=1,p=\alpha)\).
\iftoggle{fullpaper}{See Section \ref{section:sc-sv}.}{} 

It is important to note that \(d\) and \(p\) are allowed to be negative, 
as long as \(dp \ge 0\). 
Several examples are found in Table \ref{tab:classic-map-zero} and Table \ref{tab:frac-dist-mapping}. This feature is used to define the characteristic distribution of FCM below.

Since the Wright function extends an exponential function to the fractional space,
GSC is the fractional extension
of the generalized gamma (GG) distribution\cite{Stacy:1962}, 
whose PDF is defined as:
\begin{equation}
\label{eq:main-pdf-gg}
\pdfGG(x; a, d, p) =
    \frac{|p|}{a \Gamma(\frac{d}{p})} {\left(\frac{x}{a}\right)} ^{d-1} e^{-(x/a)^p}.
\end{equation}
The parallel use of parameters is obvious,
except that \(a\) in GG is replaced with \(\sigma\) in GSC to avoid confusion with \(\alpha\).

GG is subsumed to GSC in two ways:

\begin{align}\label{eq:main-gg-by-gsc}
\pdfGG(x; \sigma, d, p)
    &:= \left\{
    \begin{array}{ll}
        \scN_0(x; \sigma, d=d-p, p) 
        &, \,\, \text{at} \,\, \alpha = 0.
    \\
        \scN_{\frac{1}{2}}\left(x; \sigma=\frac{\sigma}{2^{2/p}}, d=d-\frac{p}{2}, p=\frac{p}{2} \right)
        &, \,\, \text{at} \,\, \alpha = \frac{1}{2}.
    \end{array}
    \right.
\end{align}
The first line is treated as the definition of GSC at \(\alpha = 0\).
\iftoggle{fullpaper}{See Sections \ref{section:mapping-gg-0} and \ref{section:mapping-gg-one-half} 
for more detail.}{} 

Although the first line is more obvious,
it is the second line that leads to the fractional extension of the \(\chi\) distribution.

The readers are reminded that 
\(\FWrfn{\alpha}{x}\) used here is fully supported 
by existing software packages. This is due to the fact that
it can be converted to \(L_\alpha(x)\) 
by \(L_\alpha(x) = x^{-1} \FWrfn{\alpha}{x^{-\alpha}}\) from \eqref{eq:main-sc-pdf}. 
And \(L_\alpha(x)\) can be computed via 
\texttt{scipy.stats.levy\_stable} using 1-Parameterization with 
\texttt{beta}=1, \texttt{scale}= \(\cos(\alpha \pi/2)^{1/\alpha}\)
for \(0 < \alpha < 1\). 
(\(L_1(x) = \delta(x-1)\) can't be computed.)\footnote{
See Chapter 1 of \cite{Nolan:2020} for more detail on different parameterizations.
We wouldn't go into the issue of stable parameterizations.}
It might seem somewhat peculiar that we can use the existing implementation of 
\(L_\alpha(x)\) to develop all the new distributions
for proof of concept.

\begin{center}
\qedbar
\end{center}

When working with the stable law, \(\alpha\) in GSC may become \(\alpha/2\) , which produces
the more recognizable range of \(\alpha \in [0,2]\) as in the next definition.

Let \(\chi_{k}\) be the \(\chi\) distribution of \(k\) \emph{degrees of freedom},
then Student's t distribution \(t_k\) can be constructed by
\(t_k \sim \mcN / (\chi_{k} / \sqrt{k})\)\cite{JoramSoch:2019t}.
We prefer to standardize it by dividing out the \(\sqrt{k}\) term. So we use a new convention 
of a bar over \(\chi\) such that
\(\chibar_k := \chi_{k} / \sqrt{k}\).

\begin{definition}[Fractional \(\chi\)-mean distribution (FCM)]\label{def:fcm}

FCM \(\chimeanDist\) is the fractional extension of \(\chibar_k\).
It is the \textbf{main workhorse} in this paper.
It is a two-parameter sub-family of GSC, whose PDF for \(k > 0\) is: 

\begin{align}
\label{eq:main-chimean-pdf}
\chimean{x}
    &:= \scN_{\frac{\alpha}{2}}(x; \sigma=\fcmSigma, d=k-1, p=\alpha)
    && (x \ge 0, k > 0)
\\ \notag
    &\quad \quad \where \fcmSigma = 
        \frac{{|k|}^{1/2 - 1/\alpha}}{\sqrt{2}}  
\\ \notag
    &= \frac{\alpha \Gamma(\frac{k-1}{2})}{\Gamma(\frac{k-1}{\alpha})}
    \, {\left( \frac{1}{\fcmSigma} \right)}^{k-1} 
    {x}^{k-2} \,
    W_{-\frac{\alpha}{2}, 0} \left( -{
        \left( \frac{x}{\fcmSigma} \right)
    }^\alpha \right)
\end{align}
and \(\alpha \in [0,2]\) matches exactly L\'evy's stability index
in the stable law. See Lemma \ref{lemma-frac-chi1} when \(k = 1\).

The characteristic FCM (\(\chi_\phi\) in Section \ref{section:mixture}) is defined 
in the negative \(k\) space, whose PDF is:

\begin{align}
\label{eq:main-chimean-pdf-neg-k}
\chimeanDistNegK(x)
    &:= \scN_{\frac{\alpha}{2}}(x; \sigma=(\fcmSigma)^{-1}, d=-k, p=-\alpha)
    && (x \ge 0, k > 0)
\end{align}
This is used to subsume the exponential power distribution to GSaS below.
Negative \(k\) is not in the scope until \eqref{eq:main-gexppow-pdf-neg-k}.

To use a single expression for both positive and negative \(k\)'s,
let \(h(k) = (1+\sgn(k))/2\) be the Heaviside (step) function.
The PDF can be consolidated to:
\begin{align}
\label{eq:main-chimean-pdf-union}
\chimeanDist(x)
    &= \scN_{\frac{\alpha}{2}}(x; \sigma={(\fcmSigma)}^{\sgn(k)}, 
             d = k - h(k), p = \sgn(k) \,\alpha)
    && (x \ge 0, k \in \Rn0)
\end{align}
This looks cumbersome, but it can save some duplication in a few places.

\qedbar
\end{definition}

\begin{lemma}\label{lemma-fcm-reflection} (FCM Reflection Formula)
Assume \(k > 0\), FCM has the reflection formula that resembles \eqref{eq:mix-chi-phi}:
\begin{align}
\label{eq:main-fcm-reflection}
\chibar_{\alpha,-k}(x)
    &= \frac{1 }{x^3 \, \E(X|\chimeanDist) } 
       \, \chimeanDist\left(\frac{1}{x}\right)
\\ \notag
    & \where \E(X|\chimeanDist) = 
        \fcmSigma \,
        \frac{\Gamma(\frac{k-1}{2})} {\Gamma(\frac{k-1}{\alpha})} 
        \frac{\Gamma(\frac{k}{\alpha})} {\Gamma(\frac{k}{2})} 
    \,\, \text{(See \eqref{eq:fcm-moment})}
\end{align}
It is equivalent to the following relation in terms of the moments:
\begin{align}\label{eq:main-fcm-moment-reflection}
\E(X^n|\chibar_{\alpha,-k}) 
    &= \frac{\E(X^{-n+1}|\chimeanDist)} {\E(X|\chimeanDist) },
    \,\, k > 0.
\end{align}
In particular, when \(n=1\), \(\E(X|\chibar_{\alpha,-k}) = 1/\E(X|\chimeanDist)\).
The first moment is reciprocal.

\qedbar
\end{lemma}

When \(k\) is positive, 
\(k\)'s meaning is exactly the same as
that of Student's t, since \(\chibar_{1,k} = \chi_{k} / \sqrt{k}\),
which can be proven easily from \eqref{eq:main-gg-by-gsc}.
On the other hand,  
\(\chibar_{\alpha,1}\) is called \emph{fractional chi-1}.
It is used to construct the SaS distribution:
\(P_{\alpha} := L_\alpha^0 \sim \mcN / \chibar_{\alpha,1}\). 
This important point will be elaborated in Lemma \ref{lemma-frac-chi1}.


The parameter space outside \((\alpha = 1\,|\, k = 1)\) is our innovation -
a main contribution of this paper.
\iftoggle{fullpaper}{See Section \ref{section:fcm} for more detail.}{} 

\begin{lemma}\label{lemma-fcm-delta1} (FCM Asymptotics)
As \(k \to \pm\infty\), FCM becomes a delta function at its asymptotic mean. 
For \(k \gg 1\),
\begin{align}
\label{eq:main-fcm-delta1}
\chibar_{\alpha,k}(x)
    &\sim \delta(x-m_{\alpha}), \,\, \where 
    m_{\alpha} = {\alpha}^{-1/\alpha}.
\end{align}
On the other hand, \(\chibar_{\alpha,-k}(x) \sim \delta(x-1/m_{\alpha})\).

This property ensures that the two-sided distributions constructed from FCM 
become a normal distribution at \(k \to \pm\infty\). Hence, this lemma is associated with
the Central Limit Theorem.

Note that \(m_1 = 1\) is preserved for Student's t; \(m_2 = 1/\sqrt{2}\) is preserved for \(\alpha\)-stable.

The speed of approaching a delta function is determined by how fast the variance decreases: 
\(\E(X^2|\chimeanDist) - \E(X|\chimeanDist)^2 \to 0\).
It is faster for larger \(\alpha\), and can be very slow for small \(\alpha\) below 1.

\qedbar
\end{lemma}

\begin{definition}\label{def:inverse-fcm} (Inverse FCM)
The inverse distribution of FCM can be constructed by 
"inverting \(\sigma\) and negating \(d\) and \(p\)" so to speak.
The results are simple due to the FCM reflection rule:
\begin{align} \label{eq:inverse-fcm-pdf}
\chimeanDist^\dagger(x) &:=
    \frac{x \, \chibar_{\alpha, -k}(x)} {\E(X|\chibar_{\alpha,-k})}  
    = \scN_{\frac{\alpha}{2}}(x; \sigma=(\fcmSigma)^{-1}, d=-(k-1), p=-\alpha)
    && (k > 0)
\\ \label{eq:inverse-fcm-pdf-neg-k}
\chibar_{\alpha, -k}^\dagger(x) &:=
    \frac{x \, \chimeanDist(x)} {\E(X|\chibar_{\alpha,k})}  
    = \scN_{\frac{\alpha}{2}}(x; \sigma=\fcmSigma, d=k, p=\alpha)
\end{align}

Note that \eqref{eq:inverse-fcm-pdf-neg-k} is similar to \eqref{eq:main-chimean-pdf}
except \(d\) is incremented by one.
And \eqref{eq:inverse-fcm-pdf} is similar to \eqref{eq:main-chimean-pdf-neg-k}
except \(d\) is incremented by one.

\qedbar
\end{definition}

We are jumping ahead to emphasize that the inverse FCM allows 
the generalized exponential power distribution \(\gepDist\) in \eqref{eq:main-gexppow-pdf} 
defined as \(\mcN / \,\chibar_{\alpha, -k}^\dagger\) to be mapped to the negative \(k\) space of GSaS
as \(\mcN / \,\chibar_{\alpha, -k}^\dagger = \mcN \, \chibar_{\alpha, -k} = L_{\alpha,-k}\).

As a validation, \(\chibar_{\alpha, -1}^\dagger(x)\) is equal to the stable vol distribution
\(V_\alpha(x)\) in \eqref{eq:main-sv-pdf}, as expected by its original design
(explained in Section \ref{section:sc-sv}).


In Figure \ref{fig:pyro_sp500_fcm_dist}, we plot the FCM distribution
that came out of the S\&P 500's GSaS fit, which is a distribution of
the inverse of volatility.

\begin{figure}[htp]
    \centering
    \includegraphics[width=4in]{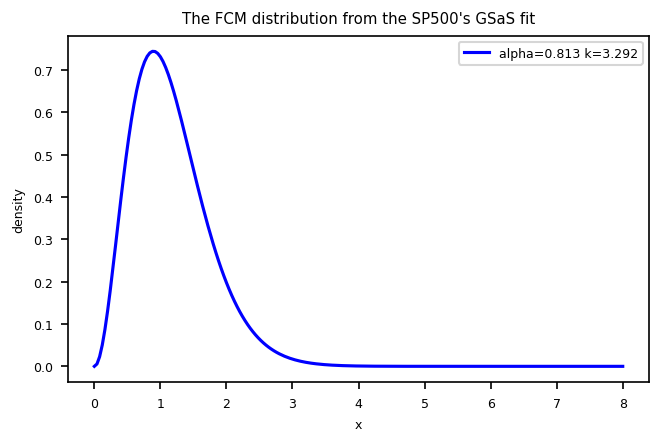}
    \caption{The FCM distribution from the S\&P 500's GSaS fit.
    Its scale has been standardized because the variance of the log-return distribution was normalized to one.
    This is a distribution of the inverse of volatility. 
    The smaller $x$ the more volatile the market was.}
    \label{fig:pyro_sp500_fcm_dist}
\end{figure}

All aspects of FCM are well defined up to this point. 
Next, we are going to take apart the \(\alpha\)-stable distribution.

\begin{definition}[The skew-Gaussian kernel]\label{def:skew-kernel}
The skew-Gaussian kernel \(\gskew\) is derived from the \(\alpha\)-stable law.

The minus-log of the stable CF is defined 
in Feller's parameterization\cite{Feller:1971} as 
\(\psi_\alpha^\theta(\zeta) := \exp\left( \text{sgn}(\zeta) \frac{i\theta\pi}{2} \right) {|\zeta|}^\alpha\),
where \(\theta\) is the skewness parameter confined by the so-called 
\emph{Feller-Takayasu diamond}: 
\(\theta \le \min\{\alpha, 2-\alpha\}\) and \(\alpha \in (0,2]\)\cite{Mainardi:2007, Mainardi:2020}.

The stable  PDF is \(\afstable{x} := 
    \frac{1}{2\pi} \int_{-\infty}^{\infty} d\zeta  
    \, \exp(-\psi_\alpha^\theta(\zeta)) \, e^{-ix\zeta}\).
We want it to be decomposed according to
\(\afstable{x} = \int_{0}^{\infty} s \, ds \, \gskew \, \chibar_{\alpha,1}(s)\). 
The solution of \(\gskew\) is

\begin{align}\label{eq:main-skew-gaussian}
\gskew
    &= \frac{1}{q\pi} \int_{0}^{\infty} dt
        \, \cos\left(\tau\, {(st)}^\alpha + \frac{x}{q}\, st \right)  e^{-t^2/2} 
    && (x \in \mathbb{R}, s \ge 0)
\\  & \where \notag
        q = \cos(\theta\pi/2)^{1/\alpha},
        \,\, \tau = \tan(\theta\pi/2), \,\, \theta \ne 1. 
    && 
\end{align}

\qedbar
\end{definition}

Feller's parameterization provides a clean environment to dissect the subordination structure
in the stable CF.
\iftoggle{fullpaper}{See Section \ref{section:skew-kernel} for more detail.}{}  

In general, \(\gskew\)
is not a distribution, since some parts of the function can be negative. 
But it satisfies \(\int_{-\infty}^{\infty} \gskew dx = 1/s\). 
One exception is \(g_{\alpha}^{0}(x,s) = \mcN(xs)\).
Also note that \(g_{\alpha}^{\theta}(x,0) = 1/q \sqrt{2\pi}\), independent of \(x\).

Most of the difficulty in computing the \(\alpha\)-stable PDF arises from integrating 
some kind of highly oscillatory functions.
See Section 3 of \cite{Nolan:2020} for a general discussion.

In our approach, that difficulty is concentrated in \(\gskew\) for a non-zero \(\theta\).
The range of the integral is confined by \(e^{-t^2/2}\), which is nice.
But when the oscillation frequency is high, it gets harder to integrate with 
the ordinary \texttt{quad} utility in the \texttt{scipy} package.
We switch to the \texttt{quadosc} function in the \texttt{mpmath} package\cite{mpmath:130}
with more success in accuracy, but some performance is sacrificed.


\begin{lemma}\label{lemma-frac-chi1}
\textbf{Fractional Chi-1}.
When the skewness is absent: \(\theta = 0\), we get \(q = 1, \tau = 0\), 
there is no "complicated oscillation" from the cosine function.
\(g_{\alpha}^{0}(x,s) = \mcN(xs)\) greatly simplifies the matter.
This leads to the SaS distribution as
\(P_{\alpha} := L_\alpha^0 \sim \mcN / \chibar_{\alpha,1}\).
That is, all the shape information in SaS is determined by \(\chibar_{\alpha,1}\).

The \(\chibar_{\alpha,1}\) term is called \emph{fractional chi-1},
which is the fractional version of the \(\chi\) distribution of one degree of freedom 
from the absolute of a random variable defined as such. 
This is a new concept, used as a building block for \(\chimeanDist\).

Here we briefly explain how to go from \(\chibar_{\alpha,1}\) to \(\chi_{\alpha,k}\).
\iftoggle{fullpaper}{See Section \ref{section:frac-chi-physical} in full detail.}{}

\(\chibar_{\alpha,1}\) is \(\scN_{\frac{\alpha}{2}}(x; \sigma=\frac{1}{\sqrt{2}}, d=0, p=\alpha)\),
whose PDF is
\begin{align} \label{eq:main-frac-chi1-pdf}
\chibar_{\alpha,1}(x)
    &=  2 x^{-1} \, \FWrfn{\alphahalf}{{(\sqrt{2}x)}^{\alpha}}
\end{align}

\(\chibar_{\alpha,1}\) has some peculiar behaviors. 
For \(\alpha < 1\), \(\chibar_{\alpha,1}(x)\) diverges as \(x^{\alpha-1}\) for small \(x\). 
And
\begin{align*}
\lim_{\alpha \to 0} \chibar_{\alpha,1}(x) &\approx \alpha/(e x)  &(x > \alpha \,\, \text{approximately})
\\
\chibar_{1,1}(x) &= \sqrt{2/\pi} \, e^{-x^2/2}
\\
\chibar_{2,1}(x) &= \delta(x-1/\sqrt{2})
\end{align*}
The first line implies 
\(\displaystyle \lim_{\alpha \to 0} P_{\alpha}(x) \approx \chibar_{\alpha,1}(x) / 2 \approx \alpha/(2e x)\),
for not too small \(x\), e.g. \(x > \alpha/2\).

\textbf{The \(x^{k-2}\) Term.}
In the \(k > 0\) case, the PDF \(\chimeanDist(x)\) comes from \(\chibar_{\alpha,1}(x)\) 
multiplied by the surface area of the \((k-1)\)-dimensional sphere, which is \(\propto x^{k-1}\).
Note that \(\chibar_{\alpha,1}(x) \propto x^{-1} \FWrfn{\frac{\alpha}{2}}{ {\left( 2 x^2 \right)}^{\alpha/2} } \). 
Hence, \(\chi_{\alpha,k}(x) \propto x^{k-2} \FWrfn{\frac{\alpha}{2}}{ {\left( 2 x^2 \right)}^{\alpha/2} } \).
This is the origin of the crucial \(x^{k-2}\) term.

Once \(\chi_{\alpha,k}\) is defined, it is standardized to FCM
such that \(\chimeanDist := \chi_{\alpha,k} / \fcmSigma\). 
FCM is used to construct all the subsequent two-sided distributions.

\end{lemma}

%% file: section-main-result2.tex
\subsection{New Two-Sided Distributions}
\label{section:two-sided-definitions}

We focus on the symmetric case where \(\theta = 0\) first, 
and defer the skew case \(\theta \ne 0\) 
to the latter.

When \(\theta = 0\), we have \(g_{\alpha}^{0}(x,s) = \mcN(xs)\),
independent of \(\alpha\).
The new generalized symmetric distribution is solely based on FCM 
that subsumes the SaS, Student's t, and exponential power distributions.

\begin{definition}\label{def:main-gsas}
(Generalized symmetric \(\alpha\)-stable distribution (GSaS))
GSaS is elegantly constructed as a ratio distribution: 
\(L_{\alpha,k} \sim \mcN / \, \chimeanDist\).
Hence, its PDF is

\begin{align}\label{eq:main-gsas-def}
\gsas{x}
    &:= \int_{0}^{\infty} s \, ds \,
      \mcN(xs)
      \, \chimeanDist(s)
    && (x \in \R, \alpha \in (0,2], k \in \Rn0)
\end{align}
It follows naturally that \(L_{\alpha,1} = P_{\alpha}\) for SaS,
and \(L_{1,k} = t_k\) for Student's t when \(k > 0\). On the other hand,
\(L_{\alpha,-1} = \mathcal{E}_{\alpha}\) for the exponential power.

\begin{lemma} \label{lemma-gsas-clm} 
(GSaS Version of Central Limit Theorem)
When \(|k| \gg 1\), it is straightforward from Lemma \ref{lemma-fcm-delta1} that 
\(\gsasDist\) becomes a normal distribution.
This statement is similar to the Central Limit Theorem (CLM). 
The variance asymptotics of \(\gsasDist\)
is \((\alpha^{2/\alpha})^{\sgn(k)}\) at \(\alpha \to 2\) or \(|k| \to \infty\).
This is consistent with both \(\alpha\)-stable at \(\alpha=2\) and Student's t at \(\alpha=1\) by design.

\qedbar
\end{lemma}

Its CF becomes a product distribution, according to \eqref{mix-cf}, such as

\begin{align}\label{eq:main-gsas-cf}
\phi_{\alpha,k}(\zeta) := 
\text{CF}\{\gsasDist\}(\zeta)
    &= \sqrt{2\pi} \int_{0}^{\infty} \, ds \,
      \mcN\left(\frac{\zeta}{s}\right)
      \, \chimeanDist(s)
    && (\zeta \in \mathbb{R})
\end{align}

For \(k > 0\), it is an expanded form of the inverse lambda decomposition \eqref{eq:lambda-decomp}.
It describes, in the \(\theta=0\) case, how the \(\alpha\)-stable law is modified from
\(\exp(-{|\zeta|}^\alpha)\) at \(k = 1\)
to a generalized form of \(k > 1\).

Since \(\gsas{x}\) with \(k > 0\) covers Student's t as well,
mathematically speaking, \(\phi_{\alpha,k}(\zeta)\) forms a high transcendental function
that subsumes both the stretched exponential function and 
the modified Bessel function of the second kind (due to \(\text{CF}\{t_{k}\}(\zeta)\)).
\iftoggle{fullpaper}{See Section \ref{section:gsas-cf} for more detail.}{} 

\qedbar
\end{definition}

\begin{definition}[Generalized exponential power distribution (GEP)]\label{def:gexppow}
According to Lemma \ref{lemma-mix-cf-dist}, \(\phi_{\alpha,k}(\zeta)\) can be treated as a distribution.
In this case, its \(k=1\) base is the exponential power distribution whose PDF is 
\(\mathcal{E}_\alpha(x) = e^{-{|x|}^\alpha}/2\GammaAlpha\)\cite{GSL_ExponentialPowerDistribution, Wolfram_ExponentialPowerDistribution}.
We take it to the next level:
\begin{align} \label{eq:main-gexppow-pdf}
\gepDist(x) 
    &:= \frac{1}{\E(X|\chimeanDist)} 
    \, \int_{0}^{\infty} \, ds \,
      \mcN\left(\frac{x}{s}\right)
      \, \chimeanDist(s) 
    && (x \in \mathbb{R}, k > 0)
\\ \label{eq:main-gexppow-pdf-by-dagger}
    &= \int_{0}^{\infty} \, \frac{ds}{s} \,
      \mcN\left(\frac{x}{s}\right)
      \, \chibar_{\alpha, -k}^\dagger(s)  
\end{align}
where the normalization constant \(\E(X|\chimeanDist)\) is
FCM's first moment in \eqref{eq:fcm-moment}. 
And \(\mathcal{E}_{\alpha,1} = \mathcal{E}_\alpha\).

Furthermore, according to Lemma \ref{lemma-mix--ratio-vs-prod},
\eqref{eq:main-gexppow-pdf-by-dagger} is inverted to a ratio distribution form:
\begin{align} \label{eq:main-gexppow-pdf-neg-k}
\gepDist(x) 
    &= \int_{0}^{\infty} s \, ds \,
      \mcN\left(xs\right)
      \, \chibar_{\alpha, -k}(s) 
\\ \notag
    &= L_{\alpha,-k}(x)
\end{align}
Therefore, GSaS subsumes GEP beautifully
when the negative \(k\) domain is defined properly in FCM.

The \(k\) dimension mitigates a flaw in \(\mathcal{E}_\alpha(x)\): when \(\alpha \le 1\),
the slope of the PDF is discontinuous at \(x = 0\).
That is, \(\lim_{x \to 0^-} \frac{d}{dx} \, \mathcal{E}_{\alpha}(x)\) \(\ne 
\lim_{x \to 0^+} \frac{d}{dx} \, \mathcal{E}_{\alpha}(x) \).  
Now as long as \(k \gt 1\), \(\frac{d}{dx} \, \mathcal{E}_{\alpha,k}(x)\) 
is continuous everywhere.
\iftoggle{fullpaper}{See Section \ref{section:gen-exp-pow-dist} for more detail.}{}

\qedbar
\end{definition}

Since \(\mcN(xs)\) in \eqref{eq:main-gsas-def} 
doesn't take either \(\alpha\) or \(k\). All the shape information comes
from our innovation, \(\chimeanDist(s)\). 
For instance, the peak PDF \(L_{\alpha,k}(0)\) is 
the first moment of \(\chimeanDist\).


\textbf{Excess Kurtosis}. 
The moments of GSaS have closed form solutions. Of particular interest is the excess kurtosis, 
plotted in Figure \ref{fig:gsas_ex_kurt_by_alpha} in the \((k,\alpha)\) coordinate. 
Notice that a major division occurs along the line of \(k = 5 - \alpha\). 
In the region where \(0 < k \le 5 - \alpha\),
there are complicated patterns caused by the infinities of the gamma function. 
Only small pockets of valid kurtosis exist.

\begin{lemma}\label{lemma-main-kurtosis}
In the region where \(k > 5 - \alpha\), the excess kurtosis is 
a continuous function with positive values.
At large \(k\)'s, the closed form of the moments can be expanded by Sterling's formula.
The excess kurtosis (exKurt) becomes part of a linear equation:

\begin{align}\label{eq:gsas-linear-kurtosis}
{\left( s -\frac{1}{2} \right)} = 
    \left( \frac{k-3}{4} \right) \,
    \log\left( 1 + \frac{\text{exKurt}}{3} \right),
    \quad \where \, s = 1/\alpha
\end{align}

This equation shows how GSaS works under the \textbf{Central Limit Theorem}.
GSaS approaches a normal distribution: \(L_{\alpha,k} \to \mcN\) when the excess kurtosis becomes zero.
This can happen from two directions:
when \(\alpha \to 2\), or when \(k \to \infty\).

\qedbar
\end{lemma}

The contour plot of excess kurtosis is shown in the \((k,s)\) coordinate 
in Figure \ref{fig:gsas_ex_kurt_by_s}. 
It is visually amusing. Notice the singular point at \(s=1/2, k=3\).

The excess kurtosis of \(\gepDist\) has a closed form too:
\begin{align} \label{eq:gexppow-exkurt}
\text{exKurt}(\gepDist) &=
    \frac
    {3 \, \E(X|\chimeanDist) \, \E(X^5|\chimeanDist)}
    {\E(X^3|\chimeanDist)^2} 
    - 3
\end{align}
which is shown in Figure \ref{fig:gexppow_ex_kurt_by_s}.
The linear contour levels are clear indication that the distribution 
has an orderly structure in terms of kurtosis.

\begin{figure}[htp]
    \centering
    \includegraphics[width=7in]{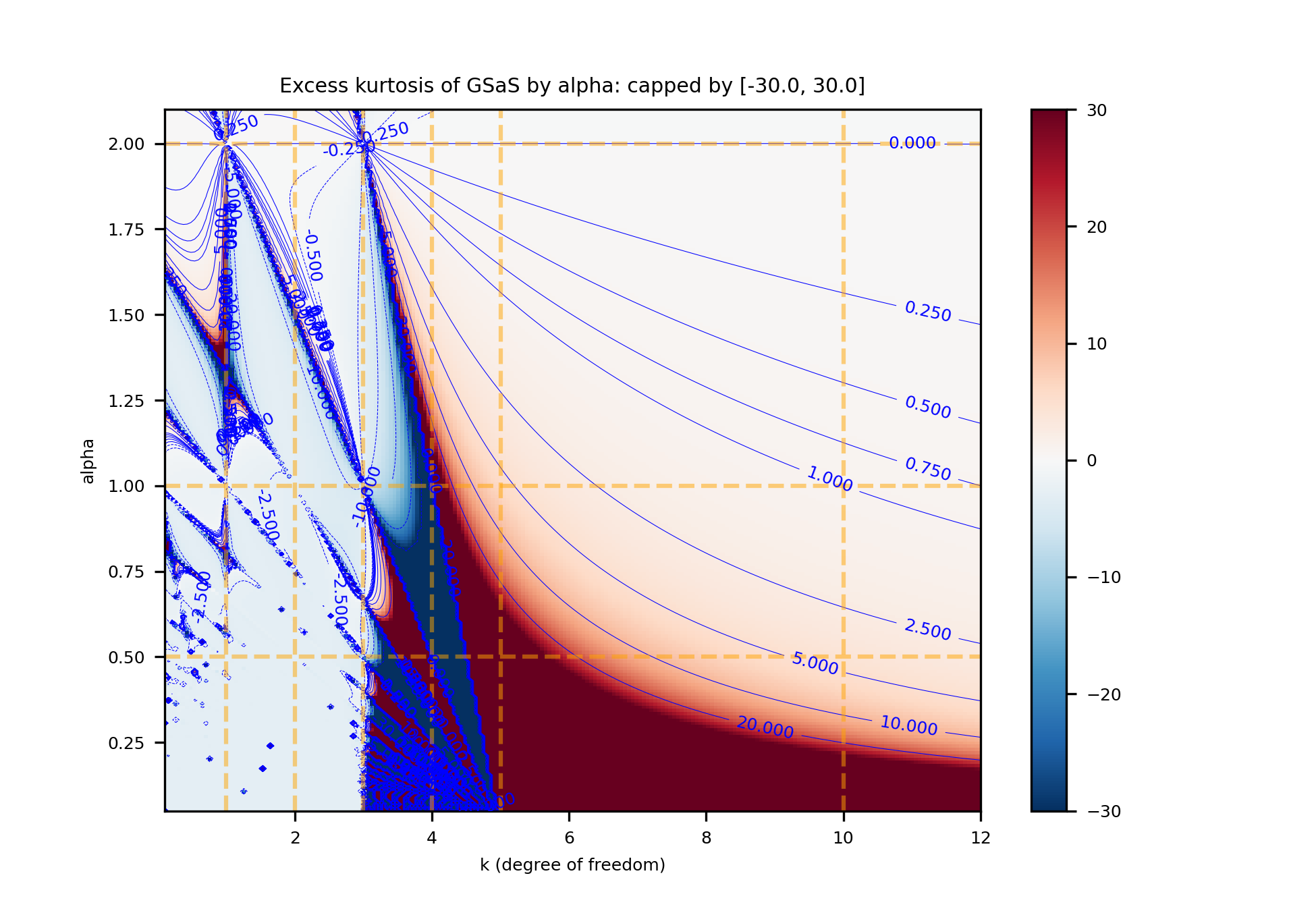}
    \caption{The contour plot of excess kurtosis in GSaS by \((k, \alpha)\). }
    \label{fig:gsas_ex_kurt_by_alpha}
\end{figure}

\begin{figure}[htp]
    \centering
    \includegraphics[width=7in]{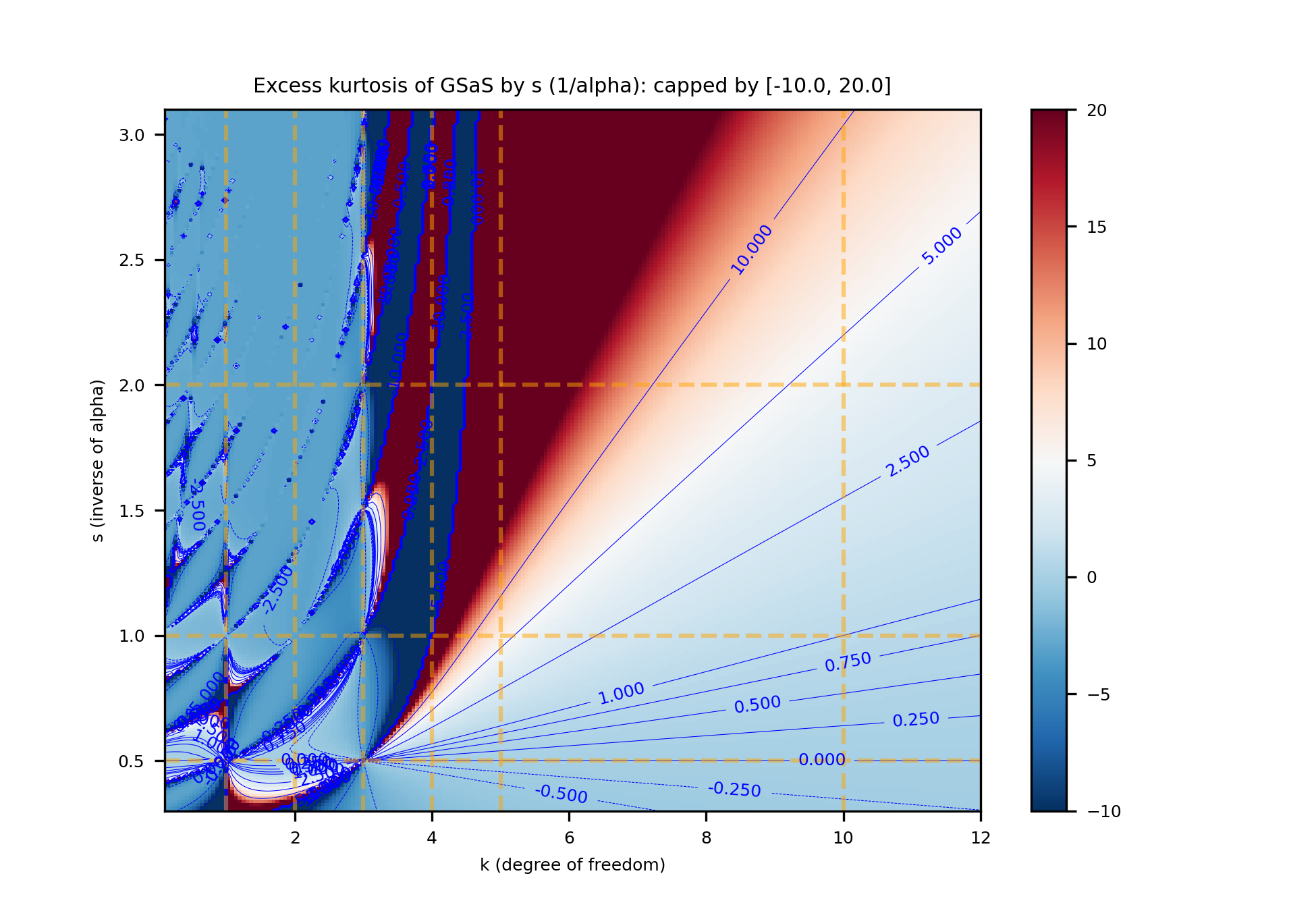}
    \caption{The contour plot of excess kurtosis in GSaS by \((k, s)\)
        where \(s = 1/\alpha\),
        in order to show the linear equation described in 
        \eqref{eq:gsas-linear-kurtosis} for large \(k\)'s.}
    \label{fig:gsas_ex_kurt_by_s}
\end{figure}

\begin{figure}[htp]
    \centering
    \includegraphics[width=7in]{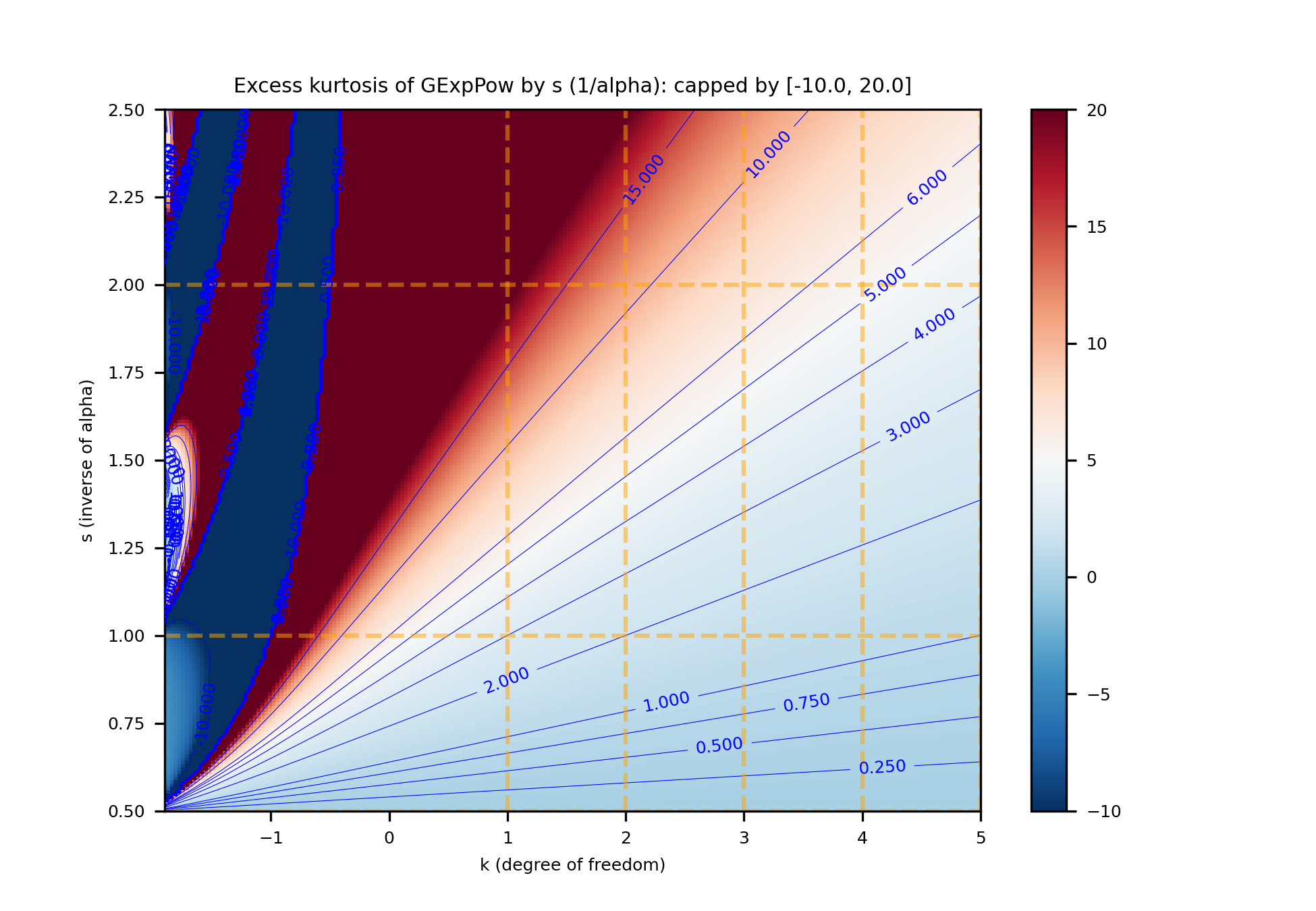}
    \caption{The contour plot of excess kurtosis \eqref{eq:gexppow-exkurt}
        of the generalized exponential power distribution \(\gepDist\)
        by \((k, s)\) where \(s = 1/\alpha\).
        The linear contour levels can be clearly observed.
        We allow the formula to be evaluated down to \(k = -2\)
        to show the singular point at \(s = 1/2, k = -2\).}
    \label{fig:gexppow_ex_kurt_by_s}
\end{figure}

\begin{definition}[Generalized \(\alpha\)-stable distribution (GAS)]\label{def:gas}

GAS is an experimental skew distribution based on GSaS. 
The GAS \(L_{\alpha,k}^{\theta}\) is a ratio such that 
\(L_{\alpha,k}^{\theta} \sim g_{\alpha}^{\theta} / \chimeanDist\). 
Hence, its PDF is

\begin{align}\label{eq:main-gas-def}
\gas{x}
    &:= \int_{0}^{\infty} s \, ds \,
      g_{\alpha}^{\theta}(x,s) \,
      \, \chimeanDist(s)
    && (x \in \mathbb{R})
\end{align}
However, the PDF of a distribution must be positive-definite (Bochner's Theorem).
This narrows the range of \(\theta\) as \(k\) increases.

\qedbar
\end{definition}
GAS adds the \(k\) dimension to \(\alpha\)-stable by expanding 
from \(\chibar_{\alpha,1}\) to \(\chimeanDist\) in the integral.
It subsumes \(\alpha\)-stable as \(L_{\alpha}^{\theta} \sim L_{\alpha,1}^{\theta}\),
and subsumes Student's t as \(t_{k} \sim L_{1,k}^{0}\).
\iftoggle{fullpaper}{See Section \ref{section:gas} for more detail.}{} 

GAS admits valid variance, skewness, and kurtosis, 
\emph{as long as there are a few degrees of freedom}, 
just like you would expect in Student's t. 
(See Lemma \ref{lemma-main-kurtosis} about the validity of kurtosis.)
It inherits the skewness from the \(\alpha\)-stable law - a property that 
very few two-sided skew distributions have, if any. 

Most real-world data analyses begin with the summary statistics of the first four moments.
Plus a few more, such as the peak PDF, and the 95\% and 99\% confidence levels.
GAS can address all of them. This is a large improvement to \(\alpha\)-stable.

\subsection{Explanation of the S\&P 500 Fit}
\label{section:sp500-fit-explained}

Next, we explain how we fitted the S\&P 500 daily log return data 
that produced Figure \ref{fig:sp500}.
We used GSaS in the first pass. The values 
of \(\alpha\) and \(k\) are obtained such that the GSaS agree with the empirical kurtosis (20) 
and standardized peak density (0.71).
In the second pass, a small skewness is added to GAS and we fine-tuned the fit.

However, the solution for the S\&P 500 fit falls in an unconventional region as shown 
in Figure \ref{fig:pyro_sp500_gsas_region}. 
It is in a small pocket
near \(\alpha=1, k= 3\). Conventionally, as we know from Student's t, the kurtosis only exists 
for \(k\) greater than 4.
But in our new \((\alpha,k)\) space, it is possible
for \(k\) to be less than 4 with valid kurtosis, although only in small pockets. 
We found the solution by tracing the contours 
of the peak PDF and excess kurtosis as shown in Figure \ref{fig:pyro_sp500_gsas_search}.  
(The peak PDF has a closed form solution and is more stable than the kurtosis numerically.)
The intersection near \(\alpha=0.81, k=3.3\) is the solution that satisfies both properties.
This is the only solution for the S\&P 500 data 
in the low degrees of freedom region where \(k < 5\).

\begin{figure}[htp]
    \centering
    \includegraphics[width=5in]{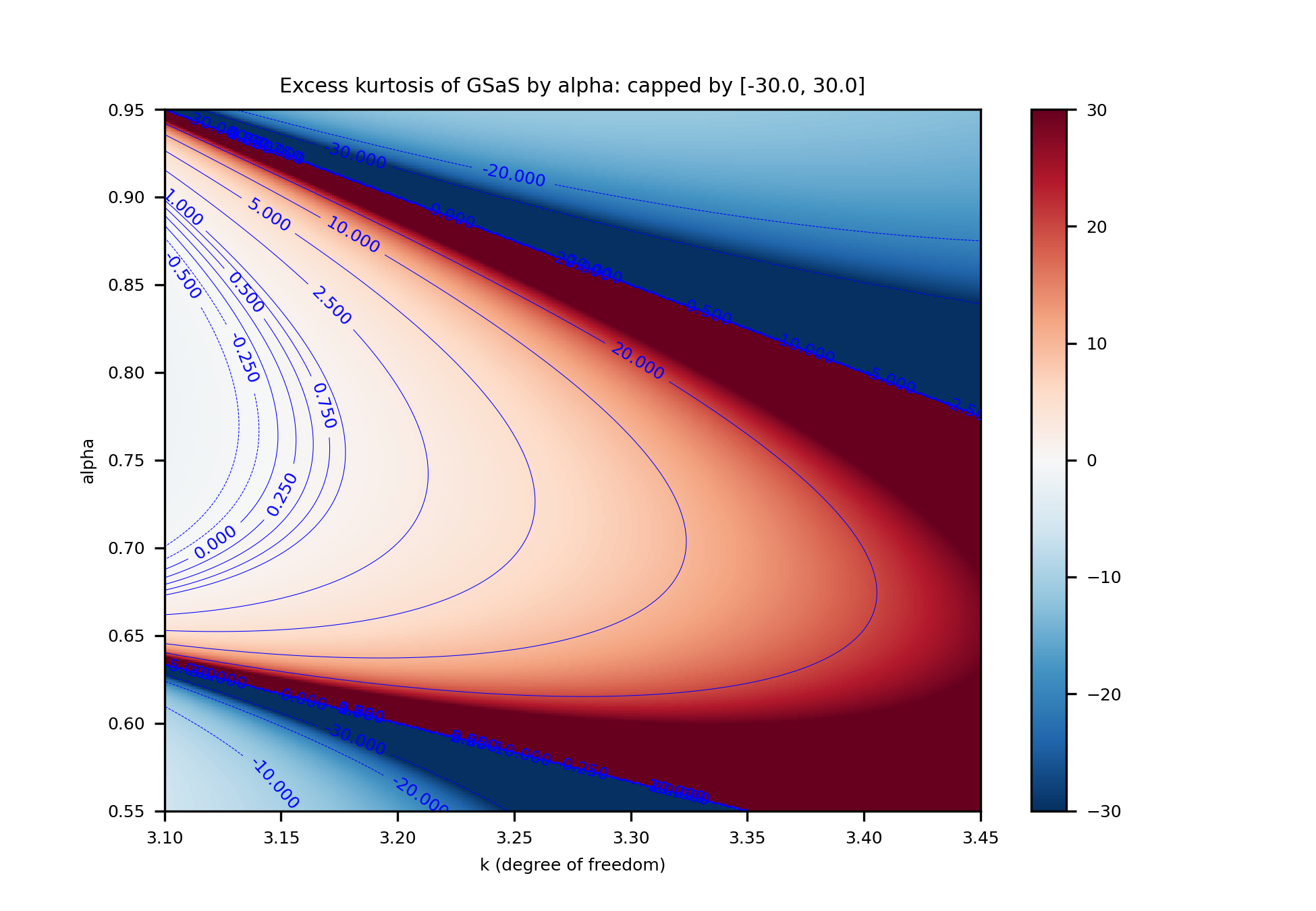}
    \caption{The region near \((\alpha=1, k= 3)\) that contains the S\&P 500 solution.}
    \label{fig:pyro_sp500_gsas_region}
\end{figure}

\begin{figure}[htp]
    \centering
    \includegraphics[width=5in]{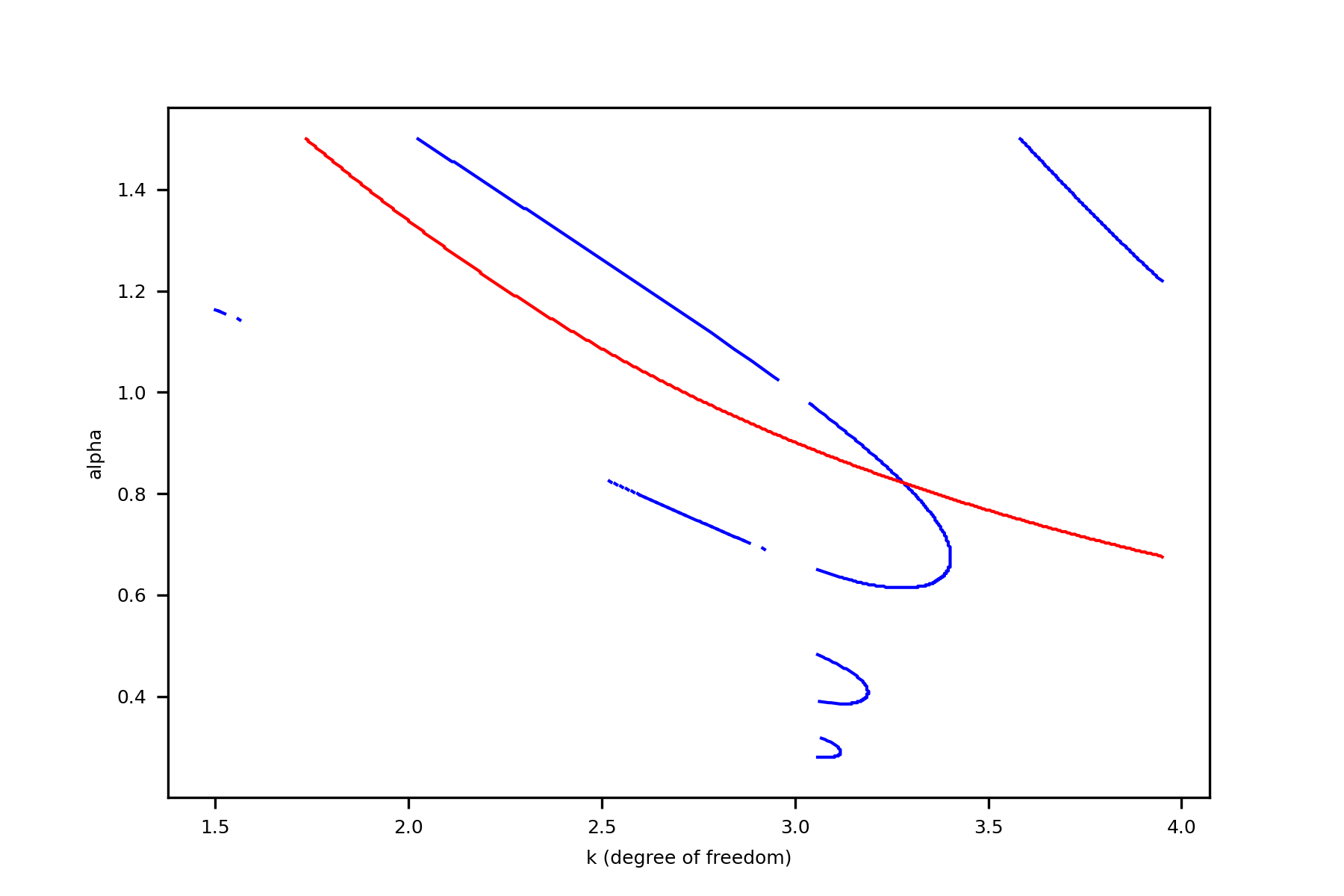}
    \caption{The S\&P 500 solution is found by tracing the contours of 
    the excess kurtosis (blue line) and the peak PDF (red line). 
    The intersection is the solution that satisfies both properties.}
    \label{fig:pyro_sp500_gsas_search}
\end{figure}

\clearpage  

%% file: section-multivariate.tex
\section{Multivariate Distribution}
\label{section:multivariate}

The appearance of \(\mcN(xs)\)
makes it possible to extend GSaS to 
an \(n\)-dimensional multivariate symmetric distribution 
with the concept of the elliptical distribution\cite{Tong:1990}.
We propose two viable options. 
And we use the daily return distributions of the S\&P 500 Index (SPX) and CBOE VIX Index (VIX) since 1990 
to examine the bivariate use case.

We need to upgrade \(\mcN(xs)\) to a multivariate normal distribution properly.
To set it up, we vectorize \(x\) to \(\mathbf{x} \in \mathbb{R}^n\),
and introduce the covariance matrix \(\mathbf{\Sigma} \in \mathbb{R}^{n \times n}\),
which is invertible to \(\mathbf{\Sigma}^{-1}\). Its determinant is \(|\mathbf{\Sigma}|\).
There are two routes to go from here.

\subsection{Multivariate of the First Kind - Elliptical}
\label{section:multi-elliptical}

\textbf{The first route} is simpler and produces very meaningful results. 
All dimensions share the same shape from a single FCM.
So \((xs)^2\) in \(\mcN(xs)\) is replaced to
\(\mathbf{x}^\intercal \mathbf{\Sigma}^{-1} \mathbf{x} \, s^2 \).
Such elliptical distribution closely resembles that of a multivariate normal.
The ellipse structure \(\mathbf{x}^\intercal \mathbf{\Sigma}^{-1} \mathbf{x}\) is perserved.

We use an extended notation \(\mcN\left( \mathbf{x}; \mathbf{\mu}, \mathbf{\Sigma} \right)\) to represent
the PDF of a multivariate normal distribution. 
But for the purpose of exploring fundamental properties, we don't want to be bothered 
with the location parameter. Therefore, we set its \(\mathbf{\mu} = 0\).

\begin{definition}(GSaS Elliptical Multivariate Distribution)
The PDF of the first kind is
\begin{align}\label{eq:ell-gsas-def}
L_{\alpha, k}(\mathbf{x}; \mathbf{\Sigma})
    &:= 
    \int_{0}^{\infty} \, ds \,
      \, \chimeanDist(s) \,
      \mcN\left( \mathbf{x}; 0, \frac{\mathbf{\Sigma}}{s^2} \right)
\\ \label{eq:ell-gsas-def2}
    &=
    \frac{1}{(2\pi)^{n/2} {|\mathbf{\Sigma}|}^{1/2}}    
    \int_{0}^{\infty} s^{n} \, ds \,
      \, \chimeanDist(s)
      \left[
        \exp\left( -\frac{s^2}{2} \, \mathbf{x}^\intercal \mathbf{\Sigma}^{-1} \mathbf{x} \right)
      \right]
\end{align}
The \texttt{scipy.stats.multivariate\_normal} function is used in our prototype implementation.
It is obvious from the first line that the total density is integrated to 1.
\end{definition}

When \(\alpha=1, k > 0\), \(L_{1, k}(\mathbf{x}; \mathbf{\Sigma})\) becomes the multivariate t distribution.
This is validated numerically with \texttt{scipy.stats.multivariate\_t}.
On the other hand, due to Lemma \ref{lemma-gsas-clm}, 
\begin{align*}
    \lim_{\alpha \to 2 \,\text{or}\, k \to \infty} 
    L_{\alpha, k}\left(\mathbf{x}; {\alpha^{-2 \alpha}} \, \mathbf{\Sigma} \right) 
    = \mcN\left( \mathbf{x}; 0, \mathbf{\Sigma} \right)
\end{align*}
These two validations are critically \emph{important}. 
They show that the univariate GSaS extends naturally to multivariate
through the combination of the elliptical distribution and FCM.

From \eqref{eq:ell-gsas-def2}, its peak PDF is
\begin{align}\label{eq:mv-ell-gsas-pdf-at-zero}
L_{\alpha, k}(\mathbf{0}; \mathbf{\Sigma})
    &= \frac{\E(X^{n}|\chimeanDist)}{(2\pi)^{n/2} {|\mathbf{\Sigma}|}^{1/2}}  
\end{align}

Its variance (aka the second moment) is
\begin{align}\label{eq:mv-ell-gsas-cov}
\E(X_i X_j)
    &= \E(X^{-2}|\chimeanDist) \, \mathbf{\Sigma}_{i,j} 
\end{align}

Its \(i\)-th dimension marginal PDF is a GSaS:
\begin{align}\label{eq:mv-ell-gsas-1d-pdf}
L_{\alpha, k}^{(i)}(x; \mathbf{\Sigma})
    &= \frac{1}{\sigma_i} \gsasDist\left(\frac{x}{\sigma_i}\right),
    \quad \where \sigma_i = \sqrt{\mathbf{\Sigma}_{i,i}}
\end{align}

\begin{proof}
The proof of \eqref{eq:mv-ell-gsas-1d-pdf} is straightforward. 
The marginal distribution of a multivariate normal distribution is 
the down-sized multivariate normal distribution with the covariance matrix from
those selected dimensions. See Section 3.3.1 of \cite{Tong:1990}.
\end{proof}

Hence, the \(i\)-th marginal variance is \(\E(X^{2}|\gsasDist) \, \mathbf{\Sigma}_{i,i}\),
the same as \(\E(X_i^2)\) in \eqref{eq:mv-ell-gsas-cov} since
\(\E(X^{2}|\gsasDist) = \E(X^{-2}|\chimeanDist)\) from Lemma \ref{lemma-gsas-moments}. 
They match nicely. 

The marginal peak density is 
\(L_{\alpha, k}^{(i)}(0; \mathbf{\Sigma}) = \frac{1}{\sqrt{2\pi} \sigma_i} \E(X|\chimeanDist)\).

The multivariate extension of GEP (aka \(k < 0\)) has not been validated 
since there is no implementation in the \texttt{scipy} package.
But we do find it useful as an alternative model in Section \ref{section:multi-spx-vix}.

\subsection{Multivariate of the Second Kind - Adaptive}
\label{section:multi-adaptive}

\textbf{The second route} is more complex, but also more adaptive for practical use.
For instance, in a multi-asset scenario, it allows each asset 
to have its own shape according to its marginal distribution.
The downside is that the elliptic structure is distorted.

In the previous method, we learn the advantage of using multivariate normal to
construct our multivariate PDF. The \(\mcN\left( \mathbf{x}; ... \right)\) term
can be integrated away easily to make sure the total density is equal to 1.
We follow the same methodology but expand on the FCM structure.

For simplicity, we use \(\Theta_i\) to represent 
the \(i\)-th parameterization: \((\alpha_i, k_i)\).
And \(\Theta\) is the collection of all \(\Theta_i\)'s.
We vectorize \(s\) to \(\mathbf{s} \in \mathbb{R}^n\).
Each \(s_i\) comes from a distinct \(\chibar_{\Theta_i}\) in that dimension.

The covariance matrix \(\mathbf{\Sigma}\) is sandwiched by \(1/\mathbf{s}\) such that 
the covariance matrix becomes
\((\mathbf{s}^\dagger)^\intercal \,\mathbf{\Sigma}\, \mathbf{s}^\dagger\)
where \(\mathbf{s}^\dagger = \diag(\mathbf{s})^{-1}\).
It is then enveloped by the integrals of the multi-dimensional FCMs.
To put it together -

\begin{definition}(GSaS Adaptive Multivariate Distribution)
The PDF of the second kind is
\begin{align}\label{eq:mv-adp-gsas-def}
L_{\Theta}(\mathbf{x}; \mathbf{\Sigma})
    &:= 
    \prod_i \int_{0}^{\infty} \, ds_i \,
      \, \chibar_{\Theta_i}(s_i) \,
        \mcN\left( \mathbf{x}; 0, 
            (\mathbf{s}^\dagger)^\intercal \,\mathbf{\Sigma}\, \mathbf{s}^\dagger
        \right),
    \quad \where \mathbf{s}^\dagger = \diag(\mathbf{s})^{-1}
\end{align}
This object is far more computationally intense than \(L_{\alpha, k}(\mathbf{x}; \mathbf{\Sigma})\)
due to the multi-dimensional integral on \(\mathbf{s}\).
\end{definition}

Note that, in the case of identity \(\mathbf{\Sigma}\), \(L_{\Theta}(\mathbf{x}; \mathbf{\Sigma})\)  
does factorize into the product of the marginal one-dimensional distributions.
This is a positive feature.
However, this PDF retains the circular contour only in the center region.
The contour could be distorted into rectangle-barbell shapes with round edges in the tail regions.

Its peak PDF is
\begin{align}\label{eq:mv-adp-gsas-pdf-at-zero}
L_{\Theta}(\mathbf{0}; \mathbf{\Sigma})
    &= \frac{\prod_i \E(X|\chibar_{\Theta_i})}{(2\pi)^{n/2} {|\mathbf{\Sigma}|}^{1/2}}  
\end{align}

Its variance (aka the second moment) is
\begin{align}\label{eq:mv-adp-gsas-cov}
\E(X_i X_j)
    &= 
    \left\{
    \begin{array}{ll}
        \E(X^{-2}|\chibar_{\Theta_i}) \, \mathbf{\Sigma}_{i,i} 
        &, \, \text{for} \,\, i = j.
    \\ 
        \E(X^{-1}|\chibar_{\Theta_i}) \, \E(X^{-1}|\chibar_{\Theta_j}) \, \mathbf{\Sigma}_{i,j} 
        &, \, \text{for} \,\, i \ne j.
    \end{array}
    \right.
\end{align}

Its \(i\)-th dimension marginal PDF is a GSaS:
\begin{align}\label{eq:mv-adp-gsas-1d-pdf}
L_{\Theta}^{(i)}(x; \mathbf{\Sigma})
    &= \frac{1}{\sigma_i} L_{\alpha_i, k_i}\left(\frac{x}{\sigma_i}\right),
    \quad \where \sigma_i = \sqrt{\mathbf{\Sigma}_{i,i}}
\end{align}
Such marginal distribution can adapt to each sample's shape.
This is the major improvement over \eqref{eq:mv-ell-gsas-1d-pdf}.
The \(i\)-th marginal variance is \(\E(X^{2}|\gsasDist) \, \mathbf{\Sigma}_{i,i} = \E(X_i^2)\)
in \eqref{eq:mv-adp-gsas-cov}. They match nicely.

The \(i\)-th marginal peak density is 
\(L_{\Theta}^{(i)}(0; \mathbf{\Sigma}) = \frac{1}{\sqrt{2\pi} \sigma_i} \E(X|\chibar_{\alpha_i,k_i})\).
Hence, the peak density ratio between marginal and multivariate is
\begin{align*}
\frac{\prod_i L_{\Theta}^{(i)}(0; \mathbf{\Sigma})}
{L_{\Theta}(\mathbf{0}; \mathbf{\Sigma})}
    &= \left[ \frac{|\mathbf{\Sigma}|} {\prod_i \mathbf{\Sigma}_{i,i}} \right]^{1/2}
\end{align*}
In the bivariate case, this ratio becomes \(\sqrt{1-\rho^2}\)
where \(\rho\) is the correlation and \(|\rho| \le 1\).
On one hand, when \(\rho = 0\), this ratio is 1.
On the other hand, the closer \(\rho\) is to 1, the smaller this ratio is.
Conversely, the ratio from the sample data implies the model correlation.

In conclusion, the differences between the elliptical version and the adaptive version
are (a) how the peak density is formed; (b) the shapes of the marginal distributions;
and (c) the off-diagonal elements of the variance.

\subsection{Bivariate Application - SPX and VIX}
\label{section:multi-spx-vix}

In this section, we use the bivariate distributions
to model the joint daily return distributions of 
the S\&P 500 Index (SPX) and CBOE VIX Index (VIX) from 1/1990 to 3/2024.
We use this case to demonstrate the fitting procedure and visualize the features mentioned above.
This is the pinnacle of this work.

First, the sample statistics are calculated from the full sample. 
The sample covariance matrix is \(\mathbf{\Sigma}=\)
[[ 0.00486232, -0.00055669], [-0.00055669, 0.0001304]].
The correlation is \(\rho = -0.70\). 
The excess kurtoses are 15.8 for VIX, and 17.4 for SPX.

Second, we fit the marginal distributions of VIX and SPX with two GSaS's.
The result is shown in Figure \ref{fig:spx_vix_1d}. 
To manage a reasonable size of the bins in the histograms, 
the data is truncated at [-0.3, 0.3] for VIX,
and [-0.05, 0.05] for SPX. 200 bins are used.
Each GSaS is rescaled to the sample's
standard deviation: \(\sqrt{\mathbf{\Sigma}_{0,0}}\) and \(\sqrt{\mathbf{\Sigma}_{1,1}}\).
The model mean is the same as the sample mean in the plots.
For demonstration purpose, the means in the bivariate models are discarded.

Each dimension has its own shapes: \(\alpha_0 = 0.64, k_0 = 5.50\) for VIX;
and \(\alpha_1 = 0.88, k_1 = 3.20\) for SPX. 
The parameters are chosen such that the peak density and the excess kurtosis
between the model and sample are matched reasonably well.'
Skewness is not handled.

The same search procedure shown in Figure \ref{fig:pyro_sp500_gsas_search} 
is applied to each marginal distribution.

According to Figure \ref{fig:gsas_ex_kurt_by_alpha},
The high sample kurtosis in conjunction with \(\alpha\) slightly below 1
indicates \(k\) in the neighborhood of 5. This observation is true for VIX, but not for SPX.

SPX distribution has a sharper peak. The height exceeds the range that any \(k \ge 4\) can afford.
The solution is still located in the pocket near \(\alpha \approx 1\) and \(k \approx 3\),
as shown in Figures \ref{fig:pyro_sp500_gsas_region} and \ref{fig:pyro_sp500_gsas_search}.
This is a peculiar feature of SPX index, regardless much shorter history in this case.

To configure the bivariate model distributions, we designed an adjustment procedure
for the model \(\mathbf{\Sigma}\).
The main obstacle is that, if we simply feed the parameters obtained above to the bivariate models,
the peak joint density will come out lower than the sample's peak density, which is 1186,
printed in the upper right plot of Figure \ref{fig:spx_vix_2d}.

We attribute the shortage to insufficient correlation. The adjustment is treated 
as an optimization problem: The free variable is an adjustment factor.
The model variance of each dimension is scaled down by the factor, while
the absolute of the correlation is scaled up by the factor. 
Construct a new model covariance matrix.
The objective is to increase the model's peak joint density 
such that its percentage deviation to the sample peak density 
is less than the adjustment factor.

For the adaptive bivariate model, \(\rho\) is moved to -0.81. 
For the elliptical bivariate model, \(\rho\) is moved to -0.79. 
Since  the latter only takes one set of the shape parameter, 
we use the simple averages from the marginals: 
\(\langle\alpha\rangle = (\alpha_0 + \alpha_1)/2 = 0.76\), 
\(\langle{k}\rangle = (k_0 + k_1)/2 = 4.35\).

The model outputs are shown in the lower two contour plots in Figure \ref{fig:spx_vix_2d}.
The results are satisfactory. The tilts of the contours are in agreement with that of the sample
(Upper right). The adaptive contour is particularly interesting, or even a bit exaggerating. 
It has a rectangle-barbell shape, with four corners sticking out,
which looks more like what's from the data.
Elliptical contour appears to be too simple and naive.

\begin{figure}[htp]
    \centering
    \includegraphics[width=6in]{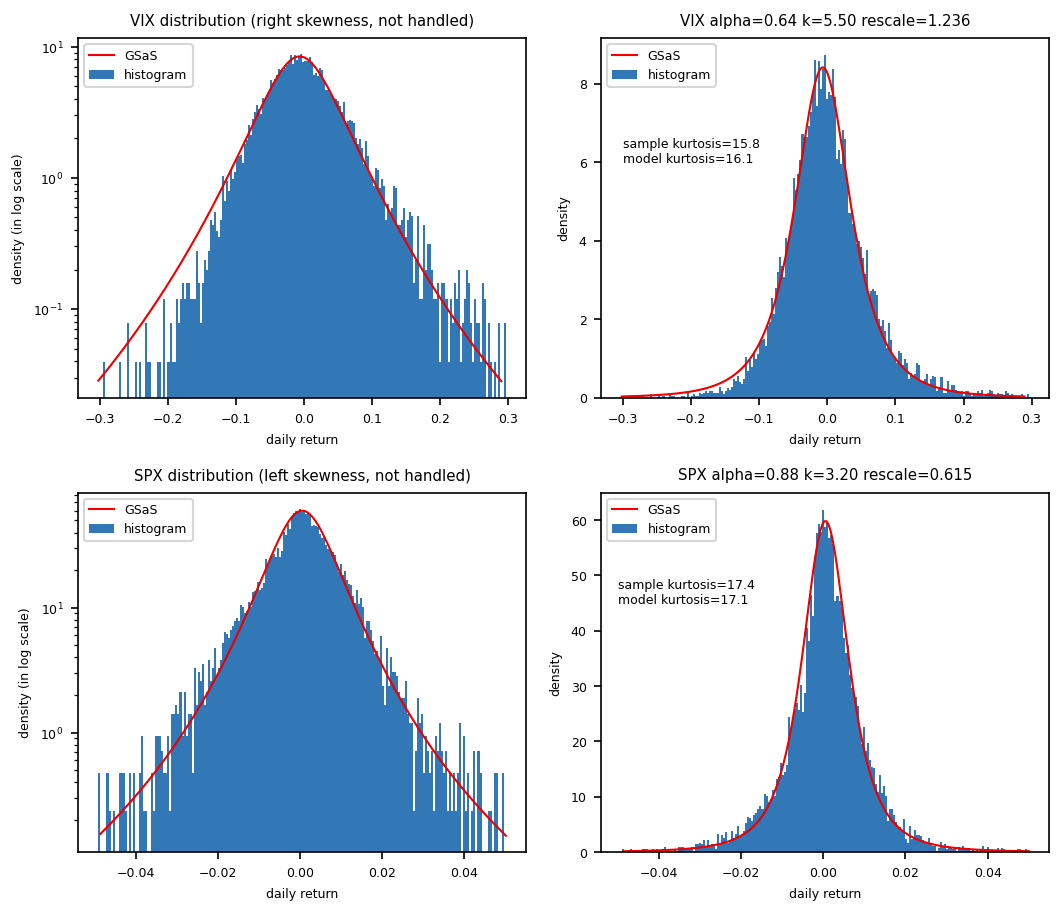}
    \caption{The The marginal distributions of VIX and SPX daily returns
        from 1/1990 to 3/2024. Each is fitted with a GSaS 
        by matching their peak densities, standard deviations, excess kurtosis,
        and shapes of the distribution.
        Each GSaS is shifted to match the mean of the data.
        The y-axes of the left plots are in the log scale. 
        }
    \label{fig:spx_vix_1d}
\end{figure}

\begin{figure}[htp]
    \centering
    \includegraphics[width=6in]{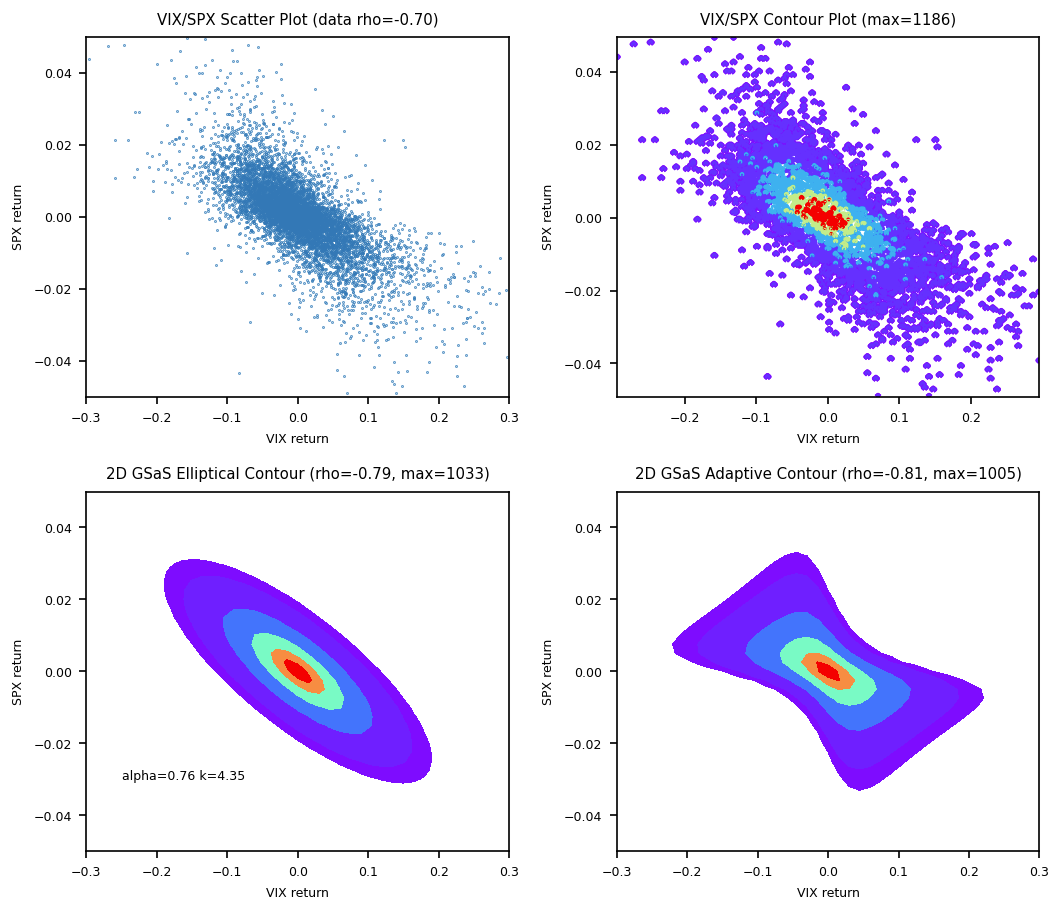}
    \caption{The bivariate distribution of VIX and SPX daily returns
        from 1/1990 to 3/2024. 
        \emph{Upper left}: The scatter plot of the data points,
            from which the sample covariance matrix and correlation are calculated.
        \emph{Upper right}: The contour plot from the 2D histogram of the data.
            200 bins in each dimension.
            This shows approximately what kind of contours should be expected. 
            They are more like rectangles than ellipses.
        \emph{Lower left}: The contour plot from the elliptical bivariate model distribution.
            The tilt of the ellipses agrees with the data (about \(45^\circ\) in the chart).
        \emph{Lower right}: The contour plot from the adaptive bivariate model distribution.
            Its shape is like rectangle-barbell shape, with four corners sticking out,
            which looks more like what's in the raw data.
        }
    \label{fig:spx_vix_2d}
\end{figure}

\clearpage

%% file: section-gsc.tex
\section{GSC: Generalized Stable Count Distribution}
\label{section:gsc}

\subsection{GSC PDF and Explanation}
\label{section:gsc-definition}

We recap the PDF of GSC \eqref{eq:main-gsc-pdf} from Definition \ref{def:gsc} 
and compare it with that of GG \eqref{eq:main-pdf-gg}\footnote{
The CDF of a GG is a regularized incomplete gamma function \(\Gamma(\cdot, \cdot)\).
Hence, it can also be parameterised such as 
\(\text{GenGamma}(\frac{x}{a}; s=\frac{d}{p}, c=p) := \Gamma\left( s, {\left(\frac{x}{a} \right)}^c \right)\).
} side by side:
\begin{align*}
\gscNx &:= 
    C \quad \quad {\left( \frac{x}{\sigma} \right)}^{d-1}
    \, W_{-\alpha,0} \left( -{\left( \frac{x}{\sigma} \right)}^{p} \right) 
\\
\pdfGG(x; a, d, p) &:=
    \frac{|p|}{a \Gamma(\frac{d}{p})} {\left(\frac{x}{a}\right)} ^{d-1} 
    \,\, e^{-(x/a)^p}
\end{align*}
The parallel in notations between the two PDFs should be obvious.

GSC is the fractional extension of GG.
The exponential function in GG is \emph{replaced} with 
the Wright function:
\(e^{-z} \rightarrow W_{-\alpha,0}(-z)\).
The additional parameter \(\alpha \in [0,1]\) controls the
shape of the Wright function. The scale parameter is changed
from \(a\) to \(\sigma\) to avoid confusion 
between \(a\) in GG and \(\alpha\) in GSC.

The other two parameters \(d, p\) are kept the same intentionally:
\(d\) is the \emph{degree of freedom} parameter;
\(p\) is the second shape parameter controlling the tail behavior \((p\ne 0, dp \ge 0)\).
In our use cases, \(p\) is either \(\alpha\) or \(2 \alpha\).
In the later case, \(p\) is exactly the stability index.

The range of \(\alpha\) is constrained to \([0,1]\) due to the analytic property of \(W_{-\alpha,0}(-z)\).
When \(\alpha \to 1\), it becomes a Dirac delta function \(\delta(z-1)\). 
This is a well known result.

On the other hand, when \(\alpha \to 0\), GSC becomes GG. 
See Section \ref{section:mapping-gg-0} for proof.
A broad class of distributions belong to GG, see Table \ref{tab:classic-map-zero}.

\subsection{Determination of C}
\label{section:gsc-const}

We derive the value of the normalization constant \(C\) in \eqref{eq:main-gsc-const} as following.

\begin{proof}
The normalization constant \(C\) in \eqref{eq:main-gsc-pdf} is obtained from the requirement 
that the integral of the PDF must be 1:
\begin{align*}
\int_0^{\infty} \, N_\alpha(x; \sigma, d, p) \, dx 
    &= \, \frac{C \sigma}{|p|} 
    \frac{\Gamma(\frac{d}{p})} {\Gamma(\frac{d}{p} \alpha)}
    = 1
\end{align*}
where the integral is carried out 
by the moment formula of the Wright function \eqref{eq:wright-moments}.

Hence, \(C\) is determined as
\begin{align*}
    C 
    &= \frac{|p|}{\sigma} 
    \frac{\Gamma(\frac{d}{p} \alpha)} {\Gamma(\frac{d}{p})} 
    &(\alpha \ne 0, d \ne 0)
\end{align*}

We typically constrain \(dp \ge 0\) and \(p\) is typically positive.
But it becomes negative in the inverse distribution and/or characteristic distribution types.
So we need \(|p|\) to ensure \(C\) is positive. 

For the case of \(\alpha \ne 0\) and \(d \to 0\), 
due to \eqref{eq:gamma-zero}, we have 
\begin{align*}
    C 
    &= \frac{|p|}{\sigma \alpha}
    &(\alpha \ne 0, d = 0)
\end{align*}
These two cases are combined to form \eqref{eq:main-gsc-const}.
\end{proof}

\input{sub-gsc-conn-gg}

\input{sub-gsc-mnt-sc-sv}

%% file: sub-gsc-conn-gg.tex

\subsection{Mapping between GSC and GG at alpha = 0}\label{section:mapping-gg-0}

All of the gamma-based, 
exponential-based distributions can be expressed by GSC at \(\alpha = 0\).
We prove the first line of \eqref{eq:main-gg-by-gsc} as following.

\begin{proof}
For certain proofs, it is more convenient to express the PDF of GSC in terms of 
the M-Wright function \(M_\alpha(z)\) in \eqref{eq:m-wright-M}:
\begin{align}
\label{eq:pdf-gsc-mw}
\gscNx &= \alpha C \, {\left( \frac{x}{\sigma} \right)}^{d+p-1}
         \, M_{\alpha} \left( {\left( \frac{x}{\sigma} \right)}^{p} \right) 
\end{align}

We apply the property \(M_0(z) = \exp(-z)\) mentioned below \eqref{eq:m-wright-series}
to \(\gscNx\) from \eqref{eq:pdf-gsc-mw}, and obtain:

\begin{align*}
\lim_{\alpha \to 0} \gscNx
    &= \lim_{\alpha \to 0}
        \alpha C 
        \, {\left( \frac{x}{\sigma} \right)}^{d+p-1}
        \, \exp\left( -{\left( \frac{x}{\sigma} \right)}^{p} \right)
\\
    &= \left( \lim_{\alpha \to 0} \alpha \Gamma\left( \frac{d}{p} \alpha \right)\right) 
        \frac{p/\sigma} {\Gamma(\frac{d}{p})} 
        \, {\left( \frac{x}{\sigma} \right)}^{d+p-1}
        \, \exp\left( -{\left( \frac{x}{\sigma} \right)}^{p} \right)
\\
    &= \frac{p / \sigma} { \Gamma(\frac{d+p}{p})} 
        \, {\left( \frac{x}{\sigma} \right)}^{d+p-1}
        \, \exp\left( -{\left( \frac{x}{\sigma} \right)}^{p} \right)
\\
    &= 
    \pdfGG(x; \sigma, d+p, p)
\end{align*}

Hence, we reach the mapping of GG to GSC at \(\alpha = 0\) 
described in the first line of \eqref{eq:main-gg-by-gsc}:

\begin{align}\label{eq:gg-by-gsc}
\pdfGG(x; \sigma, d, p)
    &= \scN_0(x; \sigma, d-p, p) 
\end{align}
\end{proof}

The GSC mapping of classic distributions at \(\alpha=0\) is illustrated in 
Table \ref{tab:classic-map-zero}.

\addtocounter{table}{-1}
\begin{table}[H]
    \smaller  
    \centering
    \renewcommand{\arraystretch}{1.5}
\begin{longtable}{|p{3.5cm}|p{3.0cm}|l|l|l|l|l|l|l|}
\hline
\multicolumn{1}{|c|}{} & \multicolumn{1}{|c|}{}
    & \multicolumn{3}{c|}{GG: \(\pdfGG(x; \sigma, d, p)\)} 
    & \multicolumn{4}{c|}{GSC: \(\scN_0(x; \sigma, d, p)\)} \\
\cline{3-9}
\multicolumn{1}{|c|}{Distribution (PDF)} 
    & \multicolumn{1}{|c|}{Classic Equiv.} 
    & \(\sigma\) & \(d\) & \(p\)
    & \(\alpha\) & \(\sigma\) & \(d\) & \(p\) \\
\hline
Stretched Exp:
\(\mathcal{E}_\alpha^{(1)}(x)\) 
    &   & 1 & 1 & \(\alpha\)
    & 0 & 1 & \(1-\alpha\) & \(\alpha\) \\
Half-Normal:   \(2 \, \mcN(x)\)
    &   & \(\sqrt{2}\) &  1 & 2 
    & 0 & \(\sqrt{2}\) & -1 & 2 \\
Weibull:  \(\text{Wb}(x;k)\) 
    &   & 1 & \(k\) & \(k\) 
    & 0 & 1 & 0 & \(k\) \\
Exponential: \(2 \, L(x)\) & \(\text{Wb}(x;k=1)\) 
        & 1 & 1 & 1 
    & 0 & 1 & 0 & 1 \\
\(\text{Rayleigh}(x)\) 
    & \(\text{Wb}(\frac{x}{\sqrt{2}};k=2)\)
        & \(\sqrt{2}\) & 2 & 2 
    & 0 & \(\sqrt{2}\) & 0 & 2 \\
Gamma: \(\Gamma(\frac{x}{\sigma}; s)\) 
    &   & \(\sigma\) & \(s\) & 1 
    & 0 & \(\sigma\) & \(s-1\) & 1 \\
\(\chi_k(x)\) & \(\Gamma(\frac{x^2}{2}; \frac{k}{2})\) 
        & \(\sqrt{2}\) & \(k\) & 2 
    & 0 & \(\sqrt{2}\) & \(k-2\) & 2 \\
\(\chi_k^2(x)\) & \(\Gamma(\frac{x}{2};\frac{k}{2})\) 
        & 2 & \(\frac{k}{2}\) & 1 
    & 0 & 2 & \(\frac{k}{2}-1\) & 1 \\
\(\text{GenGamma}(\frac{x}{\sigma}; s,c)\) 
    & \(\Gamma((\frac{x}{\sigma})^{c}; s)\)
        & \(\sigma\) & \(sc\) & \(c\) 
    & 0 & \(\sigma\) & \(c(s-1)\) & \(c\) \\
\,\, or \(\pdfGG(x; \sigma, d, p)\) 
    &   & \(\sigma\) & \(d\) & \(p\) 
    & 0 & \(\sigma\) & \(d-p\) & \(p\) \\
\({\text{IG}}(x;k)\)
    &   & 1 & \(-k\) & \(-1\) 
    & 0 & 1 & \(-k+1\) & \(-1\) \\
\({\text{IWb}}(x;k)\)
    &   & 1 & \(-k\) & \(-k\) 
    & 0 & 1 & 0 & \(-k\) 
\\
\hline
\end{longtable}
\caption{\label{tab:classic-map-zero}
Mapping classic distributions to \(\scN_0(x; \sigma, d, p)\).
\(\mathcal{E}_\alpha^{(1)}(x) = \exp(-{x^\alpha})/\GammaAlpha\) 
is a one-sided exponential power distribution.
}
\end{table}

\subsection{Mapping between GSC and GG at alpha = 1/2}
\label{section:mapping-gg-one-half}

We prove the second line of \eqref{eq:main-gg-by-gsc} as following.

\begin{proof}
From the property mentioned below \eqref{eq:m-wright-series} that
\(M_{\frac{1}{2}}(z) = \frac{2}{z} W_{-\frac{1}{2},0}(-z) = \frac{1}{\sqrt{\pi}} e^{-z^2/4}\),
the exponential function can be re-interpreted as 
\(e^{-z} = \sqrt{\pi} \, M_{\frac{1}{2}}(2\sqrt{z})\).

By a change of variable \(z = {\left( \frac{x}{a} \right)}^{p}\),
GG is mapped to GSC at \(\alpha = 1/2\) as

\begin{align*}
\pdfGG(x; a, d, p) 
    &= \frac{\sqrt{\pi} \, p/a}{\Gamma(\frac{d}{p})} {\left(\frac{x}{a}\right)} ^{d-1} 
     \, M_{\frac{1}{2}}\left( 2 {\left( \frac{x}{a} \right)}^{p/2} \right)
\\
    &= \frac{\sqrt{\pi} \, p/a}{\Gamma(\frac{d}{p})} {\left(\frac{x}{a}\right)} ^{d-1} 
     \, {\left( \frac{x}{a} \right)}^{-p/2} W_{-\frac{1}{2},0}\left( -2 {\left( \frac{x}{a} \right)}^{p/2} \right)
\\
    &= \frac{\sqrt{\pi} \, p/a}{\Gamma(\frac{d}{p})} {\left(\frac{x}{a}\right)} ^{d-p/2-1} 
     \, W_{-\frac{1}{2},0}\left( -{\left( \frac{2^{2/p} \, x}{a} \right)}^{p/2} \right)
\\
    &= \scN_{\frac{1}{2}}\left(x; \sigma=\frac{a}{2^{2/p}}, d=d-\frac{p}{2}, p=\frac{p}{2} \right)
\end{align*}

Hence, we reach the second line of \eqref{eq:main-gg-by-gsc}.
\end{proof}

If a classic distribution can be mapped to \(\scN_0(x;...)\), 
it can also be mapped to \(\scN_{\frac{1}{2}}(x;...)\).
That is how we map the \(\chi\) distribution from
\({\chi}_k(s) = \scN_0( s; \sigma=\sqrt{2}, d=k, p=2) \) to its fractional counterpart
\({\chi}_k(s) = \scN_{\frac{1}{2}}( s; \sigma=\frac{1}{\sqrt{2}}, d=k-1, p=1) \)
in \eqref{eq:chi-at-half-alpha} that eventually leads to \eqref{eq:main-chimean-pdf}.

The \(\alpha = 1/2\) mapping is illustrated in Table \ref{tab:classic-mapping-half-alpha}.

\addtocounter{table}{-1}
\begin{table}[H]
    \small  
    \centering
    \renewcommand{\arraystretch}{1.5}
\begin{longtable}{|p{3.5cm}|p{3.5cm}|l|l|l|l|l|}
\hline
\multicolumn{1}{|c|}{} & \multicolumn{1}{|c|}{}
    & \multicolumn{4}{c|}{GSC: \(\gscNx\)} \\
\cline{3-6}
\multicolumn{1}{|c|}{Distribution (PDF)} 
    & \multicolumn{1}{|c|}{Classic Equiv.} 
    & \(\alpha\) & \(\sigma\) & \(d\) & \(p\) \\
\hline
Weibull:  \(\text{Wb}(x;k)\) & 
    & \(\frac{1}{2}\) & \(2^{-2/k}\) & \(\frac{k}{2}\) & \(\frac{k}{2}\) \\
Gamma: \(\Gamma(\frac{x}{\sigma}; s)\) & 
    & \(\frac{1}{2}\) & \(\sigma/4\) & \(s-\frac{1}{2}\) & \(\frac{1}{2}\) \\
\(\chi_k(x)\) & \(\Gamma(\frac{x^2}{2}; \frac{k}{2})\) 
    & \(\frac{1}{2}\) & \(\frac{1}{\sqrt{2}}\) & \(k-1\) & 1 \\
\(\chi_k^2(x)\) & \(\Gamma(\frac{x}{2};\frac{k}{2})\) 
    & \(\frac{1}{2}\) & \(\frac{1}{2}\) & \((k-1)/2\) & \(\frac{1}{2}\) \\
\(\text{GenGamma}(\frac{x}{\sigma}; s,c)\) 
    & \(\Gamma((\frac{x}{\sigma})^{c}; s)\)
    & \(\frac{1}{2}\) & \(2^{-2/c} \sigma\) & \((s-\frac{1}{2})c\) & \(c/2\) \\
\,\, or \(\pdfGG(x; \sigma, d, p)\) 
    & & \(\frac{1}{2}\) & \(2^{-2/p} \sigma\) & \(d-\frac{p}{2}\) & \(\frac{p}{2}\) \\
\({\text{IG}}(x;k)\)
    & & \(\frac{1}{2}\) & 4 & \(\frac{1}{2}-k\) & \(-\frac{1}{2}\) \\
\({\text{IWb}}(x;k)\)
    & & \(\frac{1}{2}\) & \(2^{2/k}\) & \(-k/2\) & \(-k/2\)
\\
\hline
\end{longtable}
\caption{\label{tab:classic-mapping-half-alpha}Mapping of some classic distributions 
to \(\scN_{\frac{1}{2}}(x; \sigma, d, p)\)}
\end{table}

\subsection{GSC Hankel Integral}
\label{section:gsc-hankel}

The Hankel integral establishes another connection between GSC and GG.
We begin with the Hankel integral of \(F_\alpha(z)\) in \eqref{eq:m-wright-F-hankel}
where \(F_\alpha(z) = W_{-\alpha,0}(-z) =  \frac{1}{2\pi i} \int_{H} dt \, \exp{(t-z t^\alpha)}\).
Hence, for \(\alpha d/p \ne 0\), the Hankel integral of GSC is

\begin{align}
\gscNx &= \notag
    \frac{1}{2\pi i} \int_{H} dt 
    \, C \, {\left( \frac{x}{\sigma} \right)}^{d-1} 
    \, \exp{\left( t- {\left( \frac{x}{\sigma} \right)}^{p} t^\alpha \right)}
\\ &= 
    \frac{1}{2\pi i} \int_{H} dt \, e^t 
    \, C \, {\left( \frac{x}{\sigma} \right)}^{d-1} 
    \, \exp{\left( - {\left( \frac{x \, t^{\alpha/p}}{\sigma } \right)}^{p}  \right)}
\notag 
\\ \label{eq:fcm-hankel-intg}
    &= 
    \,  \Gamma\left(\frac{\alpha d}{p}\right)
    \, 
    \frac{1}{2\pi i} \int_{H} dt 
    \, \left( \frac{e^t }{t^{\alpha d/p}} \right) 
    \, \pdfGG\left( x; a=\frac{\sigma }{ t^{\alpha/p}}, d=d, p=p \right)
\end{align}

This shows that GSC can be obtained from the Hankel integral over its GG counterpart.

%% file: sub-gsc-mnt-sc-sv.tex
\subsection{GSC Moments}
\label{section:gsc-moments}

GSC's moment formula plays a foundational role since FCM and GSaS's moment formulas are 
derived from it. 

\begin{lemma}\label{lemma-gsc-moment}
The \(n\)-th moment of GSC is
\begin{align}\label{eq:gsc-moment}
\E(X^n|\gscNx) &= \sigma^{n} \,
    \frac{\Gamma(\frac{d}{p} \alpha)} {\Gamma(\frac{d}{p})} 
    \frac{\Gamma(\frac{n+d}{p})} {\Gamma(\frac{n+d}{p} \alpha)}
    , & \, \, \text{for} \,\, \alpha \ne 0, d \ne 0.
\end{align}

\end{lemma}

\begin{proof}
The integral of the \(n\)-th moment is
\begin{align*}
\E(X^n|\gscNx) 
    &= \int_{0}^{\infty} dx \,
        x^n \, C \, 
        {\left( \frac{x}{\sigma} \right)}^{d-1}
        \, \FWrfn{\alpha}{{\left( \frac{x}{\sigma} \right)}^{p}}
\\
    &= \sigma^{n+1} \, \frac{C}{p}
    \int_{0}^{\infty} dz \,
        \, {z}^{\frac{n+d}{p}-1}
        \, \FWrfn{\alpha}{z} 
    , \quad \where z = {\left( \frac{x}{\sigma} \right)}^{p}
\end{align*}

This is basically the \(\left(\frac{n+d}{p}\right)\)-th moment of 
the Wright function in \eqref{eq:wright-moments}.
We arrive at
\begin{align*}
\E(X^n|\gscNx) 
    &= \sigma^{n+1} \, \frac{C}{p} 
    \frac{\Gamma(\frac{n+d}{p})} {\Gamma(\frac{n+d}{p} \alpha)}
\\ 
    &= \sigma^{n} \,
    \frac{\Gamma(\frac{d}{p} \alpha)} {\Gamma(\frac{d}{p})} 
    \frac{\Gamma(\frac{n+d}{p})} {\Gamma(\frac{n+d}{p} \alpha)}
    , & \,\, \text{for} \,\, \alpha \ne 0, d \ne 0.
\end{align*}

\end{proof}

Although the formula exists, whether a specific moment exists can be complicated
by the choice of \(n\) relative to the shape parameters \(\alpha, d, p\).

We caution a common scenario in this paper that 
such formula needs special handling due to \(\Gamma(z) \approx 1/z\) when \(z \to 0\). 
For instance, in this case, \(d = 0\) or \(n = -d\).
This is documented in \eqref{eq:gamma-zero}. 
There are numerous situations like these, we may not mention them repeatedly in subsequent encounters.

As a validation, when \(\alpha \to 0\), it becomes the moments of 
\(\pdfGG(x; \sigma, d+p, p)\), as expected from Section \ref{section:mapping-gg-0}:
\begin{align*}
\E(X^n|\gscNx)
    &= \sigma^{n} \,  
    \frac{\Gamma(\frac{n+d+p}{p})} {\Gamma(\frac{d+p}{p})}
\end{align*}


When \(\alpha \to 1\), all gamma terms cancel out in \(\E(X^n)\),
leaving it to \(\E(X^n) = \sigma^n\). Therefore, the mean is \(\sigma\),
and the variance is 0, which indicates it is a Dirac delta function \(\delta(x-\sigma)\).

\subsection{Stable Count and Stable Vol Distributions}
\label{section:sc-sv}

GSC traces its root to the discovery of
the stable count distribution in 2017\cite{Lihn:2017} and the stable vol 
distribution in 2020\cite{Lihn:2020}. It is important to reiterate the following relations.

The PDF of the original stable count distribution (SC) is defined as 
\begin{equation}\label{eq:sc}
\scN_\alpha(\nu) := \frac{1}{\GammaAlpha} 
    \frac{1}{\nu} {L}_{\alpha} \left( \frac{1}{\nu} \right)
    = \frac{1}{\GammaAlpha} 
    \FWrfn{\alpha}{{\nu}^{\alpha}}
\end{equation}
where \({L}_{\alpha}(x)\) is the PDF of the one-sided stable distribution.
\(L_\alpha(x) = x^{-1} W_{-\alpha,0}(-x^{-\alpha})\).
SC is subsumed to GSC as \(\scN_\alpha(\nu; \sigma=1,d=1,p=\alpha)\).

The stable vol distribution (SV) is a variant of SC where 
\begin{equation}\label{eq:sv}
V_\alpha(s) := \frac{\sqrt{2\pi} \GammaAlphaN{2}}{\GammaAlpha} 
   \scN_{\frac{\alpha}{2}}(2s^2)
    = \frac{\sqrt{2\pi}}{\GammaAlpha} 
    \FWrfn{\alphahalf}{{(\sqrt{2}s)}^{\alpha}}.
\end{equation}
It is subsumed to GSC as \(\scN_\alphahalf(s; \sigma=1/\sqrt{2},d=1,p=\alpha)\).

The original idea of having two variants, SC and SV, came from 
the so-called "Lambda decomposition": transforming \(e^{-{|z|}^\alpha}\) 
as a product distribution with
the Laplace distribution \(L(x)\) and the normal distribution \(\mcN(x)\):

\begin{align}\label{eq:lambda-decomp}
\frac{1}{2 \GammaAlpha} e^{-{|z|}^\alpha} &= 
        \left\{
        \begin{array}{ll} \displaystyle
            \int_0^\infty \frac{d\nu}{\nu} 
                L\left( \frac{|z|}{\nu} \right) \, \scN_\alpha(\nu)
            \, &, \, \text{for} \,\, 0 < \alpha \le 1.
            \\ \displaystyle
            \int_0^\infty \frac{ds}{s} \mcN\left( \frac{|z|}{s} \right) \, V_\alpha(s)
            \, &, \, \text{for} \,\, 0 < \alpha \le 2.
        \end{array}
        \right.
\end{align}

It turns out the second line in \eqref{eq:lambda-decomp} is very useful.
First, the normal distribution represents the Brownian noise. 
This is the early form of the Gaussian mixture presented in Section \ref{section:mixture}.

Secondly, in the fractional \(\chi\) notation,
\(V_\alpha(x) = \chi_{\alpha, 2}(x)\), which has two degrees of freedom.
This laid the groundwork for FCM. 

Thirdly, it is in fact decomposing the CF of the \(\alpha\)-stable law
into a Gaussian mixture.
It is used in the proof of Section \ref{section:skew-kernel} that follows.

In closing, several known fractional distributions are expressed in GSC
in Table \ref{tab:frac-dist-mapping}.

\addtocounter{table}{-1}
\begin{table}[H]
    \small  
    \centering
    \renewcommand{\arraystretch}{1.5}
\begin{longtable}{|p{4.0cm}|p{3.5cm}|l|l|l|l|l|}
\hline
\multicolumn{1}{|c|}{} & \multicolumn{1}{|c|}{}
    & \multicolumn{4}{c|}{GSC: \(\gscNx\)} \\
\cline{3-6}
\multicolumn{1}{|c|}{Distribution (PDF)} 
    & \multicolumn{1}{|c|}{Wright Equiv.} 
    & \(\alpha\) & \(\sigma\) & \(d\) & \(p\) \\
\hline
One-sided stable: \(L_\alpha(x)\) & \(x^{-1} W_{-\alpha,0}(-x^{-\alpha})\) 
    & \(\alpha\) & 1 & 0 & \(-\alpha\) \\
SC: \(\scN_\alpha(\nu)\)          & & \(\alpha\) & 1 & 1 & \(\alpha\) \\
SV: \(V_\alpha(s) \)              & & \(\frac{\alpha}{2}\) & \(\frac{1}{\sqrt{2}}\) & 1 & \(\alpha\) \\
M-Wright: \(M_\alpha(z)\) & \(\frac{1}{\alpha z} W_{-\alpha,0}(-z)\) 
    & \(\alpha\) & 1 & 0 & 1 
\\
M-Wright: \(\Gamma(\alpha) F_\alpha(z)\) & \(\Gamma(\alpha) W_{-\alpha,0}(-z)\) 
    & \(\alpha\) & 1 & 1 & 1 
\\
\hline
\end{longtable}
\caption{\label{tab:frac-dist-mapping}
    GSC mapping of several known fractional distributions in the literature.
    }
\end{table}

%% file: section-skew-kernel.tex
\section{The Skew-Gaussian Kernel and Fractional Chi-1}
\label{section:skew-kernel}

\subsection{Derivation of the Skew-Gaussian Kernel}
\label{section:derive-skew-kernel}

We prove how the kernel \(\gskew\) in \eqref{eq:main-skew-gaussian} 
comes into existence as following.

\begin{proof}
The minus-log of \(\alpha\)-stable CF in Feller's parameterization is defined as
(Section XVII.6, p.581-583 of \cite{Feller:1971}):

\begin{align*}
\psi_\alpha^\theta(\zeta) 
    &= \exp\left( \text{sgn}(\zeta) \frac{i\theta\pi}{2} \right) {|\zeta|}^\alpha
\\
    &= \left( \cos(\Theta) + i \sin(\Theta) \right) {|\zeta|}^\alpha
    , \quad \where \Theta = \text{sgn}(\zeta) \frac{\theta\pi}{2}
\\
    & \quad \text{and} \,\, 
        0 \lt \alpha \le 2, |\theta| \le \min\{\alpha, 2-\alpha\}.
\end{align*}

Put it into the integral of the \(\alpha\)-stable PDF (See (F.35) of \cite{Mainardi:2010}), we get

\begin{align*}
\afstable{x}
    &= \frac{1}{2\pi} \int_{-\infty}^{\infty} d\zeta  
    \, \exp(-\psi_\alpha^\theta(\zeta)) \, e^{-ix\zeta}
\\
    &= \frac{1}{2\pi} \int_{-\infty}^{\infty} d\zeta  \, e^{-\cos(\Theta) {|\zeta|}^\alpha} 
    \, e^{-i (B + x\zeta)}
    , \quad \where B = \sin(\Theta) {|\zeta|}^\alpha
\end{align*}

The imaginary part of the integral cancels out due to the sign of \(\zeta\).
The real part of the integral can be simplified to 
the integral on the positive axis \((\zeta \ge 0)\).
We can remove \(\text{sgn}(\zeta)\) from \(\Theta\). Hence,

\begin{align*}
\afstable{x}
    &= \frac{1}{\pi} \int_{0}^{\infty} d\zeta  \, e^{-\cos(\Theta) {\zeta}^\alpha} 
    \, \cos(B + x\zeta) 
    , \quad \where \Theta = \frac{\theta\pi}{2}, B = \sin(\Theta) {\zeta}^\alpha
\end{align*}

Next, we apply the Lambda decomposition \eqref{eq:lambda-decomp} for the stable-vol \(V_\alpha(\zeta)\)
on the \(e^{-\cos(\Theta) {\zeta}^\alpha}\) term. 
Let \(u = \cos(\Theta)^{1/\alpha} \zeta\), then 
\(e^{-\cos(\Theta) {\zeta}^\alpha} = e^{-{u}^\alpha}\). We have

\begin{align}
\afstable{x} \notag
    &= \int_{0}^{\infty} d\zeta  
    \, \cos(B + x\zeta)
      \, \int_{0}^{\infty} \frac{ds}{s} 
       \left( \frac{1}{\sqrt{2\pi}} e^{-(u/s)^2/2} \right)
       \left( \frac{2 \Gamma(\frac{1}{\alpha}+1)}{\pi}  V_\alpha(s) \right)
\\ \label{eq:astable-pdf-mid}
    &= \int_{0}^{\infty} s \, ds \,
       \left( \frac{1}{\pi s} \int_{0}^{\infty} d\zeta  
            \, \cos(B + x\zeta)  e^{-(u/s)^2/2} \right)
       \left( \frac{2 \Gamma(\frac{1}{\alpha}+1)}{\sqrt{2\pi}} s^{-1} V_\alpha(s) \right)
\end{align}

The terms inside the first big parenthesis of \eqref{eq:astable-pdf-mid} 
is the skew-Gaussian kernel in its raw format:

\begin{align*}
\gskew
    &= \frac{1}{\pi s} \int_{0}^{\infty} d\zeta  
            \, \cos(B + x\zeta)  e^{-(u/s)^2/2}
    , \quad \where u = \cos(\Theta)^{1/\alpha} \zeta
\end{align*}
Make one more change of variable: \(t = u/s = \cos(\Theta)^{1/\alpha} \zeta/s\),
and \(B = \tan(\Theta) {(st)}^\alpha\),
we arrive at the final form presented in \eqref{eq:main-skew-gaussian} as

\begin{align*}
\gskew
    &= \frac{1}{q \pi} \int_{0}^{\infty} dt
        \, \cos\left(\tau\, {(st)}^\alpha + \frac{x}{q}\, st \right)  e^{-t^2/2} 
        \quad (s \ge 0)
\\  & \where \notag
        q = \cos(\theta\pi/2)^{1/\alpha},
        \,\, \tau = \tan(\theta\pi/2)
\end{align*}
\end{proof}

Next, we prove that \(\afstable{x}\) is indeed 
\(\int_{0}^{\infty} s \, ds \, \gskew \, \chibar_{\alpha,1}(s)\).

\begin{proof}
The terms inside the second big parenthesis of \eqref{eq:astable-pdf-mid} is simplified by 
 \eqref{eq:sv}, 
\begin{align} \label{eq:stable-decomp-fcm}
\left( \frac{2 \Gamma(\frac{1}{\alpha}+1)}{\sqrt{2\pi}} s^{-1} V_\alpha(s) \right)
    &=  2 s^{-1} \FWrfn{\alphahalf}{{(\sqrt{2}s)}^{\alpha}}
\\ \notag
    &= \scN_{\frac{\alpha}{2}}(s; \sigma=\frac{1}{\sqrt{2}}, d=0, p=\alpha)
\\ \notag
    &= \chibar_{\alpha,1}(s)
\end{align}

So according to \eqref{eq:stable-decomp-fcm} and \eqref{eq:stable-chi1}, 
the \(\alpha\)-stable PDF \eqref{eq:astable-pdf-mid} becomes a ratio distribution 
of the kernel and an FCM of one degree of freedom:

\begin{align}\label{eq:afstable-decomp-final}
\afstable{x} 
    &= \int_{0}^{\infty} s \, ds \,
       g_{\alpha}^{\theta}(x,s)
       \, \chibar_{\alpha,1}(s)
\end{align}
\end{proof}

\subsection{Some Properties of the Skew Kernel}
\label{section:properties-skew-kernel}

Some properties of \(\gskew\) are analyzed in this section. 
This will help us better understand its behavior intuitively, as well as 
provide a better implementation.

First, the integral inside \(\gskew\) is bounded by the Gaussian integral 
\(\int_{0}^{\infty} dt \, e^{-t^2/2} \), where \(e^{-t^2/2}\) is the upper bound 
to the absolute of the integrand since \(|\cos(z)| \le 1\). Therefore,
the contribution of the integrand decreases rapidly as \(t\) gets larger.
This is the main reason such representation is chosen.
Implementation-wise, it is better to cut off the integral at a certain large \(t\) (e.g. 100).

Indeed, when \(\theta=0\), the Gaussian integral becomes exact, that is,
\(g_{\alpha}^0(x,s) = \mcN(xs)\); and \(L_\alpha^0(x)\) becomes symmetric.
\begin{proof}
When \(\theta=0\), \(q=1\) and \(\tau=0\). Therefore,
\begin{align*}
g_{\alpha}^0(x,s)
    &= \frac{1}{\pi} \int_{0}^{\infty} dt
        \, \cos\left(xst \right)  e^{-t^2/2} 
\end{align*}
From \eqref{eq:gauss-cosine}, this is \(\frac{1}{\sqrt{2\pi}} e^{-(xs)^2/2} = \mcN(xs)\).
\end{proof}

When \(\theta \ne 0\), some parts of \(\gskew\) can be negative.
Hence, it is not a distribution per se.

When \(\alpha=1\), we have 
\(g_{1}^\theta(x,s) = \frac{1}{q} \,\mcN\left(\left(\tau + \frac{x}{q}\right) s\right)\), 
and \(L_{1}^\theta(x)\) becomes symmetric over \(x_0 = -q \tau = -\sin(\theta \pi / 2)\).
That is, the effect of \(\theta\) is no longer skewness, but a shift in location.
\begin{proof}
When \(\alpha=1\), the argument of \(\cos()\) in \eqref{eq:main-skew-gaussian} 
becomes \(\tau\, {(st)} + \frac{x}{q}\, st = Y t\) 
where \(Y = \left(\tau + \frac{x}{q} \right) s\). 
Therefore, from \eqref{eq:gauss-cosine},
\begin{align*}
g_{1}^\theta(x,s)
    &= \frac{1}{q \pi} \int_{0}^{\infty} dt
        \, \cos\left(Y t \right)  e^{-t^2/2} 
\\
    &= \frac{1}{q} \mcN(Y) 
    = \frac{1}{q}\,\mcN\left(\left(\tau + \frac{x}{q}\right) s\right)
\end{align*}
\end{proof}

At \(s = 0\), \(g_{\alpha}^{\theta}(x,0)\) is simply a constant \(1/\sqrt{2\pi} q\).
This is straightforward since \(g_{\alpha}^{\theta}(x,0) = 
\frac{1}{q \pi} \int_{0}^{\infty} dt\,  e^{-t^2/2}\). Furthermore, for a given \((\alpha,\theta)\) pair, 
\(|\gskew| \le 1/\sqrt{2\pi} q\) for all \(x\)'s and \(s\)'s.

Next we prove:
\begin{lemma}
For a given \(s > 0\), \( \int_{-\infty}^{\infty} g_{\alpha}^{\theta}(x,s) \, dx = 1/s\).
\end{lemma}
\begin{proof}
We integrate the \(\alpha\)-stable PDF \eqref{eq:afstable-decomp-final} 
by the full range of \(x\), and it should be one:

\begin{align*}
\int_{-\infty}^{\infty} \afstable{x} \, dx
    &=
    \int_{-\infty}^{\infty} \, dx
    \int_{0}^{\infty} s \, ds \,
       g_{\alpha}^{\theta}(x,s)
       \, \chibar_{\alpha,1}(s) = 1
\end{align*}
Rearranging the double integrals, it becomes
\begin{align*}
\int_{0}^{\infty} s \, ds \,
    \left(
        \int_{-\infty}^{\infty} 
            g_{\alpha}^{\theta}(x,s) \, dx
    \right) 
    \, \chibar_{\alpha,1}(s) 
    &= 1
\end{align*}

Since \(\chibar_{\alpha,1}(s)\) is a PDF, we have \(
\int_{0}^{\infty} \chibar_{\alpha,1}(s) \, ds = 1\). 
This forces the integral inside the big parenthesis to be \(1/s\), which is
\( \int_{-\infty}^{\infty} g_{\alpha}^{\theta}(x,s) \, dx = 1/s\).

\end{proof}

%% file: section-fcm.tex
\section{FCM: Fractional Chi-Mean Distribution}
\label{section:fcm}

FCM is the core of the innovation in this paper.
We elaborate the thought process that led to FCM in this section. 
There is a theoretical side of thinking. There is also the aspect of design,
mainly on how to choose the scale of standardization such that it is consistent at both 
\(k=1\) and \(k \to \infty\) limits.

\subsection{Derivation of Fractional Chi}
\label{section:frac-chi}

First of all, after many experiments (elaborated here 
and also the subordination results in Appendix \ref{section:subordination}), 
I've come to realization that we need to use 
a ratio distribution \eqref{eq:ratio-dist},
instead of a product distribution \eqref{eq:product-dist}, to describe
the inner structures of both Student's t and \(\alpha\)-stable. 
This is a paradigm change.

Secondly, after many analyses, I've settled down on using a normal distribution, 
instead of other alternatives, such as Laplace or Cauchy, as the unit distribution in the middle. 
This has been formalized into the continuous Gaussian mixture in Section \ref{section:mixture}.

The advantage of a normal distribution is obvious from the perspective of stochastic calculus. 
Whatever distributional equation we come up with can be translated to the random variable generation directly.
Symbolically, \(L(x) = \int_{0}^{\infty} s ds \mcN(sx) S(s)\) in the former
becomes \(dL_t = dW_t / S_t\) in the later, where \(dW_t\) is the Brownian noise.
Section \ref{section-gsc-random} is dedicated to the subject of 
random variable generation, that is, how \(S_t\) can be generated.

Thirdly, for the expression of the normal distribution, we've evaluated two options:
(a) use the notion of variance (which we like to use \(\nu\) for it in \eqref{eq:lambda-decomp});
or (b) use the notion of volatility (that is, squared root of variance, 
which we like to use \(s\) for it in \eqref{eq:lambda-decomp}).
After many thoughts, we didn't chose the variance route. 

We wouldn't say the last point is totally set in stone. 
Searches on "chi-squared distribution" vs "chi distribution" showed that 
people like to talk about sum of square of variables, 
instead of squared root of that sum. 
But we find that our formula is better presented by the later, instead of the former.

During the following discussions, we assume \(k > 0\). 
\(k\) has the ordinary physical meaning as the \emph{degrees of freedom}.

Now let's begin with an expression of Student's t distribution \(t_k\) as the ratio of a normal and 
the root-mean of a chi-squared: \(t_k \sim \mcN / \sqrt{\chi^2_{k} / k}\) 
as described in \cite{JoramSoch:2019t} and \cite{Tong:1990}.

We prefer to combine the root-mean-square together. Therefore, we define the \(\chi\)-mean 
distribution as \(\chibar_k = \sqrt{\chi^2_{k} / k} = \chi_{k} / \sqrt{k}\), 
which turns out to be a GSC:
\begin{align}\label{eq:chi-at-half-alpha}
\chibar_k(s) 
    &:= \scN_{\frac{1}{2}}\left( s; \sigma=\frac{1}{\sqrt{2k}}, d=k-1, p=1 \right) 
\end{align} 

Next, we consider the more complicated matter of \(\alpha\)-stable. 
In Section \ref{section:skew-kernel}, after a lengthy derivation, we arrive at 
\eqref{eq:stable-decomp-fcm}, where we find ourselves dealing with a similar GSC:
\begin{align}\label{eq:stable-chi1}
\chibar_{\alpha,1}(s) 
    &:= \scN_{\frac{\alpha}{2}}(s; \sigma=\frac{1}{\sqrt{2}}, d=0, p=\alpha)
\end{align}

Here \(\chibar_{\alpha,1}\) has a physical interpretation: It is the \(\alpha\)-stable's
\(\chi\) distribution with \textbf{one} degree of freedom. 
This interpretation will be used in \eqref{eq:fcm-prod-assign} below.

By comparing the GSC parameters between \eqref{eq:chi-at-half-alpha} and \eqref{eq:stable-chi1},
it becomes obvious that it is possible to merge \(\chibar_k(s)\) and \(\chibar_{\alpha,1}(s)\)
carefully and come up with a super distribution family nicknamed FCM.
The result is presented in \eqref{eq:frac-chi-mean}.
But it takes quite some reasonings to get there...

\subsection{Fractional Chi-1}
\label{section:frac-chi-1}

In Lemma \ref{lemma-frac-chi1}, we mentioned that \(\chibar_{\alpha,1}\) 
is the fractional version of classic \(\chi_1\) and it is the building block of SaS
since \(P_\alpha = L_\alpha^0 \sim \mcN / \chibar_{\alpha,1}\), or in a distributional form:
\begin{align}
P_\alpha(x) := L_\alpha^0(x) \notag
    &= \int_{0}^{\infty} s \, ds \,
       \mcN(xs) \, \chibar_{\alpha,1}(s)
\end{align}

Compare it to \(t_k \sim \mcN / \chi_k\). We can re-interpret \(P_\alpha\)
as a sort of t-statistics in the fractional space. And \(\chibar_{\alpha,1}\)
is the fractional \emph{chi of one degree of freedom}\footnote{
Wikipedia has a paragraph for this subject: 
\url{https://en.wikipedia.org/wiki/Proofs_related_to_chi-squared_distribution\#Derivation_of_the_pdf_for_one_degree_of_freedom}}.

The PDF of \(\chibar_{\alpha,1}\) has some peculiar behaviors 
(listed under Lemma \ref{lemma-frac-chi1}) 
that span from a Dirac delta function, to a half-normal distribution, to \(1/x\) type. 
This affects the behavior of SaS directly.
We prove them here.

First, we show \(\chibar_{2,1}(x)\) is a Dirac delta function. 
From FCM's moment formula \eqref{eq:fcm-moment},
\(\E(X^n|\chibar_{2,1}) = 2^{-n/2}\). Its first moment is \(\frac{1}{\sqrt{2}}\). 
Its second moment is \(\frac{1}{2}\). Its variance is zero. 
So it is a Dirac delta function at \(\frac{1}{\sqrt{2}}\), 
aka \(\chibar_{2,1}(x) = \delta\left(x - 1/\sqrt{2} \right)\).

Secondly, we show \(\chibar_{1,1}(x)\) is a half-normal distribution. From \eqref{eq:frac-chi-mean},
its PDF is \(2 x^{-1} W_{-\frac{1}{2},0}(-\sqrt{2} x)\). From Section \ref{section:mapping-gg-one-half},
it becomes \(\sqrt{2/\pi} \, e^{-x^2/2}\).

Thirdly, we prove that \(\chibar_{\alpha,1}(x) \approx \alpha/(ex)\), 
and subsequently, \(P_\alpha(x) \approx \chibar_{\alpha,1}(x)/2 \approx \alpha/(2ex)\), 
when \(\alpha\) is very small but \(x\) is not too small.

\begin{proof}
We use the result of GSC at \(\alpha \to 0\) in Section \ref{section:mapping-gg-0}.
Combine \eqref{eq:stable-chi1} with \eqref{eq:gg-by-gsc}, we obtain
\begin{align*}
    \chibar_{\alpha,1}(x) 
    &\to \pdfGG(x; \sigma=\frac{1}{\sqrt{2}}, d=\alpha, p=\alpha)
\\
    &= \frac{\alpha}{x} {(\sqrt{2} x)}^\alpha \exp\left( -{(\sqrt{2} x)}^\alpha \right)
\end{align*}
The condition that \(\alpha\) is very small but \(x\) is not too small means that 
\({(\sqrt{2} x)}^\alpha \approx 1\). That is, \(\alpha |\log(x)| \to 0\).
Under such condition, \(\chibar_{\alpha,1}(x) \approx \alpha/(ex)\).

For \(P_\alpha(x)\) at \(\alpha \to 0\), we have
\begin{align*}
    P_\alpha(x)
    &= \int_{0}^{\infty}  s ds \, \mcN(sx) \, \bar{\chi}_{\alpha,1}(s)
\\
    &= (\alpha / e) \, \int_{0}^{\infty} ds \, \mcN(sx)
\\
    &= \alpha / (2ex) 
\end{align*}

\end{proof}

This example illustrates it is increasingly difficult to calculate \(P_\alpha(x)\) 
through integrating \(s \,\chibar_{\alpha,1}(s)\) numerically
when \(\alpha\) is small (e.g. \(\alpha \lt 0.3\)) and \(x \to 0\).

\subsection{Physical Interpretation of Fractional Chi}
\label{section:frac-chi-physical}

To gain physical understanding on \(\chibar_{\alpha,k}\), let's repeat how the
\(\chi^2\) distribution is derived classically, as described in \cite{JoramSoch:2019chi2},
but adapt it to the fractional interpretation needed for our purpose.

A fractionally generalized chi-distributed random variable \(R\) 
with \(k\) degrees of freedom is defined 
as the squared root of the sum of \(k\) squared random variables of the \(S_\alpha\) distribution,
which shall be clarified later that  \(S_\alpha\) is \(\chibar_{\alpha,1}\). So we have
\begin{align*}
X_1, ..., X_k \sim S_\alpha 
    &\implies
    R := \sqrt{\sum_{i=1}^{k} X_i^2} 
\end{align*}

Let \(x_1, ..., x_k\) be the values of \(X_1, ..., X_k\) and consider \((x_1, ..., x_k)\)
to be a point in the \(k\)-dimensional space. Define the distance to that point: 
\begin{align*}
    r = \sqrt{\sum_{i=1}^{k} x_i^2}
\end{align*}

Let \(f_R(r)\) and \(F_R(r)\) be the PDF and CDF of \(R\). Their relation is
\begin{align*}
    F_R(r+dr) - F_R(r) &= f_R(r) dr
\\
    &= \int_V \prod_{i=1}^{k} S_\alpha(x_i) dx_i
\end{align*}
where \(S_\alpha(x_i)\) is the PDF of the \(S_\alpha\) distribution,
and \(V\) is the volume on the \((k-1)\)-dimensional surface in the \(k\)-space 
at the radius \(r\).

Assume we can refactor the above integral such that 
\begin{align*}
    f_R(r) dr
    &=  
    \left( \int_V dx_1 ... dx_k \right)
    \left( \prod_{i=1}^{k} S_\alpha(x_i) \right)
\end{align*}
where we deal with the two big parentheses separately.

The first part is an integral of the surface area of the \((k-1)\)-dimensional sphere
with radius \(r\), which is
\begin{align}\label{eq:fcm-sphere-area}
    A(r) dr
    &= \left( \int_V dx_1 ... dx_k \right) = 2 \,r^{k-1} \frac{\pi^{k/2}}{\Gamma(k/2)} \, dr
\end{align}

For the second part, we argue it should just be \(\chibar_{\alpha,1}\) in \eqref{eq:stable-chi1}:
\begin{align}\label{eq:fcm-prod-assign}
    \left( \prod_{i=1}^{k} S_\alpha(x_i) \right) 
    &\sim \chibar_{\alpha,1}(r)
    = \frac{2}{r} \, W_{-\frac{\alpha}{2}, 0} \left( -{(\sqrt{2} r )}^\alpha \right) 
\end{align}

Combining \eqref{eq:fcm-sphere-area} and \eqref{eq:fcm-prod-assign}, 
we arrive at the PDF for \(R\) (apart from a normalization constant):
\begin{align*}
    f_R(r)
    &\propto r^{k-2} \, W_{-\frac{\alpha}{2}, 0} \left( -{(\sqrt{2} r )}^\alpha \right) 
\\
    &= \scN_{\frac{\alpha}{2}}(r; \sigma=\frac{1}{\sqrt{2}}, d=k-1, p=\alpha)
\end{align*}
Therefore, \(f_R(r)\) is the fractional-chi \(\chi_{\alpha,k}(r)\) that we seek.

Note that we glossed over the constant term.
We picked up the \(r^{k-1}\) factor from \(A(r)\) but ignored the rest of it.
This "surface area of the \((k-1)\)-dimensional sphere" is more complicated than 
a simple Euclidean sphere.
This topic is left for future research.

\subsection{From Fractional-Chi to Fractional Chi-Mean}
\label{section:frac-chi-to-fcm}

Now we have \(\chi_{\alpha,k}(x) = \scN_{\frac{\alpha}{2}}(x; \sigma=\frac{1}{\sqrt{2}}, d=k-1, p=\alpha)\).
We need to "standardize" its mean, so to speak. 
Our cue is that \(\chibar_k = \chi_{k} / \sqrt{k}\) works for Student's t. 
A simple multiplier of \(1/\sqrt{k}\) for \(\sigma\) suffices.
The question is: What works for \(\chi_{\alpha,k}\) in general?

We examine the mean of \(\chi_{\alpha,k}\). From \eqref{eq:gsc-moment}, its mean is
\begin{align*}
\E(X|\chi_{\alpha,k}) &= 
    \frac{1}{\sqrt{2}} \,
    \frac{\Gamma(\frac{k-1}{2})} {\Gamma(\frac{k-1}{\alpha})} 
    \frac{\Gamma(\frac{k}{\alpha})} {\Gamma(\frac{k}{2})}
\end{align*}
At \(k=1\), \(\E(X|\chi_{\alpha,1}) = \sqrt{\frac{2}{\pi}} \,\Gamma\left(\frac{1}{\alpha}+1\right)\),
which we want to preserve for \(\alpha\)-stable.

Its large \(k\) asymptotic is 
\begin{align}
\E(X|\chi_{\alpha,k}) &\sim 
    k^{1/\alpha-1/2} \,\alpha^{-1/\alpha} ,
    \quad k \gg 1;
\end{align}
whose \(k\) dependency we would like to eliminate so that the mean doesn't go to 0 or infinity.

So our requirement for FCM is to modify the \(\sigma\) parameter in \(\chi_{\alpha,k}\), 
according to the \(k\) component of \(\chi_{\alpha,k}\)'s asymptotic mean,
such that FCM has an asymptotic mean independent of \(k\).

The solution is the following. Define
\begin{align} \label{eq:fcm-sigma}
\sigma_{\alpha,k} &:= 
    \frac{k^{1/2 - 1/\alpha}}{\sqrt{2}}  
\end{align}
which has the desired properties: \(\sigma_{1,k} = 1/\sqrt{2k}\), 
and \(\sigma_{\alpha,1} = 1/\sqrt{2}\), in order to preserve both Student's t
and \(\alpha\)-stable distributions.

Hence, FCM is constructed as
\begin{align} \label{eq:frac-chi-mean}
\chimeanDist(x) 
    &:= \scN_{\frac{\alpha}{2}}(x; \sigma=\sigma_{\alpha,k}, d=k-1, p=\alpha)
\\ \notag
    &= \scN_{\frac{\alpha}{2}}(x; \sigma=\frac{1}{\sqrt{2}} \left( k^{1/2 - 1/\alpha} \right) , d=k-1, p=\alpha)
\\ \notag
    &= \frac{\alpha \Gamma(\frac{k-1}{2})}{\Gamma(\frac{k-1}{\alpha})}
    \, {\left( \frac{\sqrt{2}}{k^{1/2 - 1/\alpha}} \right)}^{k-1} 
    {x}^{k-2} \,
    W_{-\frac{\alpha}{2}, 0} \left( -{
        \left( \frac{\sqrt{2} \,x}{k^{1/2 - 1/\alpha}} \right)
    }^\alpha \right)
\end{align}

To verify, its large \(k\) asymptotic mean doesn't have \(k\) dependency:
(See Section \ref{section:fcm-at-inf-k})
\begin{align*}
\lim_{k \to \infty} \E(X|\chimeanDist) &= \alpha^{-1/\alpha} 
\end{align*}
This meets our design goal.
\(\E(X|\chibar_{\alpha,1}) = \E(X|\chi_{\alpha,1})\) 
and \(\lim_{k \to \infty} \E(X|\chibar_{1,k}) = 1\) as expected.
And it is nice to have \(\E(X|\chibar_{2,k}) = 1/\sqrt{2}\) for every \(k\).

\textbf{Interpretation:} The \(\alpha\)-stable law states that 
\begin{align*}
    X_1 + X_2 + ... + X_k &\sim k^{1/\alpha} X 
\end{align*}
Therefore, the \(\left( k^{1/2 - 1/\alpha} \right) \) term can be viewed as the relative scale (ratio)
between the normal scaling factor \(k^{1/2}\) and \(\alpha\)-stable scaling factor \(k^{1/\alpha}\).
That is what we added to \(\sigma_{\alpha,k}\) as the way to "demean" the \(k\) dependency.

\subsection{FCM Hankel Integral}
\label{section:fcm-hankel}

From \eqref{eq:fcm-hankel-intg}, make the following substitutions: 
\(\alpha/p=1/2, d=k-1, \alpha d/p=(k-1)/2, p=\alpha\),
and \(\sigma=\sigma_{\alpha,k}\).
The Hankel integral of FCM for \(k > 0\) is

\begin{align}
\chimean{x} &= 
    \,  \Gamma\left(\frac{k-1}{2}\right)
    \, 
    \frac{1}{2\pi i} \int_{H} dt 
    \, \left( \frac{e^t }{t^{(k-1)/2}} \right) 
    \, \pdfGG\left( x; a=\frac{\sigma_{\alpha,k}}{\sqrt{t}}, d=k-1, p=\alpha \right)
\end{align}
This expression requires \(k \ne 1\).

\subsection{FCM Moments}
\label{section:fcm-moments}

FCM's moment is a special case of GSC's moment, which is carried out in Section \ref{section:gsc-moments}. 

For \(k > 0\), based on \eqref{eq:main-chimean-pdf} and \eqref{eq:gsc-moment}, 
substitute \(\alpha/p = 1/2, d=k-1, p=\alpha\) 
and \(\sigma=\sigma_{\alpha,k}\), and we obtain FCM's moments as
\begin{align}\label{eq:fcm-moment}
\E(X^n|\chimeanDist) 
    &= {\left( \frac{k^{1/2 - 1/\alpha}}{\sqrt{2}} \right)}^{n} \,
    \frac{\Gamma(\frac{k-1}{2})} {\Gamma(\frac{k-1}{\alpha})} 
    \frac{\Gamma(\frac{n+k-1}{\alpha})} {\Gamma(\frac{n+k-1}{2})}
    , & \,\, \text{for} \,\, \alpha \ne 0, k > 0, k \ne 1.
\end{align}
When \(k = 1\) or \(k+n = 1\), it requires typical handling of gamma functions at zero \eqref{eq:gamma-zero}.

For \(k < 0\), based on \eqref{eq:main-chimean-pdf-neg-k} and \eqref{eq:gsc-moment}, 
substitute \(\alpha/p = -1/2, d=k, p=-\alpha\) 
and \(\sigma=1/\sigma_{\alpha,k}\), and we obtain FCM's moments as
\begin{align}\label{eq:fcm-moment-nge-k}
\E(X^n|\chimeanDist) 
    &= {\left( \frac{{|k|}^{1/2 - 1/\alpha}}{\sqrt{2}} \right)}^{-n} \,
    \frac{\Gamma(\frac{|k|}{2})} {\Gamma(\frac{|k|}{\alpha})} 
    \frac{\Gamma(\frac{|k|-n}{\alpha})} {\Gamma(\frac{|k|-n}{2})}
    , & \,\, \text{for} \,\, \alpha \ne 0, k < 0.
\end{align}
When \(n = |k|\), it requires typical handling of gamma functions at zero 
\eqref{eq:gamma-zero}. Higher moments may not exist when \(n > |k|\).

As a validation, when \(\alpha = 1\), \eqref{eq:fcm-moment} is reduced to 
the moment formula of a \(\chi\) distribution 
with \(k\) degrees of freedom and scale of \(1/\sqrt{k}\)
due to \eqref{eq:gamma-reflect}.

FCM's first moment is directly related to GSaS peak PDF in \eqref{eq:gsas-peak-pdf}.
Numeric calculation of FCM's moments is important in that it directly impacts the accuracy
of GAS PDF calculation, especially at small \(\alpha\) and when \(x\) is near zero.


\subsection{FCM Reflection Formula}

We prove the reflection formula shown in \eqref{eq:main-fcm-reflection}
and \eqref{eq:main-fcm-moment-reflection}.
\begin{align} \notag
\E(X^n|\chibar_{\alpha,-k}) 
    &= \frac{\E(X^{-n+1}|\chimeanDist)} {\E(X|\chimeanDist) },
    \,\, k > 0.
\end{align}
The approach is to show that the moments evaluated from the RHS via \eqref{eq:fcm-moment}
is the same as 
those of the LHS from \eqref{eq:fcm-moment-nge-k} with \(k\) replaced with \(-k\).

\begin{proof}
Assume \(k > 0\), the \(n\)-th moment on the LHS is
\begin{align*}
\E(X^n|\text{LHS}) 
    &= \E(X^n|\chibar_{\alpha,-k}) 
    = {\left( \frac{{k}^{1/2 - 1/\alpha}}{\sqrt{2}} \right)}^{-n} \,
    \frac{\Gamma(\frac{k}{2})} {\Gamma(\frac{k}{\alpha})} 
    \frac{\Gamma(\frac{k-n}{\alpha})} {\Gamma(\frac{k-n}{2})}
\end{align*}

The \(n\)-th moment on the RHS is
\begin{align*}
\E(X^n|\text{RHS})
    &= \int_0^\infty \frac{x^{n-3} }{\E(X|\chimeanDist) } 
       \, \chimeanDist\left(\frac{1}{x}\right) \,dx
\end{align*}
Make a change of variable \(t = 1/x\), we get
\begin{align*}
\E(X^n|\text{RHS})
    &= \int_0^\infty \frac{t^{-n+1} }{\E(X|\chimeanDist) } 
       \, \chimeanDist\left(t\right) \,dt
\\
    &= \frac{\E(X^{-n+1}|\chimeanDist)} {\E(X|\chimeanDist) } 
\end{align*}
Hence, \eqref{eq:main-fcm-moment-reflection} is true 
if \(\E(X^n|\text{RHS})\) = \(\E(X^n|\text{LHS})\).

Plug \eqref{eq:fcm-moment} to both the numerator and denominator,
\begin{align*}
\E(X^n|\text{RHS})
    &= \left( 
        {\left( \frac{{|k|}^{1/2 - 1/\alpha}}{\sqrt{2}} \right)}^{-n+1} \,
        \frac{\Gamma(\frac{k-1}{2})} {\Gamma(\frac{k-1}{\alpha})} 
        \frac{\Gamma(\frac{-n+k}{\alpha})} {\Gamma(\frac{-n+k}{2})}
    \right) \bigg/
    \left( 
        {\left( \frac{{|k|}^{1/2 - 1/\alpha}}{\sqrt{2}} \right)} \,
        \frac{\Gamma(\frac{k-1}{2})} {\Gamma(\frac{k-1}{\alpha})}
        \frac{\Gamma(\frac{k}{\alpha})} {\Gamma(\frac{k}{2})}  
    \right)
\\
    &= 
    {\left( \frac{{|k|}^{1/2 - 1/\alpha}}{\sqrt{2}} \right)}^{-n} \,
        \frac{\Gamma(\frac{k}{2})} {\Gamma(\frac{k}{\alpha})} 
        \frac{\Gamma(\frac{k-n}{\alpha})} {\Gamma(\frac{k-n}{2})}
\end{align*}
It is clear that \(\E(X^n|\text{RHS})\) and \(\E(X^n|\text{LHS})\) are the same.

\end{proof}


\subsection{FCM at Infinite Degrees of Freedom}
\label{section:fcm-at-inf-k}

First, we establish the following lemma:
\begin{lemma}

We use \eqref{eq:gamma-infty} to simplify \eqref{eq:fcm-moment} and
obtain the asymptotic value of the \(n\)-th moment for \(k \gg 1\):
\begin{align} \label{eq:fcm-moment-large-k}
\E(X^n|\chimeanDist) 
    &\sim 
    {\left( \frac{{|k|}^{1/2 - 1/\alpha}}{\sqrt{2}} \right)}^{n} \,
    {\left(\frac{k}{2}\right)}^{-n/2} 
    {\left(\frac{k}{\alpha}\right)}^{n/\alpha} 
    =
    {\alpha}^{-n/\alpha} 
\end{align}
It follows from the reflection rule that 
\(\E(X^n|\chibar_{\alpha,-k}) \sim {\alpha}^{n/\alpha} \).
Notice that the asymptotic value has no dependency on \(k\).

\end{lemma}

The proof for Lemma \ref{lemma-fcm-delta1} becomes very easy,
which states that FCM becomes a delta function at \(x = m_{\alpha}\) when \(k \gg 1\),
where \(m_{\alpha}\) is its first moment.
We prove it by showing 
\(m_{\alpha}\) is \({\alpha}^{-1/\alpha}\) 
and its variance is zero.

\begin{proof}

Let \(m_{\alpha} = {\alpha}^{-1/\alpha}\), we have 
\(\E(X^n|\chimeanDist) \sim {(m_{\alpha})}^n\).
Hence, the first moment is obviously 
\begin{align*}
\E(X|\chimeanDist) 
    \sim {\alpha}^{-1/\alpha}
\end{align*}

The second moment is \(\E(X^2|\chimeanDist) \sim m_{\alpha}^2\).
Therefore, its variance is \(\E(X^2|\chimeanDist) - m_{\alpha}^2 = 0\).
FCM becomes a delta function as \(k \to \infty\).

On the negative side, the reflection rule ensures the same proof applies
such that
\(\E(X^2|\chibar_{\alpha,-k}) - \E(X|\chibar_{\alpha,-k})^2 = 0\).
FCM becomes a delta function as \(-k \to -\infty\).

\end{proof}

%% file: section-gsas.tex
\section{GSaS: Generalized Symmetric alpha-Stable Distribution}
\label{section-gsas}

The PDF of a GSaS takes the form of an elegant ratio distribution:
\begin{align*}
\gsas{x} 
    &:= \int_{0}^{\infty} s \, ds \,
      \mcN(xs) \, \chimeanDist(s)
\end{align*}
This two-sided distribution is one of the most important outcomes of this work.
We prove some closed form results in this section.

\subsection{GSaS Peak PDF}

The peak PDF of a GSaS at \(x=0\) comes from the first moment of FCM 
in Section \ref{section:fcm-moments}:
\begin{align}
\gsas{0} \notag
    &= \frac{1}{\sqrt{2\pi}} \int_{0}^{\infty} s \, ds \, \chimeanDist(s)
    = \frac{1}{\sqrt{2\pi}} \E(X|\chimeanDist) 
\\ \label{eq:gsas-peak-pdf}
    &= 
    \left\{
    \begin{array}{ll} \displaystyle
        \frac{k^{1/2 - 1/\alpha}}{2 \sqrt{\pi}} \,
        \frac{\Gamma(\frac{k-1}{2})} {\Gamma(\frac{k-1}{\alpha})} 
        \frac{\Gamma(\frac{k}{\alpha})} {\Gamma(\frac{k}{2})}
        , & \,\, k > 0.
    \\[0.2in] \displaystyle
        \frac{1}{\sqrt{\pi} {\left( {|k|}^{1/2 - 1/\alpha} \right)}} \,
        \frac{\Gamma(\frac{|k|}{2})} {\Gamma(\frac{|k|}{\alpha})} 
        \frac{\Gamma(\frac{|k|-1}{\alpha})} {\Gamma(\frac{|k|-1}{2})}
        , & \,\, k < 0.
    \end{array}
    \right.
\end{align}
When \(k = {\pm}1\), it requires the typical special handling of the gamma function at zero.

The \textbf{standardized peak PDF} is the peak PDF \(\gsas{0}\) multiplied by 
the standard deviation, \(\sqrt{\E(X^2|\gsasDist)}\), if it exists.
We argue that the standardized peak PDF is an important statistic since it is invariant 
to the scale of the sample data. In finance, this means that the standardized peak PDF
measures the similarity of two data sets
taken from two periods of different volatility environments,
or from two different assets altogether.

Aligning the peak PDF was often impossible in the past due to the limited 
one dimensional parameter space in either \(\alpha\) or \(k\).
GSaS overcomes it with a two-dimensional parameter space. 
This opens up much greater flexibility,
e.g., to match the standard deviation, the peak density, and the kurtosis simultaneously.

The first validation of \(\gsas{0}\) is that
it becomes the peak PDF of the original \(\alpha\)-stable distribution when \(k \to 1\):
\begin{align*}
L_{\alpha,1}(0) = P_{\alpha}(0)
    &= \frac{1}{\pi} \,
    \Gamma\left( 1 + 1/\alpha \right)
\end{align*}

The second validation is that
it becomes the peak PDF of the exponential power distribution when \(k \to -1\):
\begin{align*}
L_{\alpha,-1}(0) = \mathcal{E}_{\alpha}(0)
    &= \frac{1}{2 \Gamma\left( 1 + 1/\alpha \right)} 
\end{align*}

\subsection{GSaS Moments}

GSaS has closed form solution for all the moments. This is a great improvement over 
the original \(\alpha\)-stable distribution.

\begin{lemma} \label{lemma-gsas-moments}
The \(n\)-th moment of GSaS is
\begin{align}\label{eq:gsas-moment}
\E(X^n|\gsasDist) 
&= \frac{2^{n/2}}{\sqrt{\pi}} 
    \, \Gamma\left(\frac{n+1}{2}\right) \, 
    \E(X^{-n}|\chimeanDist) ,
    \quad k \ne 0
\end{align}
when \(n\) is even. The odd moments are zero.

Its closed form is
\begin{align}\label{eq:gsas-moment-expanded}
\E(X^n|\gsasDist) 
    &= 
    \left\{
    \begin{array}{ll} \displaystyle
        \frac{2^{n}}{\sqrt{\pi} \left( k^{1/2 - 1/\alpha} \right)^n} \,
        \Gamma\left(\frac{n+1}{2}\right) \,
        \left( 
            \frac{\Gamma(\frac{k-1}{2})} {\Gamma(\frac{k-1}{\alpha})} 
            \frac{\Gamma(\frac{k-n-1}{\alpha})} {\Gamma(\frac{k-n-1}{2} )}
        \right) 
        , & \,\, k > 0.
    \\[0.2in] \displaystyle
        \frac{ {( {|k|}^{1/2 - 1/\alpha} )}^{n} }{\sqrt{\pi}} \,
        \Gamma\left(\frac{n+1}{2}\right) \,
        \left( 
            \frac{\Gamma(\frac{|k|}{2})} {\Gamma(\frac{|k|}{\alpha})} 
            \frac{\Gamma(\frac{|k|+n}{\alpha})} {\Gamma(\frac{|k|+n}{2})}
        \right) 
        , & \,\, k < 0.
    \end{array}
    \right.
\end{align}
When \(k = {\pm}1, n+1, -n\), it requires the typical special handling of the gamma function at zero.

This formula is continuous and well-behaves when \(k\) is 
\emph{sufficiently large}, but it gets very complicated when \(k\) is small. 
This is a complexity inherent from the \(\alpha\)-stable distribution.

\end{lemma}

\begin{proof}
The integral of the \(n\)-th moment is
\begin{align*}
\E(X^n|\gsasDist) 
    &= \int_{-\infty}^{\infty} dx \,
        x^n \, 
        \int_{0}^{\infty} s \, ds \,
            \mcN(xs) \, \chimeanDist(s)
\\
    &=  \int_{0}^{\infty} s \, ds \,
        \left( \int_{-\infty}^{\infty} dx \, x^n \,  \mcN(xs) \right) 
        \, \chimeanDist(s)
\end{align*}

Inside the big parenthesis is the \(n\)-th moment of a normal distribution
(see \eqref{eq:HN-moments}), which is 
\(\frac{2^{n/2}}{\sqrt{\pi}}  s^{-n-1} \Gamma\left(\frac{n+1}{2}\right)\).
Therefore, we get
\begin{align*}
\E(X^n|\gsasDist) 
    &= \frac{2^{n/2}}{\sqrt{\pi}} 
    \, \Gamma\left(\frac{n+1}{2}\right) \, 
    \int_{0}^{\infty} s^{-n} \, ds \, \chimean{s}
\end{align*}
The remaining integral is the "\(-n\)"-th moment of FCM.  
Hence, we get \eqref{eq:gsas-moment}.

To obtain the first line of \eqref{eq:gsas-moment-expanded},
copy \eqref{eq:fcm-moment} into \eqref{eq:gsas-moment}, 
and replace every \(n\) with \(-n\). 

Similarly, copy \eqref{eq:fcm-moment-nge-k} into \eqref{eq:gsas-moment}, 
replace every \(n\) with \(-n\), 
and we reach the second line of \eqref{eq:gsas-moment-expanded}.

\end{proof}

\textbf{GSaS Variance:} 
The second moment \(m_2\), aka the variance, is
\begin{align*}
m_2 = \E(X^2|\gsasDist) 
    &= \E(X^{-2}|\chimeanDist)
\end{align*}
In general, \(m_2 \ne 1\) for finite \(k\).

At \(k \gg 1\), according to \eqref{eq:fcm-moment-large-k},
the variance of the normal-like distribution is
\begin{align*}
\lim_{k \to \infty} m_2 &= {\alpha}^{2/\alpha}
\\
\lim_{k \to -\infty} m_2 &= {\alpha}^{-2/\alpha}
\end{align*}
Both are independent of \(k\). 
This is the intended result from Section \ref{section:fcm-at-inf-k}.

\subsection{GSaS Kurtosis}
Let \(m_4 = \E(X^4|\gsasDist)\), it is straightforward to get
the kurtosis as
\begin{align} \label{eq:gsas-kurtosis-mnt}
\frac{m_4}{m_2^2}
    &= 3 \, \frac{\E(X^{-4}|\chimeanDist)} {\E(X^{-2}|\chimeanDist)^2}
\end{align}
The excess kurtosis is \(\text{exKurt} = m_4/m_2^2 - 3\).

For \(k > 0\), the closed form is
\begin{align}
\frac{m_4}{m_2^2}
    &= \displaystyle \notag
3 \, { 
    \left( 
        \frac{\Gamma\left(\frac{k-1}{2}\right)}      {\Gamma\left(\frac{k-1}{\alpha}\right)} 
        \frac{\Gamma\left(\frac{k-5}{\alpha}\right)} {\Gamma\left(\frac{k-5}{2} \right)}
    \right) 
} \bigg/ {
    {\left( 
        \frac{\Gamma\left(\frac{k-1}{2}\right)}      {\Gamma\left(\frac{k-1}{\alpha}\right)} 
        \frac{\Gamma\left(\frac{k-3}{\alpha}\right)} {\Gamma\left(\frac{k-3}{2} \right)}
    \right)}^2
}
\\ \label{eq:kurtosis-exact}
    &=
    3 \,
    \left(\frac{k-5}{k-3}\right)
    \frac{
        \Gamma(s(k-5)) \Gamma(s(k-1))
    }{  \Gamma(s(k-3))^2 }
    , \quad \where \, s = 1/\alpha.
\end{align}

We prove the linear relation of excess kurtosis in \eqref{eq:gsas-linear-kurtosis} 
for \(k \gg 1\) as follows.
\begin{proof}
In Figure \ref{fig:gsas_ex_kurt_by_s}, notice the singular point at \(s=1/2, k=3\).
Hence we make a change of variable by \(k_3 := k - 3\):
\begin{align*}
\frac{m_4}{m_2^2}
    &=
    3 \,
    \left( 1-\frac{2}{k_3} \right)
    \frac{ \Gamma(s(k_3-2)) \Gamma(s(k_3+2)) }
    { \Gamma(s k_3)^2 }
\end{align*}

Next, the gamma functions are expanded via the Sterling's formula \eqref{eq:gamma-sterling}
for large \(k_3\). It is simplified to
\begin{align*}
\frac{m_4}{m_2^2}
    &=
    3 \,
    {\left( 1-\frac{2}{k_3} \right)}^{s k_3-2s+1/2} 
    \, {\left( 1+\frac{2}{k_3} \right)}^{s k_3+2s-1/2}
\end{align*}

Take the log on both sides, and retain the Taylor series of \(\log(1+x)\) 
up to \(k_3^{-2}\). We have
\begin{align*}
\log{\left(\frac{m_4}{3 \, m_2^2}\right)}
    &=
    \frac{4}{k_3} {\left( s -\frac{1}{2} \right)}
\end{align*}

Since \(m_4/m_2^2 = 3 + \text{exKurt}\) where \(\text{exKurt}\) is the excess kurtosis,
we arrive at the proof that
\begin{align*}
{\left( s -\frac{1}{2} \right)} 
    &= 
    \left(\frac{k-3}{4}\right)
    \log{\left(1 + \frac{\text{exKurt}}{3}\right)}
\end{align*}

\end{proof}

As \(k \to \infty\), the excess kurtosis approaches zero.
GSaS becomes a normal distribution asymptotically,
\(\mcN(0,m_2)\).
Only when \(\alpha = 1\), \(\lim_{k \to \pm\infty} m_2 = 1\).
It becomes a \(\mcN(0,1)\).

\subsection{GSaS CDF}\label{section:gsas-cdf}
The CDF formula of GSaS is straightforward since 
\(\int_{-\infty}^{x} \, \mcN\left(z\right) dz\) 
\(= \frac{1}{2}+\frac{1}{2} \erf\left(\frac{x}{\sqrt{2}}\right)\). 

\begin{lemma}
The CDF of a GSaS is 
\begin{align}\label{eq:gsas-cdf}
\Phi[L_{\alpha,k}](x) 
    &:= \int_{-\infty}^{x} \, L_{\alpha,k}\left(z\right) dz
    = \frac{1}{2} + 
       \frac{1}{2}
    \int_{0}^{\infty} \, ds \, \erf\left( \frac{xs}{\sqrt{2}} \right)
        \, \chimeanDist(s)
\end{align}
\end{lemma}

\begin{proof}
We begin with the double integral via the PDF, such that
\begin{align*}
\Phi[L_{\alpha,k}](x) 
    &= \int_{-\infty}^{x} \, L_{\alpha,k}\left(z\right) dz
\\
    &= \int_{-\infty}^{x} \, dz
      \int_{0}^{\infty} s \, ds \,
      \mcN(zs) \, \chimeanDist(s)
\end{align*}
The order of the two integrals is switched, and we make a change of variable \(u = zs\):
\begin{align*}
\Phi[L_{\alpha,k}](x) 
    &= \int_{0}^{\infty} s \, ds \,
      \left( \int_{-\infty}^{x} \, dz \mcN(zs) \right) 
      \, \chimeanDist(s)
\\
    &= \int_{0}^{\infty} \, ds \,
      \left( \int_{-\infty}^{xs} \, du \mcN(u) \right) 
      \, \chimeanDist(s)
\end{align*}
The inner integral is replaced with the error function:
\begin{align*}
\Phi[L_{\alpha,k}](x)
    &= \int_{0}^{\infty} \, ds \,
      \left( \frac{1}{2}+ \frac{1}{2} \erf\left( \frac{xs}{\sqrt{2}} \right) \right) 
      \, \chimeanDist(s)
\\
    &= \frac{1}{2} + 
       \frac{1}{2}
    \int_{0}^{\infty} \, ds \, \erf\left( \frac{xs}{\sqrt{2}} \right)
        \, \chimeanDist(s)
\end{align*}
We arrive at the desired formula for the CDF.
\end{proof}

\subsection{Fractional Hypergeometric Function}\label{section:frac-hyper}

Due to the connection between the CDF of Student's t and 
the Gauss hypergeometric function 
\(\,_{2}F_{1}\left( a,b;c; z \right)\) (defined in 15.2.1 of DLMF\cite{NIST:DLMF}), 
we propose a framework to 
extend \(\,_{2}F_{1}\) fractionally.

\begin{definition}[Fractional hypergeometric function (FHF)]
Let \(M(b,c; z)\) be the confluent hypergeometric function, or called the Kummer function,
defined in 13.2.2 of DLMF.
(Note we avoid using \(a\) because it is confusing to \(\alpha\).)

The fractional hypergeometric function is defined in a ratio-distribution style as
\begin{align}
    M_{\alpha,k}(b,c; x) 
    &:= \sqrt{\frac{k}{2\pi}}
    \int_0^\infty s\,ds \, M\left(b,c; \frac{x k s^2}{2} \right) \, 
    \chimeanDist(s)
\end{align}

Note that \(x\) usually takes a negative value. This is due to the asymptotic behavior:
\(M(b,c; z) \sim e^z z^{b-c} / \Gamma(b)\) as \(z \to \infty\) (13.2.23 of DLMF).

\qedbar
\end{definition}

\begin{lemma}
The CDF of a GSaS can be expressed by the fractional hypergeometric function as
\begin{align}\label{eq:gsas-cdf-hyper}
\Phi[L_{\alpha,k}](x) 
    &= 
    \frac{1}{2} + 
    \frac{x}{\sqrt{k}} 
    M_{\alpha,k}\left( \frac{1}{2}, \frac{3}{2}; -\frac{x^2}{k} \right) 
\end{align}
\end{lemma}

\begin{proof}
The error function can be expressed by the Kummer function as (13.6.7 of DLMF)
\begin{align*}
    \erf(x) &= \frac{2x}{\sqrt{\pi}} \, M\left( \frac{1}{2}, \frac{3}{2}; -x^2\right) 
\end{align*}
Substitute it to \eqref{eq:gsas-cdf}, we obtain \eqref{eq:gsas-cdf-hyper}.
\end{proof}

\begin{lemma}
The fractional hypergeometric function subsumes the Gauss hypergeometric function at \(\alpha=1\):
\begin{align} \label{eq:hyp2f1-equiv}
    M_{1,k}(b,c; x) \, B\left(\frac{k}{2}, \frac{1}{2} \right)
    &= 
    \,_{2}F_{1}\left( b, \frac{k+1}{2}; c; x \right)
\end{align}
where \(B(a,b)\) is the beta function (8.17.3	of DLMF).
Notice that \(k\) plays the role of the second parameter in \(\,_{2}F_{1}\).
 
\end{lemma}

\begin{proof}
Since \(\chibar_{1,k}(s) = \frac{k^{k/2}}{2^{k/2-1} \Gamma(k/2)} s^{k-1} e^{-k s^2/2}\),
\eqref{eq:hyp2f1-equiv} becomes
\begin{align*}
    \int_0^\infty ds \, M\left(b,c; \frac{x k s^2}{2} \right) \, s^{k} \, e^{-k s^2/2} 
    &= \frac{1}{2} \Gamma\left(\frac{k+1}{2}\right) \, 
    \left(\frac{2}{k}\right)^{(k+1)/2}
    \,_{2}F_{1}\left( b, \frac{k+1}{2}, c; x \right)
\end{align*}
We make two changes of variable selectively, \(a = (k+1)/2\) and \(z = k/2\), such that
\begin{align*}
    \int_0^\infty ds \, M\left(b,c; x z s^2 \right) \, s^{2a-1} \, e^{-z s^2} 
    &= \frac{1}{2} \Gamma(a) \, z^{-a} \,
    \,_{2}F_{1}\left( b, a, c; x \right)
\end{align*}
Lastly, let \(t = s^2\), and \(x = q/z\), we get
\begin{align*}
    \int_0^\infty dt \, M\left(b,c; qt \right) \, t^{a-1} \, e^{-z t} 
    &= \Gamma(a) \, z^{-a} \,
    \,_{2}F_{1}\left( b, a, c; q/z \right)
\end{align*}
which is the Laplace transform of the Kummer function 
mentioned in 13.10.3 of DLMF.
Therefore, this Lemma is essentially a disguised Laplace transform.

\end{proof}

It is then straightforward to obtain the known expression of the CDF of Student's t as\footnote{
See also Wikipedia at \url{https://en.wikipedia.org/wiki/Student\%27s_t-distribution\#Cumulative_distribution_function}
}
\begin{align}\label{eq:t-cdf-hyp2f1}
\Phi[L_{1,k}](x) 
    &= 
    \frac{1}{2} + 
    \frac{x}{\sqrt{k}} \, \frac{1}{ B(\frac{k}{2}, \frac{1}{2})}
    \,_{2}F_{1}\left( \frac{1}{2}, \frac{k+1}{2}; \frac{3}{2}; -\frac{x^2}{k} \right)
    && (x^2 < k)
\end{align}

Next, the CDF of Student's t can also be expressed in the regularized incomplete beta function
\(I_p(a,b)\), where \(0 \le p \le 1\) 
(See 8.17.2 in DLMF) such that
\begin{align}\label{eq:t-cdf-betainc}
\Phi[L_{1,k}](x) 
    &= 
    1 - \frac{1}{2} I_p\left( \frac{k}{2}, \frac{1}{2} \right)
    = \frac{1}{2} + \frac{1}{2} I_q\left( \frac{1}{2}, \frac{k}{2} \right)
    && (x \ge 0)
\\ \notag
    & \quad \where p = \frac{k}{x^2 + k}, q = 1-p.
\end{align}
This is an application of the following equation (by combining 8.17.7 and 15.8.1 in DLMF):
\begin{align*}
    I_q(a,b)
    &= 
    \frac{1}{a B(a,b)} \, \left(\frac{q}{p} \right)^{a} \, 
    \,_{2}F_{1}\left(a, a+b, a+1; -\frac{q}{p} \right)
\end{align*}
where we assign \(a = 1/2\), \(b = k/2\), and \(q/p = x^2/k\).

\subsection{GSaS Characteristic Function}\label{section:gsas-cf}

We explore the CF of GSaS that was mentioned in \eqref{eq:main-gsas-cf}:
\begin{align*}
\phi_{\alpha,k}(\zeta) &= 
    \sqrt{2\pi} \int_{0}^{\infty} \, ds \,
        \mcN\left(\frac{\zeta}{s}\right)
        \, \chimeanDist(s)
    && (k > 0)
\end{align*}
It forms a high transcendental function,
fractionally extended such that it subsumes both 
the stretched exponential function \(e^{-\zeta^\alpha}\) in \eqref{eq:gsas-cf-alpha-piece},
and the \emph{modified Bessel function of the second kind}\footnote{
See also Wikipedia at \url{https://en.wikipedia.org/wiki/Bessel_function\#Modified_Bessel_functions}.
} \(K_{\frac{k}{2}}(\zeta \sqrt{k})\) in \eqref{eq:gsas-cf-k-piece}
(See \textbf{scipy.special.kn} and (11.117) in \cite{Arfken:1985}).

The first part is straightforward. Set \(k=1\), 
the LHS by definition of the \(\alpha\)-stable distribution is
\begin{align}\label{eq:gsas-cf-alpha-piece}
\phi_{\alpha,1}(\zeta) = e^{-\zeta^\alpha} 
    &= \sqrt{2\pi} \int_{0}^{\infty} \, ds \,
      \mcN\left(\frac{\zeta}{s}\right)
      \, \chibar_{\alpha,1}(s)
    && (\zeta \ge 0)
\\ \notag
    &= \int_{0}^{\infty} \, \frac{ds}{s} \,
      \mcN\left(\frac{\zeta}{s}\right)
      \, \left( \sqrt{2\pi}  s \, \chibar_{\alpha,1}(s) \right)
\end{align}
The second line is essentially the same as the second line of \eqref{eq:lambda-decomp}.
Note the term in the parenthesis: \(\sqrt{2\pi}  s \, \chibar_{\alpha,1}(s) \propto V_\alpha(s)\).
And \(V_\alpha(s) = \chi_{\alpha, 2}(s)\).
We can see that, dimensionally speaking, the \(s^k\) term is correct.

The second part is based on the CF of Student's t,
\(\phi_{1,k}(\zeta)\), such that
\begin{align}\label{eq:gsas-cf-k-piece}
\phi_{1,k}(\zeta) = 
    \frac{t^{k/2} }{ 2^{k/2-1} \, \Gamma(\frac{k}{2})} K_{\frac{k}{2}}(t)
    &= \sqrt{2\pi} \int_{0}^{\infty} \, ds \,
      \mcN\left(\frac{\zeta}{s}\right)
      \, \chibar_{1,k}(s) 
    , \quad \where t = \zeta \sqrt{k}, \zeta \ge 0
\end{align}

The complex normalization constant between \(\phi_{1,k}(\zeta)\) and \(K_{\frac{k}{2}}(t)\) is
to ensure the general requirement of a CF such that \(\phi_{\alpha,k}(0) = 1\). 
The \(t^{k/2}\) term is due to the divergent asymptotic of \(K_{\frac{k}{2}}(t)\) at \(t=0\)
(See (11.124) in \cite{Arfken:1985}):
\begin{align*}
\lim_{t \to 0} K_\nu(t) &= 
    2^{\nu-1} \, \Gamma(\nu) \, t^{-\nu} 
    , \quad \nu = \frac{k}{2} > 0
\end{align*}

\subsection{GSaS Series Representation of PDF and Tail Behavior}\label{section:gsas-tail}

The PDF of GSaS \((k > 0)\) has two known series representations. 
One in term of \(x\) and the other in terms of \(1/x\). 
Both have issues in terms of radius of convergence.
More future work is needed in this area.

Here we present the \(1/x\) result to investigate the tail behavior. 
This is a working example of the four-parameter Wright function \eqref{eq:wright-fn-4ways}.
(The small \(x\) series is presented in Section \ref{section-gsc-mgf}.)

\begin{lemma}
The series representation of GSaS PDF in terms of  \(1/x\) is
\begin{align} \label{eq:gsas-pdf-series}
\gsas{x} 
    &= 
    \,\left(
        \frac{S_{\alpha,k} }{2 \sqrt{\pi}} 
    \right)
    \, \frac{1}{x}
    {\left( \frac{x}{\gsasSigma} \right)}^{-k+1} 
    \, W \left[ \begin{matrix}
        \frac{\alpha}{2} ,& \frac{k}{2}
        \\ -\frac{\alpha }{2} ,& 0
    \end{matrix}\right]
    \left(-\left( \frac{x}{\gsasSigma} \right)^{-\alpha} \right)
\\ \notag &= 
    \,\left(
        \frac{S_{\alpha,k} }{2 \sqrt{\pi} \, \gsasSigma} 
    \right)
    {\left( \frac{x}{\gsasSigma} \right)}^{-k} 
    \, \sum_{n=1}^\infty 
        \frac{\Gamma\left( \frac{\alpha n}{2} + \frac{k}{2} \right)} {\Gamma(n+1) \, \Gamma(-\frac{\alpha n}{2})}
    \left[-\left( \frac{x}{\gsasSigma} \right)^{-\alpha} \right]^n
\\ \label{eq:gsas-pdf-symbols}
    & \where 
    S_{\alpha,k} := \frac{\alpha \Gamma(\frac{k-1}{2})}{\Gamma(\frac{k-1}{\alpha})} 
    \,\, \text{and} \,\,
    \gsasSigma := \frac{\sqrt{2}}{\fcmSigma} 
        = 2 \big/ \left( k^{1/2 - 1/\alpha} \right)   
\end{align}
This series works better for smaller \(\alpha\) where the Pareto tails are more prominent.
\end{lemma}

\begin{proof}
From \eqref{eq:frac-chi-mean}, we get
\begin{align*} 
\gsas{x} &= 
    S_{\alpha,k}
    \int_{0}^{\infty} \, ds \,
      \mcN(xs) \, 
    {\left( \frac{s}{\fcmSigma} \right)}^{k-1} \,
    W_{\frac{\alpha}{2}, 0} \left( -{
        \left( \frac{s}{\fcmSigma} \right)
    }^\alpha \right)
\\ &= 
    S_{\alpha,k} \fcmSigma \,
    \int_{0}^{\infty} \, dz \,
      \mcN(xz \, \fcmSigma) \, 
    {z}^{k-1} \,
    W_{\frac{\alpha}{2}, 0} \left( -{z}^\alpha \right)
\end{align*}
where \( z = s / \fcmSigma\).

Substitute the Wright function with its series representation,
\begin{align*} 
\gsas{x} &= 
    \frac{\sqrt{2} S_{\alpha,k}}{\gsasSigma} \,
    \, \int_{0}^{\infty} dz
    \, \left(
        \frac{1}{\sqrt{2 \pi}} e^{-{(zx / \gsasSigma)}^2} \right) 
        \, 
        \sum_{n=1}^\infty 
        \frac{{(-1)}^n {z}^{\alpha n + k - 1} } {n!\,\Gamma(-\frac{\alpha n}{2})}
\end{align*}

Integrate term by term, each integral is just the moment formula of a normal distribution.
\begin{align*} 
\gsas{x} &= 
    \frac{ S_{\alpha,k}  }{\sqrt{\pi} \gsasSigma} \,
    \, \sum_{n=1}^\infty 
        \frac{1}{2} 
        \frac{{(-1)}^n  } {\Gamma(n+1) \, \Gamma(-\frac{\alpha n}{2})}
    \,  \Gamma\left( \frac{\alpha n + k}{2} \right)
    {\left( \frac{\gsasSigma}{x} \right)}^{\alpha n + k} 
\end{align*}
Re-arrange the term inside the summation, we arrive at the second line of \eqref{eq:gsas-pdf-series}.
Then replace the summation with the four-parameter Wright function \eqref{eq:wright-fn-4ways}.
\end{proof}

Take the first term of the summation, 
we observe that the tail behavior of GSaS is proportional to \(x^{-(\alpha+k)}\).
This illustrates the effect of \(k\) in which larger \(k\) makes the tail lighter.

To validate the Pareto tail behavior, when \(k \to 1\), 
it becomes \(\alpha c_\alpha x^{-(1+\alpha)}\), where 
\(c_\alpha = \Gamma(\alpha) \sin(\frac{\alpha \pi}{2}) / \pi\), 
consistent with (3.43) of Nolan (2020)\cite{Nolan:2020}. 

%% file: section-gexppow.tex
\section{GEP: Generalized Exponential Power Distribution}
\label{section:gen-exp-pow-dist}

The exponential power distribution 
\(\mathcal{E}_\alpha(x) = e^{-{|x|}^\alpha}/2\GammaAlpha\)
is mentioned in Definition \ref{def:gexppow}.
It has been used to model asset returns in finance,
as shown in the LHS of \eqref{eq:lambda-decomp}.
Its usefulness stems from its simplicity.
Unlike its \(\alpha\)-stable counterpart, the existence of moments is never an issue. 

GEP \(\gepDist\) enriches \(\mathcal{E}_\alpha\) with the \emph{degree of freedom} parameter \(k\).
We have shown that \(\gepDist = L_{\alpha,-k}\)
by the clever design of FCM. Therefore its properties 
can be obtained indirectly from those of GSaS.
We still carry out some proofs here based on its product distribution definition 
\eqref{eq:main-gexppow-pdf}, for completeness sake as well as for validation.

Its peak PDF is simply 
\(\gepDist(0) = \frac{1}{\sqrt{2\pi}} / \E(X|{\chimeanDist})\). 
It can be standardized
by multiplying with its standard deviation, which is derived below.

All the moments of \(\gepDist\) exist. 
The reader can take a closer look at Figure \ref{fig:gexppow_ex_kurt_by_s}.
All the kurtosis contours are continuous on the half plane of \(k > 0\).
Its moment formula is straightforward:

\begin{lemma}
The \(n\)-th moment of \(\gepDist\) is
\begin{align}\label{eq:gexppow-moment}
\E(X^n|\gepDist) &= 
    \frac{\E(X^n|\mcN) \, \E(X^{n+1}|\chimeanDist)} {\E(X|\chimeanDist)}
\end{align}
where \(n\) is even. All the odd moments are zero.
\end{lemma}

\begin{proof}
The integral of the \(n\)-th moment from \eqref{eq:main-gexppow-pdf} is
\begin{align*}
\E(X^n|\gepDist)
    &= \frac{1}{\E(X|\chimeanDist)} 
    \, \int_{-\infty}^{\infty} x^n \, dx
    \, \int_{0}^{\infty} \, ds \,
      \mcN\left(\frac{x}{s}\right)
      \, \chimeanDist(s)
\end{align*}

With a change of variable \(z = x/s\), we can separate the two integrals:
\begin{align*}
\E(X^n|\gepDist)
    &= \frac{1}{\E(X|\chimeanDist)} 
    \, \int_{-\infty}^{\infty} z^n \mcN\left(z\right) \, dz
    \, \int_{0}^{\infty} \, s^{n+1}
      \, \chimeanDist(s) \, ds 
\\
    &= 
    \frac{\E(X^n|\mcN) \, \E(X^{n+1}|\chimeanDist)} {\E(X|\chimeanDist)}
\end{align*}

\end{proof}

The excess kurtosis of \(\gepDist\) in \eqref{eq:gexppow-exkurt}
is derived from \( \E(X^4|\gepDist) / \E(X^2|\gepDist)^2 \).

\begin{center}\qedbar\end{center}

The CDF of \(\gepDist\) is derived below by following
the same logic in Section \ref{section:gsas-cdf}.

\begin{lemma}
The CDF of \(\gepDist\) is
\begin{align}\label{eq:gexppow-cdf}
\Phi[\gepDist](x) 
    &:= \int_{-\infty}^{x} \, \mathcal{E}_{\alpha,k}\left(z\right) dz
    = \frac{1}{2} + 
    \frac{1}{2 \E(X|\chimeanDist)}
    \int_{0}^{\infty} \, s \, ds \, \erf\left( \frac{x}{\sqrt{2}\, s} \right)
        \, \chimeanDist(s)
\end{align}
\end{lemma}

\begin{proof}
We begin with the double integral via the PDF, such that
\begin{align*}
\Phi[\gepDist](x)
    &= \frac{1}{\E(X|\chimeanDist)} 
      \int_{-\infty}^{x} \, dz
      \int_{0}^{\infty} \, ds \,
      \mcN\left(\frac{z}{s}\right) \, \chimeanDist(s)
\end{align*}
The order of the two integrals are switched, and we make a change of variable \(u = z/s\):
\begin{align*}
\Phi[\gepDist](x)
    &= \frac{1}{\E(X|\chimeanDist)} 
      \int_{0}^{\infty} \, ds \,
      \left( \int_{-\infty}^{x} \, dz \, \mcN\left(\frac{z}{s}\right) \right) 
      \, \chimeanDist(s)
\\
    &= \frac{1}{\E(X|\chimeanDist)} 
      \int_{0}^{\infty} \, s \, ds \,
      \left( \int_{-\infty}^{x/s} \, du \, \mcN(u) \right) 
      \, \chimeanDist(s)
\end{align*}
The inner integral is replaced with the error function:
\begin{align*}
\Phi[\gepDist](x)
    &= \frac{1}{\E(X|\chimeanDist)} 
      \int_{0}^{\infty} \, s \, ds \,
      \left[ \frac{1}{2}+\frac{1}{2} \erf\left( \frac{x}{\sqrt{2}\, s} \right) \right]
      \, \chimeanDist(s)
\\
    &= \frac{1}{2} + 
       \frac{1}{2 \E(X|\chimeanDist)}
    \int_{0}^{\infty} \, s \, ds \, \erf\left( \frac{x}{\sqrt{2}\, s} \right)
        \, \chimeanDist(s)
\end{align*}
We arrive at the desired formula for the CDF.
\end{proof}

%% file: section-gas.tex
\section{GAS: Generalized Alpha-Stable Distribution}
\label{section:gas}

The invention of GAS is experimental. 
Since GSaS is very successful, it is a natural next step to add skewness to it. 

We take a simple approach -- 
The separation of concerns between the skew Gaussian kernel \(\gskew\) and fractional chi-1 
in \(\afstableDist\) is extended along the FCM side,
as shown below from the first line to the second line:
\begin{align} \notag
\afstable{x} = L_{\alpha,1}^\theta(x)
    &= \int_{0}^{\infty} s \, ds \,
      g_{\alpha}^{\theta}(x,s) \, \chibar_{\alpha,1}(s)
\\
\implies \quad \quad \label{eq:gas-pdf}
\gas{x}
    &:= \int_{0}^{\infty} s \, ds \,
      \gskew \, \chimean{s}
\end{align}
The skewness of GAS comes from \(\gskew\) when \(\theta \ne 0\).
When \(\theta = 0\), \(\gskew = \mcN(xs)\) and GAS becomes GSaS.

Analyzing the behavior of skewness in depth when \(|k| > 1\) is a complicated topic. 
\(\theta\) is confined by \(\gas{x} > 0\) for all \(x \in \R\).
When \(|\theta|\) is small, the skew PDF behaves well. 
But the valid range of \(\theta\)
appears to be \emph{much smaller} than the \emph{Feller-Takayasu diamond}: 
\(\theta \le \min\{\alpha, 2-\alpha\}\) \cite{Mainardi:2007, Mainardi:2020}.
The larger \(k\) is, the smaller range of \(\theta\).
This issue is left for future research.


\begin{lemma}(GAS Symmetry Relation)
The following symmetry relation holds, which is an extension from 
\(\alpha\)-stable (See Section 4 in \cite{Mainardi:2007}):
\begin{align*}
\gas{-x}
    &= L_{\alpha,k}^{-\theta}(x)
\end{align*}
Therefore, one only needs to consider the case of \(x > 0\).
\end{lemma}
\begin{proof}
Both \(\alpha, \theta\) only show up in the skew kernel, not in FCM.
We just need to prove the relation: \(g_{\alpha}^{\theta}(-x,s) = g_{\alpha}^{-\theta}(x,s)\).

This is easy to show since the cosine function is symmetric 
and the tangent function is asymmetric. 
In \eqref{eq:main-skew-gaussian}, when \(x\) and \(\tau\) change sign simultaneously,
the term \(\cos\left(\tau\, {(st)}^\alpha + \frac{x}{q}\, st \right)\) remains the same.

\end{proof}

Two subjects are unresolved for GAS as of this writing.
First, the analytic form for the moments of GAS is wanting. 
Second, the \(k\) extension of a well-known reciprocal relation 
between \(\alpha\) and \(1/\alpha\) is unknown yet. When \(k = 1\), Feller first showed 
\begin{align*}
\frac{1}{x^{\alpha+1}} \, L_{1/\alpha}^{\theta}(x^{-\alpha})
    &= L_{\alpha}^{\theta^{*}}(x),
    \quad \theta^{*} = \alpha (\theta + 1) - 1,
\end{align*}
where \(1/2 \le \alpha \lt 1\) and \(x > 0\)
(See Section 4 in \cite{Mainardi:2007} and Lemma 2 in p. 583 of \cite{Feller:1971}).
What is the relation when \(k > 1\)?

\subsection{GAS PDF at Small Alpha}
\label{section:gas-pdf-small-alpha}

When \(\alpha\) is small and \(\theta \ne 0\), the GAS integral \eqref{eq:gas-pdf}
can be difficult to converge.
Its behavior can be better understood by splitting the skew kernel in the integral into two parts: 
(a) The Gaussian kernel for small \(s\); (b) The tail kernel for large \(s\).
The integral over the former is much easier to converge than the later.

Let \(g_{\alpha}^\theta(x,s) = n(x,s,q) + h_{\alpha}^\theta(x,s)\),
where
\begin{align*}
n(x,s,q)
    &= \frac{1}{q \pi} \int_{0}^{\infty} dt
        \, \cos\left(\frac{x}{q}\, st \right)  e^{-t^2/2} 
        \quad (s \ge 0)
\end{align*}
and
\begin{align}
h_{\alpha}^\theta(x,s)
    &= \frac{1}{q \pi} \int_{0}^{\infty} dt
        \, \left[
            \cos\left(\tau\, {(st)}^\alpha + \frac{x}{q}\, st \right)
            - \cos\left(\frac{x}{q}\, st \right)
        \right] e^{-t^2/2} 
        \quad (s \ge 0)
\end{align}

Note that \(n(x,s,q)\) can be simplified by the Gaussian integral to
\begin{align*}
n(x,s,q)
    &= \frac{1}{q} \, g_{\alpha}^0\left(\frac{x}{q},s\right)
     = \frac{1}{q} \, \mcN\left(\frac{xs}{q}\right)
\end{align*}
Hence, \eqref{eq:gas-pdf} can be rewritten as
\begin{align}
\gas{x}
    &= \frac{1}{q} \, \int_{0}^{\infty} s \, ds \,
       \mcN\left(\frac{xs}{q}\right) \, \chimean{s}
     + \int_{0}^{\infty} s \, ds \,
       h_{\alpha}^\theta(x,s) \, \chimean{s}
\end{align}

The first integral converges quickly over several multiples of \(q/x\)
(assuming \(x \ne 0\)), 
and most weight of \(\gas{x}\) can be obtained from it.

However, the second integral can be hard to converge. 
It gets its contribution when \(\tau \ne 0\)
and \(s\) is away from zero. 
The tails of both \(h_{\alpha}^\theta(x,s)\) and  \(\chimean{s}\) can decay very slowly
when \(\alpha\) is small.
The larger \(|\tau|\) is, the slower the decay.

Empirically, at \(\alpha = 0.1\), the integral over \(s\) may have to 
extend to \(10^{3}\) in order to achieve \(10^{-2}\) relative tolerance.
And calculating \(h_{\alpha}^\theta(x,s)\) (and/or \(g_{\alpha}^\theta(x,s)\))
at large \(s\) is time-consuming.

%% file: section-gsc-rv.tex
\section{Generation of Random Variables for GSC, FCM, and GSaS}\label{section-gsc-random}

In this section, we address the issue of generating random numbers for a GSC.
An established stochastic framework has been written on the Wikipedia page\footnote{
See \url{https://en.wikipedia.org/wiki/Stable_count_distribution\#Generation_of_Random_Variables}}
of the stable count distribution\cite{Wikipedia_StableCountDistribution}. 
That framework is formalized and enlarged to GSC here.
If one can generate random numbers \(\{S_t\}\) for GSC, then it is straightforward to 
simulate random numbers \(\{X_t\}\) for GSaS by a ratio distribution:
\( X_t = \mcN / S_t \), where \(\mcN\) is a standard normal variable.

Given a user-specific volatility \(\sigma_u\) describing how fast \(S_t\) should change,
and a scale parameter \(\theta_u\) (default to 1),
we assume the random variable \(S_t\) should evolve by the following 
generalized \emph{Feller square-root process}\cite{Feller:1951}:
\begin{align}\label{eq:feller-sqare-root}
    dS_t = \sigma_u^2 \, \mu\left(\frac{S_t}{\theta_u} \right) \, dt 
        + \sigma_u \sqrt{S_t} \, dW_t
\end{align}
As \(t \to \infty\), \(\{S_t\}\) will distribute like the equilibrium distribution that 
\(\mu(x)\) is designated for. 

\subsection{Generation of Random Variables for GSC}\label{section-gsc-rv}

\begin{lemma}
The \(\mu(x)\) solution is derived from the Fokker-Planck equation
with \(\gscNx\) as the equilibrium distribution, which yields (assume \(\theta_u=1\))
\begin{align}\label{eq:mu-by-gsc}
\mu(x) 
    &= 
    \frac{ (x\frac{d}{dx} + 1) \, \gscNx }{ 2 \, \gscNx }
\\ \notag 
    &= 
    \frac{x}{2} \frac{d}{dx} [\log \gscNx ] + \frac{1}{2}
\end{align}
The second line is elegant, but only useful in limited cases.
\end{lemma}

\begin{proof}
As \(t \to \infty\), the stationary Fokker-Planck equation of \eqref{eq:feller-sqare-root} is 
\begin{align*}
\frac{d}{dx} \left[
    \sigma_u^2 \, \mu\left(x \right) \gscNx \right] 
    &= 
    \frac{1}{2} \frac{d^2}{dx^2} \left[
        \sigma_u^2 x \, \gscNx \right]
\end{align*}

First, \(\sigma_u^2\) cancels out from both sides and is irrelevant to the solution.
Second, apply \(\int_x^{\infty} dx\) to both sides. And we know the constant 
at \(x=\infty\) should be zero, because \(\gscNx\) is the PDF of a distribution. Hence,
\begin{align*}
    \mu\left(x \right) \gscNx 
    &= 
    \frac{1}{2} \frac{d}{dx} \left[
         x \, \gscNx \right]
\end{align*}
Move \(\gscNx\) from LHS to RHS, we obtain \eqref{eq:mu-by-gsc}.
The second line is from the simple application of 
\(\frac{d}{dx} (\log f(x)) = \frac{1}{f(x)} \frac{d}{dx} f(x)\).

\end{proof}

Note that this model subsumes the renown Cox–Ingersoll–Ross (CIR) model\cite{CIR:1985} because 
CIR's equilibrium distribution is a gamma distribution which is subsumed 
by GSC at \(\alpha=0, p=1\).

To simplify the symbology going forward, we define
\begin{align} \label{eq:q-wright-ratio}
Q_{\alpha}(z)  &:= 
        - \frac{W_{-\alpha,-1}\left(-z\right)}
               {W_{-\alpha,0 }\left(-z\right)}
\end{align}

\begin{lemma}
The mean-reverting solution of a GSC in terms of the Wright functions is
\begin{align}\label{eq:mu-sol-adv}
\mu(x) 
    &=  -\frac{p}{2\alpha} 
            \frac{W_{-\alpha,-1}\left(-{\left( \frac{x}{\sigma} \right)}^{p}\right)}
                 {W_{-\alpha,0 }\left(-{\left( \frac{x}{\sigma} \right)}^{p}\right)}
        +
        \left( \frac{d}{2} - \frac{p}{2\alpha} \right) 
    & \alpha \in (0,1)
\\ \notag
    &=  \frac{p}{2\alpha} Q_{\alpha}\left({\left( \frac{x}{\sigma} \right)}^{p}\right)
        +
        \left( \frac{d}{2} - \frac{p}{2\alpha} \right) 
\end{align}
\end{lemma}

\begin{proof}

From \eqref{eq:mu-by-gsc}, we make a change of variable,
\(z = {\left( \frac{x}{\sigma} \right)}^{p}\),
and \(d' = (d-1)/p + 1\). Then 
\begin{align*}
    \gscNx &= \scN_\alpha(z; \sigma=1, d', p=1)
\end{align*}
The change of variable also yields \(pz\frac{d}{dz} = x\frac{d}{dx}\), 
which simplifies \(\mu(x)\) to
\begin{align}\label{eq:mu-sol-replaced}
\mu(x) &= 
    \frac{ (pz\frac{d}{dz} + 1) \, \scN_\alpha(z; \sigma=1, d', p=1) }
    { 2 \, \scN_\alpha(z; \sigma=1, d', p=1) }
\end{align} 

The critical algebra is 
\begin{align*}
    & pz\frac{d}{dz} \scN_\alpha(z; \sigma=1, d', p=1)
\\
    &= 
        C p(d'-1) {z}^{d'-1} \, W_{-\alpha,0} \left( -z \right)
        +
        C p({z}^{d'-1}) \, \left( z\frac{d}{dz} W_{-\alpha,0}(-z) \right)
\\
    &= 
    C {z}^{d'-1} \left[
        p(d'-1)  \, W_{-\alpha,0} \left( -z \right)
        +
        \frac{p}{-\alpha} \left( W_{-\alpha,-1}(-z) + W_{-\alpha,0}(-z) \right)
    \right]
\\
    &= 
        -\frac{p}{\alpha} C {z}^{d'-1} W_{-\alpha,-1}(-z)
        +
        \left( pd'-\frac{p}{\alpha}-p \right) \, C {z}^{d'-1} \, W_{-\alpha,0} \left( -z \right)
\end{align*}
We use the following recurrence relation:
\begin{align} \label{eq:mu-numerator-recur-rel}
-\alpha z\frac{d}{dz} W_{-\alpha,0}(-z) 
    &= W_{-\alpha,-1}(-z) + W_{-\alpha,0}(-z)
\end{align}
to move from the second line to the third line above.
(To derive it, set \(\lambda=-\alpha, \mu=0\) 
in \eqref{eq:wright-recurrence} and \eqref{eq:wright-recurrence-diff}).

The \(\mu(x)\) solution becomes
\begin{align*}
\mu(x) 
    &= 
    \frac{1}{2} \left[
        \frac{-p}{\alpha} \frac{W_{-\alpha,-1}(-z)}{W_{-\alpha,0}(-z)}
        +
        \left( pd'-\frac{p}{\alpha}-p \right) 
    \right]
    +
    \frac{1}{2}
\\
    &= 
        -\frac{p}{2\alpha} 
            \frac{W_{-\alpha,-1}\left(-{\left( \frac{x}{\sigma} \right)}^{p}\right)}
                 {W_{-\alpha,0 }\left(-{\left( \frac{x}{\sigma} \right)}^{p}\right)}
        +
        \left( \frac{d}{2} - \frac{p}{2\alpha} \right) 
\end{align*}
which is the intended solution.

\end{proof}

The next question is how to compute \(Q_{\alpha}(z)\).
In addition to going through the series representation of \(W_{-\alpha,-1}(-z)\), 
the following lemma addresses such issue.

\begin{lemma}
There are two methods to compute \(Q_{\alpha}(z)\) depending on the knowledge of
\(F_\alpha(z) \) or \(M_\alpha(z) \). They are
\begin{align} \label{eq:mu-wright-ratio-f}
Q_{\alpha}(z) =
-\frac{W_{-\alpha,-1}\left(-z\right)}
      {W_{-\alpha,0 }\left(-z\right)}
    &=  
    1 + \frac{ \alpha z \, \frac{d}{dz} F_\alpha(z) }{F_\alpha(z)}
\\ \label{eq:mu-wright-ratio-m}
    &=  
    (\alpha + 1) + \frac{ \alpha z \, \frac{d}{dz} M_\alpha(z) }{M_\alpha(z)}
\end{align}

The second line \eqref{eq:mu-wright-ratio-m} allows for high-precision implementation
through \(M_\alpha(z) \) and its first derivative \eqref{eq:m-wright-series-diff}, 
both of which have nice analytic properties.

\end{lemma}

\begin{proof}
Note that \eqref{eq:mu-numerator-recur-rel} is the same as 
\begin{align*}
W_{-\alpha,-1}(-z) &=
    -\left( \alpha z\frac{d}{dz} + 1 \right) W_{-\alpha,0}(-z) 
\end{align*}
Divide both sides by \(F_\alpha(z) = W_{-\alpha,0}(-z) \), 
we arrive at \eqref{eq:mu-wright-ratio-f}.

To move from \eqref{eq:mu-wright-ratio-f} to \eqref{eq:mu-wright-ratio-m},
simply replace \(F_\alpha(z)\) by \(\alpha z \,M_\alpha(z) \), and expand
the derivative into it.

\end{proof}

To validate \(\mu(x)\), first we note that, by setting \(\sigma=1, d=1, p=\alpha\), 
it is reduced to that of a SC. 
Secondly, by setting \(\sigma=1/\sqrt{2}, d=1, p/\alpha=2\)
and \(\alpha\) replaced with \(\alpha/2\), 
it is reduced to that of a SV. 
Both are known solutions on the Wikipedia page\cite{Wikipedia_StableCountDistribution}.

\textbf{Alternatively}, when the closed form of \(\scN_\alpha(z; \sigma=1, d', p=1)\) 
in \eqref{eq:mu-sol-replaced} is known, 
it is better to use the log approach to obtain \(\mu(x)\):

\begin{align}
\mu(x) &= 
    \frac{pz}{2} \, \frac{d}{dz} \left[ \log \scN_\alpha(z; \sigma=1, d', p=1) \right] + \frac{1}{2}
\end{align} 

For instance, when \(\alpha = 1/2\), \(\scN_\alpha(z; \sigma=1, d', p=1) =
\frac{C}{2 \sqrt{\pi}} \, z^{d'} \, e^{-z^2/4}\).
We have the closed-form solution:
\begin{align}\label{eq:mu-sol-poly}
\mu_{\alpha = \frac{1}{2}}\left(x\right) &= 
    -\frac{p}{4} {\left( \frac{x}{\sigma} \right)}^{2p} + \frac{d + p}{2}  
\end{align}
which is reduced to the known results for SC and SV:
\begin{align}
\mu_{\alpha = \frac{1}{2}}\left(x\right) &= 
    \left\{ 
    \begin{array}{ll}
        \frac{1}{8} (6-x) , &\quad \text{for} \, \text{SC};
    \\ 
        1 -\frac{x^2}{2}  , &\quad \text{for} \, \text{SV}.
    \end{array}
    \right.
\end{align}
Both agree with what have been derived on the Wikipedia page\cite{Wikipedia_StableCountDistribution}.

\subsection{Generation of Random Variables for FCM}\label{section-fcm-rv}

Obviously, what really matters for GSaS is the solution of FCM, denoted as \(\mu_{\alpha, k}(x)\).
Note that from this point on, \(\alpha \in (0,2)\).

To further simplify the symbology for FCM, define
\begin{align} \notag
Q_{\alpha}^{(\chi)}(z)  
    &= Q_{\frac{\alpha}{2}}\left( z^{\alpha} \right)
\end{align}
Assume \(k > 0\), we set \(\sigma=\sigma_{\alpha,k}, d=k-1, p/\alpha=2\)
and \(\alpha\) replaced with \(\alpha/2\) in \eqref{eq:mu-sol-adv}. 
We get
\begin{align} \label{eq:mu-fcm-gsas}
\mu_{\alpha, k}(x)  
    &= Q_{\alpha}^{(\chi)}\left( \frac{x}{\sigma_{\alpha,k}} \right)
        +
        \left( \frac{k-3}{2} \right) 
\end{align}
For validation, \(\mu_{1,k}\left(x\right) = k (1-x^2) / 2\) can be used to simulate Student's t. 
And \(\mu_{\alpha, 1}(x)\) provides a method to simulate an SaS \(L_{\alpha,1}(x)\):
\begin{align*} 
\mu_{\alpha, 1}(x)  &= 
         Q_{\alpha}^{(\chi)}\left( \sqrt{2} x \right) -1
\end{align*}

Fig. \ref{fig:sp500-sim} shows a simulation of random variables based on the \((\alpha,k)\) parameters 
obtained from the fit of the S\&P 500 daily log returns.
The rest of parameters are in the caption of the figure. 
First, as outlined above, \(\mu_{\alpha, k}(s)\) is calculated analytically, as shown in the right chart.
Second, it enables the GSC simulation \(\{S_t\}\) as shown in the left chart.
Third, GSaS \(\{X_t\}\) is simulated via 
\( X_t = \mcN / S_t \), where \(\mcN\) is drawn from a standard normal variable.

Simulation is performed daily. Sampling duration is 200,000 years. 
The red areas are the histograms from simulated data. 
The blue lines are from theoretical formulas. They match nicely.

\begin{figure}[htp]
    \centering
    \includegraphics[width=6in]{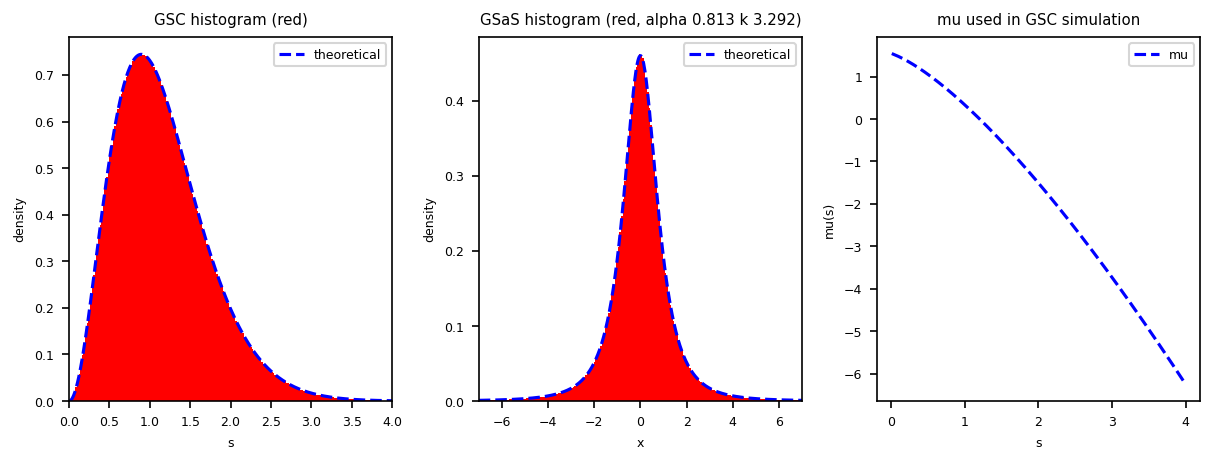}
    \caption{Simulation of random variables based on the \((\alpha,k)\) parameters 
    obtained from the fit of the S\&P 500 daily log returns. 
    The red areas are the histograms from simulated data. 
    The blue lines are from theoretical formulas.
    The settings of the simulation are \(\alpha = 0.813, k = 3.292, dt=1/365, \sigma_u = 0.85\). 
    Sampling duration is 200,000 years. The simulation takes 11 minutes in python.
    \(\mu_{\alpha, k}(s)\) is discretized to 0.01 and cached to increase performance. 
    }
    \label{fig:sp500-sim}
\end{figure}

\subsection{Generation of Random Variables for Inverse FCM}\label{section-fcm-inv-rv}

The simulation of GEP is complicated by two options: ratio vs product.
If we go with the ratio distribution, we have to deal 
evaluating \(F_\alpha(z)\) at very large \(z\),
which is technically more difficult.

To simulate \(\gepDist\), we set \(\sigma=1/\sigma_{\alpha,k}, d = -k, p/\alpha=-2\)
and \(\alpha\) replaced with \(\alpha/2\) in \eqref{eq:mu-sol-adv}, we get
\begin{align} \label{eq:mu-fcm-gep}
\mu_{\alpha, -k}(x)  &= 
    - Q_{\alpha}^{(\chi)} \left( {\left( x \,\sigma_{\alpha,k} \right)}^{-1} \right)
    +
    \left( 1 - \frac{k}{2} \right)
    && (k > 0)
\end{align}
The \(x^{-1}\) term is the added complication.

To avoid such complication, it is more straightforward to take the product distribution route, 
that is, \( X_t = S_t \, \mcN \). And \(S_t\) is generated from an inverse FCM:
\begin{align} \label{eq:mu-fcm-inverse}
\mu_{\alpha, -k}^\dagger(x)  &= 
    Q_{\alpha}^{(\chi)} \left( \frac{x} {\sigma_{\alpha,k}} \right)
    +
    \left( \frac{k}{2}-1 \right)
    && (k > 0)
\end{align}

And \(\mu_{\alpha, -1}^\dagger(x)\) provides a method to simulate 
an exponential power distribution \(\mathcal{E}_{\alpha}\):
\begin{align*} 
\mu_{\alpha, -1}^\dagger(x)  &= 
    Q_{\alpha}^{(\chi)} \left( \sqrt{2} x \right)
    - \frac{1}{2}
\end{align*}
As is expected, this is identical to \(\mu(x)\) of SV.

The simplest validation is \(\mu_{1, -1}^\dagger(x) = 1-\frac{x^2}{2}\) which can simulate 
an exponential distribution \(\mathcal{E}_{1}\).
On the other hand, to simulate \(\mathcal{E}_{1}\) by a ratio distribution,
the polynomial solution is \(\mu_{1, -1}(x) = 1/(2 x^2) - 1\).
The divergence at \(x=0\) is a very different behavior.

In summary, \(\mu_{\alpha, k}(x)\) is an \emph{amazing function} that generates
not only \(\chimeanDist\) but also \(\gsasDist\). 
In some cases, they are just simple polynomials. This is very impressive.

%% file: section-summary.tex

\section{Acknowledgement}

I would like to thank Professor John Mulvey at Princeton University for his guidance and unwavering
support. It took many years and several stages of advancements to get to this point 
since we first discussed the heavy-tail distribution in 2017. 
Thank Professor Francesco Mainardi at the University of Bologna,
although I have not met him personally, for his decades-long work on the Wright function.
Thank Professor Petter Kolm at New York University for introducing
the concept of multivariate elliptical distribution.

My wonderful son, Zachary A. Lihn\footnote{
Zachary A. Lihn. Department of Mathematics, Columbia University, New York, NY 10027, USA.
\href{mailto:zal2111@columbia.edu}{zal2111@columbia.edu}
}
has provided valuable technical support for writing this paper.
I also thank him for several stimulating discussions.

%% file: xapp-gsc-mgf.tex
\section{GSC Moment Generation Function and Implementation}
\label{section:gsc-rv}

\subsection{From GSC MGF to GSaS series representation}\label{section-gsc-mgf}

Since we have the analytic solution of GSC moments \(\E(X^n)\) from \eqref{eq:gsc-moment}, 
the moment generating function (MGF) can be defined from the sum of its moments:
\begin{align*}
M_{\text{\tiny GSC}}(t)
    &:= \E\left( e^{tX} \right) = \sum_{n=0}^{\infty} \frac{\E(X^n)}{n!}
\\
    &=
    \frac{\Gamma(\frac{d}{p} \alpha)} {\Gamma(\frac{d}{p})} \,
    \sum_{n=0}^{\infty} 
    {(\sigma t)}^n \, 
    \frac{\Gamma(\frac{n+d}{p})} { \Gamma(n+1) \Gamma(\frac{n+d}{p} \alpha)}
\\
    &=
    \frac{\Gamma(\frac{d}{p} \alpha)} {\Gamma(\frac{d}{p})} \,
    W \left[ 
        \begin{matrix}
        \frac{1}{p},& \frac{d}{p}
        \\
        \frac{\alpha}{p},& \frac{\alpha d}{p}
        \end{matrix}
    \right](\sigma t)
\end{align*}
Here we utilize the four-parameter Wright function defined in \eqref{eq:wright-fn-4ways}.
The characteristic function of GSC is simply \(\phi_{\text{\tiny GSC}}(z) = M_{\text{\tiny GSC}}(iz)\).


We would like to present an extraordinary approach to derive the small \(x\) 
series representation of GSaS PDF using the MGF of GSC.
The MGF formula above can be transformed to a ratio distribution with a normal variable, 
which resembles how GSaS is constructed in \eqref{eq:main-gsas-def}.

We show that GSC has an inherent closeness that can take
a Laplace transform-like integral to a Gaussian mixture.

\begin{lemma}
The series representation of GSaS PDF in terms of  \(x\) is
\begin{align} \label{eq:gsas-pdf-x-series}
\gsasDist(x) &=
    \frac{1}{\gsasSigma \,\sqrt{\pi}}
    \frac{\Gamma(\frac{k-1}{2})} {\Gamma(\frac{k-1}{\alpha})} \,
\begingroup
\renewcommand*{\arraystretch}{1.2}
    W \left[ 
        \begin{matrix}
        \frac{2}{\alpha},& \frac{k}{\alpha}
        \\
        1,& \frac{k}{2}
        \end{matrix}
    \right] 
\endgroup 
    \left(-\left(\frac{x}{\gsasSigma}\right)^2 \right)
\end{align}
where \(W[...](x)\) is the four-parameter Wright function \eqref{eq:wright-fn-4ways}.

\end{lemma}

\begin{proof}
Equate the series form of MGF with its integral form,  
\begin{align*}
M_{\text{\tiny GSC}}(t)
    &= 
    \int_{0}^{\infty} dx \, e^{tx} \, \gscNx
\\
    &=
    \frac{\Gamma(\frac{d}{p} \alpha)} {\Gamma(\frac{d}{p})} \,
    W \left[ 
        \begin{matrix}
        \frac{1}{p},& \frac{d}{p}
        \\
        \frac{\alpha}{p},& \frac{\alpha d}{p}
        \end{matrix}
    \right](\sigma t)
\end{align*}

Replace \(t\) with \(-\sigma t^2/2\) and \(x = s^2 / \sigma\), we get the Gaussian mixture: 
\begin{align*}
M_{\text{\tiny GSC}}\left(-\frac{\sigma t^2}{2} \right)
    &=
    \int_{0}^{\infty} \, s \, ds \, e^{-(ts)^2/2} \, 
        \left[ \frac{2}{\sigma} \,\scN_\alpha(s^2 / \sigma; \sigma, d, p) \right]
\\
    &=
    \frac{\Gamma(\frac{d}{p} \alpha)} {\Gamma(\frac{d}{p})} \,
    W \left[ 
        \begin{matrix}
        \frac{1}{p},& \frac{d}{p}
        \\
        \frac{\alpha}{p},& \frac{\alpha d}{p}
        \end{matrix}
    \right] \left(-\frac{(\sigma t)^2}{2} \right)
\end{align*}

Note that GSC is close under such transform:
\begin{align*}
\scN_\alpha(s^2 / \sigma; \sigma, d, p) 
    &= \left[
        \frac{|p|}{\sigma} 
        \frac{\Gamma(\frac{\alpha d}{p})} {\Gamma(\frac{d}{p})} 
    \right] \left[ 
        \frac{\sigma}{|2p|}
        \frac{\Gamma(\frac{2d-1}{2p})} {\Gamma(\frac{\alpha (2d-1)}{2p})}
    \right]
    \scN_\alpha(s; \sigma, 2d-1, 2p)
\\
    &= \left[
        \frac{1}{2} 
        \frac{\Gamma(\frac{\alpha d}{p})} {\Gamma(\frac{d}{p})} 
        \frac{\Gamma(\frac{d-1/2}{p})} {\Gamma(\frac{\alpha (d-1/2)}{p})}
    \right]
    \scN_\alpha(s; \sigma, 2d-1, 2p)
\end{align*}

Now we can replace \(t\) with \(x\), and define the new PDF as:
\begin{align*}
f(x) &=
    \int_{0}^{\infty} \, s \, ds \, \mcN(xs) \, 
        \scN_\alpha(s; \sigma, 2d-1, 2p)
\\
    &=
    \frac{\sigma}{\sqrt{2\pi}}
    \frac{\Gamma(\frac{\alpha (d-1/2)}{p})} {\Gamma(\frac{d-1/2}{p})} \,
    W \left[ 
        \begin{matrix}
        \frac{1}{p},& \frac{d}{p}
        \\
        \frac{\alpha}{p},& \frac{\alpha d}{p}
        \end{matrix}
    \right] \left(-\frac{(\sigma x)^2}{2} \right)
\end{align*}

We are in a good place to derive the series representation of GSaS in terms of \(x\),
which is skipped in Section \ref{section:gsas-tail}.
But it takes some tongue-twisting substitutions to accomplish.

To create \(\chimeanDist\) \((k > 0)\) in the integrand of \(f(x)\) above, 
both \(\alpha\) and \(p\) become \(\alpha/2\), thus \(\alpha/p = 1\).
And \(d = k/2, \sigma = \sigma_{\alpha,k} = \sqrt{2}/\gsasSigma\) from \eqref{eq:gsas-pdf-symbols}. 
With these substitutions, 
\(\scN_\alpha(s; \sigma, 2d-1, 2p)\) becomes \(\chimeanDist(s)\); and
\(f(x)\) is transformed to \(\gsasDist(x)\), 
which is \eqref{eq:gsas-pdf-x-series}.
\end{proof}

To validate, note that \(W[...](0) = \Gamma(\frac{k}{\alpha}) / \Gamma(\frac{k}{2})\).
It is straightforward to see that 
\(L_{\alpha,k}(0)\) in \eqref{eq:gsas-pdf-x-series}
is the same as \(L_{\alpha,k}(0)\) in \eqref{eq:gsas-peak-pdf}.

\subsection{PDF Implementation of GSC and Large-x Asymptotic}\label{section-gsc-impl}

The PDF of GSC could be converted to the PDF of SC, which is in turn implemented by the \texttt{scipy.stats.levy{\textunderscore}stable} package.
The direct computation on the Wright function can be safely bypassed. We have
\begin{align*}
\gscNx &= 
        C \, {\left( \frac{x}{\sigma} \right)}^{d-1}
        \, W_{-\alpha,0} \left( -{\nu}^{\alpha} \right) 
        \, , \where \,
        \nu = {\left( \frac{x}{\sigma} \right)}^{p/\alpha} 
\\ &= 
        C \,\Gamma\left( \frac{1}{\alpha}+1 \right) 
        \, {\left( \frac{x}{\sigma} \right)}^{d-1}
        \, \scN_\alpha \left(  {\left( \frac{x}{\sigma} \right)}^{p/\alpha}  \right) 
\end{align*}
which is due to \eqref{eq:sc}.

A naive numeric implementation has its limitation, especially
it can lose numerical precision quickly
when \(d\) is large (e.g. \(d > 60\)).
Since we often want to check the behavior of certain scenarios 
involving \(d \to \infty\), it is highly desirable the implementation
can handle such scenario. 

For the asymptotic at large \(x\), especially when \(d\) is also large, 
\(\scN_\alpha(...)\) could be replaced with its asymptotic 
from Section 4.4.2 of \cite{Lihn:2020}:
\begin{align*}
\scN_\alpha(\nu) &\approx 
    B(\alpha) \, \nu^{\frac{\alpha}{2(1-\alpha)}} 
    e^{-A(\alpha) \, \nu^{\frac{\alpha}{1-\alpha}}}
    , \quad \nu \gg 1
\\ & \where \,
    A(\alpha) = (1-\alpha) \, \alpha^{\frac{\alpha}{1-\alpha}},
    B(\alpha) = 
        \frac{B'(\alpha)}{\Gamma\left( \frac{1}{\alpha}+1 \right)},
    B'(\alpha) =
        \frac{\alpha^{\frac{1}{2(1-\alpha)}}}{\sqrt{2\pi (1-\alpha)}}  
\end{align*}
However, this formula doesn't work for small \(\alpha\) because of another competing limit 
per Section \ref{section:mapping-gg-0}.
This is embedded in the fact that \(A(\alpha)\) and \(B(\alpha)\) are undefined at \(\alpha=0\).

Therefore, as long as \(\alpha\) is not too small, GSC at large \(x\) becomes a GG-variant:
\begin{align}
\gscNx &\approx \notag 
        B(\alpha) C  \,\Gamma\left( \frac{1}{\alpha}+1 \right) 
        \, {\left( \frac{x}{\sigma} \right)}^{d-1}
        \, \nu^{\frac{\alpha}{2(1-\alpha)}} 
        \, e^{-A(\alpha) \, \nu^{\frac{\alpha}{1-\alpha}}}
        , \quad \nu = {\left( \frac{x}{\sigma} \right)}^{p/\alpha}
\\ &= \label{eq:gsc-asymp}
        B'(\alpha) C \,  
        \, {\left( \frac{x}{\sigma} \right)}^{d + \frac{p}{2(1-\alpha)} -1}
        \, \exp\left({-A(\alpha) \, {\left( \frac{x}{\sigma} \right)}^{\frac{p}{1-\alpha}}}\right)
\end{align}

\subsection{Large-x Asymptotic of FCM}\label{section-fcm-asymp}

Large-x Asymptotic of FCM follows naturally from \eqref{eq:main-chimean-pdf} and \eqref{eq:gsc-asymp}.
For an FCM,
\begin{align*}
C &= 
    \frac{\alpha}{ \fcmSigma } 
    \frac{\Gamma(\frac{k-1}{2})}{\Gamma(\frac{k-1}{\alpha})}
    , \quad \where
    \fcmSigma := 
        \frac{{|k|}^{1/2 - 1/\alpha}}{\sqrt{2}}  
\end{align*}

Therefore, at large \(x\), an FCM of \(k > 0\) becomes 
\begin{align}
\chimean{x} &\approx \label{eq:fcm-asymp}
        B'\left(\frac{\alpha}{2}\right) C \,
        \, {\left( \frac{x}{\fcmSigma} \right)}^{k + \frac{\alpha}{2-\alpha} -2}
        \, \exp\left(
            -A\left(\frac{\alpha}{2}\right) \, 
            {\left( \frac{x}{\fcmSigma} \right)}^{\frac{2\alpha}{2-\alpha}}
        \right)
\end{align}
This formula provides insight to the tail behavior of the FCM.

As a validation, the asymptotic formula becomes exact at \(\alpha=1\),
where \(\chimeanDist\) becomes \(\chi_k/\sqrt{k}\).

%% file: xapp-subord.tex
\section{Subordination Results}\label{section:subordination}

\begin{lemma}\label{lemma-sub-pdf-cf}
Assume \(f(x), g(x), \scN(x)\) are the PDFs from three distributions. 
\(f(x), g(x)\) are two-sided; and \(\scN(x)\) is one-sided. 
They form the following pair of transform:

\begin{align}\label{eq:sub-pdf-cf-pair}
f(x)
    &= \int_{0}^{\infty} s \, ds \,
      g(xs) \, \scN(s)
    && (x \in \mathbb{R})
\\ \label{eq:sub-pdf-cf-pair2}
\text{CF}\{f\}(\zeta)
    &= \int_{0}^{\infty} \, ds \,
      \left[ \text{CF}\{g\}\left(\frac{\zeta}{s}\right) \right]
      \, \scN(s)
    && (\zeta \in \mathbb{R})
\\ \label{eq:sub-pdf-cf-pair2a}
    &= \int_{0}^{\infty} s \, ds \,
      \left[ \text{CF}\{g\}\left(\zeta s\right) \right]
      \scN_{\phi}(s),
    && \where 
      \scN_{\phi}(s) := s^{-3} \, \scN\left(\frac{1}{s}\right)
\end{align}
This is an extended form of Lemma \ref{lemma-mix-cf-dist}.

\end{lemma}

\begin{proof}
For a given distribution \(f(x)\), its CF transform is
\( \text{CF}\{f\}(\zeta) := \int_{-\infty}^{\infty} dx  \, \exp(i\zeta x) \, f(x)\). 
And for \(g(x)\), its CF transform is
\( \text{CF}\{g\}(\zeta) := \int_{-\infty}^{\infty} dx  \, \exp(i\zeta x) \, g(x)\). 
Apply them to \eqref{eq:sub-pdf-cf-pair}, we get 
\begin{align} \label{eq:sub-pdf-f-integral}
\text{CF}\{f\}(\zeta)
    &= \int_{0}^{\infty} s \, ds \,
    \left( \int_{-\infty}^{\infty} dx  \, \exp(i\zeta x) \,  g(xs) \right)
      \, \scN(s)
\end{align}
Let \(u = xs\), the integral inside the parenthesis becomes
\begin{align*}
\int_{-\infty}^{\infty} dx  \, \exp\left( i\zeta x \right) \,  g(xs)
    &= s^{-1} \int_{-\infty}^{\infty} du  \, \exp\left(i \frac{\zeta u}{s} \right) \,  g(u)
\\ 
    &= s^{-1} \, \text{CF}\{g\}\left(\frac{\zeta}{s}\right)
\end{align*}
Put it back to \eqref{eq:sub-pdf-f-integral}, we arrive at \eqref{eq:sub-pdf-cf-pair2}.

To move from the product distribution form of \eqref{eq:sub-pdf-cf-pair2} 
to the ratio distribution form of \eqref{eq:sub-pdf-cf-pair2a}, let \(t = 1/s\),
then \(ds / s = - dt / t\). Furthermore, \(s \, ds = -t^{-3} \, dt\).
It is straightforward to obtain \eqref{eq:sub-pdf-cf-pair2a}.

\end{proof}

The continuous Gaussian mixture in Section \ref{section:mixture} is 
a special case where \(g(x) = \mcN(x)\). 
The CF of \(\mcN(x)\) is still itself: \(\text{CF}\{\mcN\}(\zeta) = \sqrt{2\pi} \, \mcN(\zeta)\).
Another case we encountered is the Laplace-Cauchy pair: 
\(\text{CF}\{L\}(\zeta) = \pi \,\text{Cauchy}(\zeta)\).

If \(g(x)\) and \(\text{CF}\{g\}(\zeta)\) are well understood functions, 
this lemma captures many relations
between \(f(x)\) and \(\scN(x)\), for instance, when \(\scN(x)\) is 
a high transcendental function such as GSC.

When \(\scN(x)\) is a GSC, it is easy to see that \(\scN_{\phi}(s)\) is also a GSC.
Therefore, GSC is closed under the CF transform: \(\scN(x) \to \scN_{\phi}(s)\).

We studied many of the product and ratio distributions and tabulated the results as following.

\subsection{Product Distribution Results}\label{section:product-dist}

Table \ref{tab:subord-product-dist} below illustrates how a classic distribution on the left can
be decomposed to a GSC on the right via the product distribution relation:
\begin{align}\label{eq:product-dist}
\text{Left}(x) 
    &:= \int_{0}^{\infty} 
       \frac{dt}{t} \, \text{Unit}(\frac{x}{t}) \, \text{Right}(t)
\end{align}

The discovery of Table \ref{tab:subord-product-dist} directly paved the way to 
generalize the stable count distribution.

\addtocounter{table}{-1}
\begin{table}[H]
    \smaller  
    \centering
    \renewcommand{\arraystretch}{1.5}
\begin{longtable}{|p{3cm}|p{3cm}|p{2.5cm}|l|l|l|l|}
\hline
\multicolumn{1}{|c|}{} & \multicolumn{1}{|c|}{} & \multicolumn{1}{|c|}{}
    & \multicolumn{4}{c|}{Right GSC Equiv.: \(\gscN{t}\)} \\
\cline{4-7}
\multicolumn{1}{|c|}{Left Dist. (PDF)} 
    & \multicolumn{1}{|c|}{Unit Dist.} 
    & \multicolumn{1}{|c|}{Right Dist.} 
    & \(\alpha\) & \(\sigma\) & \(d\) & \(p\) \\
\hline
Exp Pow: \(\mathcal{E}_\alpha(z)\) 
    & Laplace: \(L(\frac{\mid z \mid}{t})\)
    & SC: \(\scN_\alpha(t)\) 
    & \(\alpha\) & 1 & 1 & \(\alpha\) \\
Exp Pow: \(\mathcal{E}_\alpha(z)\) 
    & Normal: \(\mcN(z/t)\)  
    & SV: \(V_\alpha(t)\)
    & \(\alpha/2\) & \(\frac{1}{\sqrt{2}}\) & 1 & \(\alpha\) \\
Exp Pow: \(\mathcal{E}_\alpha(z)\) 
    & \(\mathcal{E}_k(z/t)\), any \(k\in \mathbb{Z}\) 
    & GSC & \({\alpha}/{k}\) & 1 & 1 & \(\alpha\) \\
Weibull: \(\text{Wb}(x;k)\) & \(L(\frac{x}{t})\) aka \(\text{Wb}(\frac{x}{t});1)\) 
    & GSC: \(t^{-1} \scN_k(t)\)
    & \(k\) & 1 & 0 & \(k\) \\
Weibull: $\text{Wb}(x;k) $ & \(\text{Rayleigh}(\frac{x}{t})\) 
        \newline aka \(\text{Wb}(\frac{x}{t});2)\) 
    & \(\chi_{k,1}\): \(t^{-1} V_k(t)\) 
    & \({k}/{2}\) & \(\frac{1}{\sqrt{2}}\) & 0 & \(k\) \\
Gamma: \(\Gamma(x;s)\) 
    & \(\text{Wb}(\frac{x}{t};s)\) 
    & GSC & \(1/s\) & 1 & \(s\) & 1 \\
Poisson: \(f(k;\lambda)\) \newline\(= \Gamma(\lambda;s=\lfloor k+1 \rfloor)\) 
    & \(\text{Wb}(\frac{\lambda}{t};k+1)\) 
    & GSC & \(1/(k+1)\) & 1 & \(k+1\) & 1 \\
\(\chi_k^2\): \(\Gamma(\frac{x}{2};\frac{k}{2})\) &
    $\text{Wb}(\frac{x}{t};\frac{k}{2}) $ 
    & GSC & \(2/k\) & 2 & \(k/2\) & 1 \\
\(\chi_k\): \(\Gamma(\frac{x^2}{2}; \frac{k}{2})\) &
    $\text{Wb}(\frac{x}{t};k) $ 
    & GSC & \(2/k\) & \(\sqrt{2}\) & \(k\) & 2 \\
\(\text{GenGamma}(\frac{x}{\sigma}; s,c)\) &
    $\text{Wb}(\frac{x}{t};sc) $ 
    &GSC  & \(1/s\) & \(\sigma\) & \(sc\) & \(c\) \\
\hline
\end{longtable}
\caption{\label{tab:subord-product-dist}Subordination by a product distribution to GSC.
}
\end{table}

\subsection{Ratio Distribution Results}\label{section:ratio-dist}

Table \ref{tab:subord-ratio-dist} below illustrates how a Student's t and/or SaS distribution on the left can
be decomposed to a GSC on the right via a ratio distribution:
\begin{align}\label{eq:ratio-dist}
\text{Left}(x) 
    &:= \int_{0}^{\infty} 
       s\,ds \, \text{Unit}(sx) \, \text{Right}(s)
\end{align}

We've also included the new GSaS from this work. This table shows that we should be flexible
in choosing either the path of product or ratio distributions. 
As Lemma \ref{lemma-sub-pdf-cf} shows,
they are really the flip side of the same coin.

\addtocounter{table}{-1}
\begin{table}[H]
    \smaller 
    \centering
\renewcommand{\arraystretch}{1.5}
\begin{longtable}{|p{3.5cm}|p{2.75cm}|p{3cm}|l|l|l|l|}
\hline
\multicolumn{1}{|c|}{} & \multicolumn{1}{|c|}{} & \multicolumn{1}{|c|}{}
    & \multicolumn{4}{c|}{Right GSC: \(\gscNs\)} \\
\cline{4-7}
\multicolumn{1}{|c|}{Left Dist. (PDF)} 
    & \multicolumn{1}{|c|}{Unit Dist.} 
    & \multicolumn{1}{|c|}{Right Dist.} 
    & \(\alpha\) & \(\sigma\) & \(d\) & \(p\) \\
\hline
\({\text{IG}}(x;k)\)
    & \(\text{IWb}(sx; k)\) & GSC & \(1/k\) & 1 & \(k\) & \(k\) \\
\({\text{IWb}}(x;\alpha)\)
    & \(\text{IWb}(sx;k)\),\newline any \(k\in \mathbb{Z}\) 
    & GSC & \(\alpha/k\) & 1 & 0 & \(\alpha\) \\
Student's t: \(t_k(x)\) 
    & Normal: \(\mcN(sx)\) 
    & \(\chibar_{1,k}\): \(s^{k-2} V_{1}(\sqrt{k} s)\) 
    & \(1/2\) & \(\frac{1}{\sqrt{2k}}\) & \(k-1\) & 1 \\
SaS: \(P_\alpha(x)\) 
    & \(\text{Cauchy}(s x)\) 
    & GSC: \(s^{-1} \scN_\alpha(s)\)
    & \(\alpha\) & 1 & 0 & \(\alpha\) \\
SaS: \(P_\alpha(x)\) 
    & Normal: \(\mcN(sx)\)
    & \(\chibar_{\alpha,1}\): \(s^{-1} V_\alpha(s)\)
    & \(\alpha/2\) & \(\frac{1}{\sqrt{2}}\) & 0 & \(\alpha\) \\
GSaS: \(\gsas{x}\)
    & Normal: \(\mcN(sx)\) 
    & \(\chibar_{\alpha,k}\): \(s^{k-2} V_{\alpha}(\sqrt{k} s)\) 
    & \(\alpha/2\) & \(\frac{1}{\sqrt{2k}}\) & \(k-1\) & \(\alpha\) \\
ML: \(E_\alpha(-x)\) 
    & Expon: \(\exp(-sx)\) 
    & \(s^{-1} M_\alpha(s)\) 
    & \(\alpha\) & 1 & \(-1^{+}\) & 1 \\
ML: \(E_{\frac{\alpha}{2}}(-\frac{x^2}{2})\) 
    & Normal: \(\mcN(sx)\) 
    & GSC: \(M_{\frac{\alpha}{2}}(s^2)\) 
    & \(\alpha/2\) & 1 & -1 & 2 \\

\hline
\end{longtable}
\caption{\label{tab:subord-ratio-dist}Subordination by a ratio distribution to GSC.
Note on the Mittag-Leffler (ML) relations: The first relation is 
a Laplace transform, neither side is normalized properly.
The LHS of the second relation needs a normalization constant of 
\(\frac{\Gamma(1-\alpha/4)}{\sqrt{2} \,\pi}\).
}
\end{table}

%% file: xapp-wright.tex
\section{Useful Results on the Wright Function}
\label{section:formula-wright}

This section collects useful results of the Wright function that are used in this paper.
It is expanded from Section \ref{section:main-recap-wright}.

\textbf{The moments of the Wright function} are (See (1.4.28) of \cite{Mathai:2017})

\begin{align}\label{eq:wright-moments}
\E(X^{d-1}) =
\int_{0}^{\infty} x^{d-1} W_{-\lambda,\delta}(-x) dx 
    &= \frac{\Gamma(d)} {\Gamma(d\lambda+\delta)}
\end{align}

\textbf{The recurrence relations} of the Wright function are (Chapter 18, Vol 3 of \cite{Bateman:1955})
\begin{align}\label{eq:wright-recurrence}
\lambda z \, W_{\lambda, \lambda+\mu}(z)  
    &= W_{\lambda, \mu-1}(z)  + (1-\mu) W_{\lambda, \mu}(z) 
\\ \label{eq:wright-recurrence-diff}
\frac{d}{dz} W_{\lambda,\mu}(z) &= W_{\lambda, \lambda+\mu}(z) 
\end{align}

\textbf{The four-parameter Wright function} is defined as

\begin{align}\label{eq:wright-fn-4ways}
    W \left[ 
        \begin{matrix}
        a,& b
        \\
        \lambda,& \mu
        \end{matrix}
    \right](z)
    &:=
    \sum_{n=0}^{\infty} 
    \frac{{z}^n}{ \Gamma(n+1) } \, 
    \frac{\Gamma(an+b)} { \Gamma(\lambda n+\mu)}
\end{align}
As far as I know, this function is used seriously for the first time in this work.

\subsection{The M-Wright Functions}\label{section:formula-m-wright}

Mainardi introduced two auxiliary functions of the Wright type (See F.2 of \cite{Mainardi:2010}):
\begin{align}\label{eq:m-wright-F}
F_\alpha(z) &:= W_{-\alpha,0}(-z)  & (z > 0)
\\ \label{eq:m-wright-M}
M_\alpha(z) &:= W_{-\alpha,1-\alpha}(-z) = \frac{1}{\alpha z} F_\alpha(z) & (z > 0)
\end{align}
The relation between \(M_\alpha(z)\) and \(F_\alpha(z)\) in \eqref{eq:m-wright-M} 
is an application of \eqref{eq:wright-recurrence} by setting \(\lambda= -\alpha, \mu=1\).

\(F_\alpha(z)\) has the following Hankel integral representation:
\begin{align}\label{eq:m-wright-F-hankel}
F_\alpha(z) &= W_{-\alpha,0}(-z) 
    = \frac{1}{2\pi i} \int_{H} dt \, \exp{(t-z t^\alpha)}
\end{align}

\(F_\alpha(z)\) is used to define GSC. 
But its series representation isn't very useful computationally.
It requires a lot more terms to converge to a prescribed precision.

On the other hand, 
\(M_\alpha(z)\) has a more computationally friendly series representation, 
especially for small \(\alpha\)'s:
\begin{align}\label{eq:m-wright-series}
M_\alpha(z) &= 
    \sum_{n=0}^\infty
        \frac{{(-z)}^n}{n!\,\Gamma(-\alpha n + (1-\alpha))} 
    = \frac{1}{\pi} \sum_{n=1}^\infty
        \frac{{(-z)}^{n-1}}{(n-1)!} \Gamma(\alpha n) \sin(\alpha n \pi)
    && (0 < \alpha < 1)
\end{align}
\(M_\alpha(z)\) also has very nice analytic properties at \(\alpha = 0, 1/2\),
where \(M_0(z) = \exp(-z)\) and \(M_{\frac{1}{2}}(z) = \frac{1}{\sqrt{\pi}} \exp(-z^2/4)\).

\(M_\alpha(z)\) can be computed to high accuracy when properly implemented with arbitrary-precision 
floating-point library, such as the \texttt{mpmath} package. 
In this regard, it is much more "useful" than \(F_\alpha(z)\).
\(L_\alpha(x)\) (therefore, \(\scN_\alpha(x)\)) can be computed via
\(L_\alpha(x) = \alpha x^{-\alpha-1} M_{\alpha}(x^{-\alpha})\)
without relying on the \texttt{scipy.stats.levy{\textunderscore}stable} package.

This is particularly important in working with large degrees of freedom 
and extreme values of \(\alpha\), mainly at 0 and 1. Typical 64-bit floating-point
is quickly overflowed.

\(M_\alpha(z)\) has the asymptotic representation in GG-style: (See F.20 of \cite{Mainardi:2010})
\begin{align} 
M_\alpha\left(\frac{x}{\alpha}\right) &= 
    A \, x^{d-1} \, e^{-B \, x^{p}}
\\ \notag
    & \where p = 1 / (1-\alpha), d = p/2, \,\,
    A = \sqrt{p/(2\pi)}, \,\,
    B = 1 /(\alpha p).
\end{align}
This formula is important in guiding the series sum to high precision.

Treating \(M_\alpha(x)\) as the PDF of a one-sided distribution such that 
\(\int_0^\infty M_\alpha(x) dx = 1\), then its \(n\)-th moment is
(See Section 4 of \cite{Mainardi:2020})
\begin{align}\label{eq:m-wright-mnt}
\E(X^n|M_\alpha) = 
    \int_0^\infty x^n M_\alpha(x) dx &= 
    \frac{\Gamma(n+1)}{\Gamma(n\alpha + 1)}
\end{align}

Differentiating \(M_\alpha(z)\), and from \eqref{eq:m-wright-series}, we get
\begin{align}\label{eq:m-wright-series-diff}
\frac{d}{dz} M_\alpha(z) &= 
    -W_{-\alpha,1-2\alpha}(-z)
    =
    \frac{-1}{\pi} \sum_{n=2}^\infty
        \frac{{(-z)}^{n-2}}{(n-2)!} \Gamma(\alpha n) \sin(\alpha n \pi)
\end{align}
Note that \(\frac{d}{dz} M_\alpha(0) = -\frac{1}{\pi} \Gamma(2\alpha) \sin(2\alpha \pi)\).
This also indicates that
\begin{align}\label{eq:f-wright-series-diff}
\frac{d}{dz} F_\alpha(z) &= 
    \alpha \left( 1 + z \frac{d}{dz}\right) M_\alpha(z)
\end{align}
which can be implemented from \(M_\alpha(z)\) 
through \eqref{eq:m-wright-series} and \eqref{eq:m-wright-series-diff}.
These results are important for the implementation of 
\(Q_\alpha(z)\) in \eqref{eq:mu-wright-ratio-m}.

%% file: xapp-formula.tex
\section{List of Useful Formula}
\label{section:formula}

\subsection{Gamma Function}\label{section:formula-gamma}

Gamma function is used extensively in this paper. First, note that 
\(\Gamma(\frac{1}{2}) = \sqrt{\pi}\). 
Its \textbf{reflection formulas} are:
\begin{align}
\label{eq:gamma-refl-sin}
\Gamma(z) \Gamma(1-z)           &= \frac{\pi}{\sin(\pi z)} 
\\ \label{eq:gamma-reflect}
\Gamma(z) \Gamma(z+\frac{1}{2}) &= 2^{1-2z} \sqrt{\pi} \, \Gamma(2z)
\end{align}

\textbf{Gamma function Asymptotic}:
At \(x \to 0\), gamma function becomes 
\begin{align}\label{eq:gamma-zero}
\lim_{x \to 0} \Gamma(x) &\sim \frac{1}{x}
\\ \notag
\lim_{x \to 0} \frac{\Gamma(ax)} {\Gamma(bx)} &= \frac{b}{a}
    \quad \quad (ab \ne 0)
\end{align}

For a very large \(x\), assume \(a,b\) are finite, 
\begin{align}\label{eq:gamma-infty}
\lim_{x \to \infty} \frac{\Gamma(x+a)}{\Gamma(x+b)} \sim x^{a-b}
\end{align}

\textbf{Sterling's formula} is used to expand the kurtosis formula for a large \(k\), which is:
\begin{align}\label{eq:gamma-sterling}
\lim_{x \to \infty} \Gamma(x+1) &\sim \sqrt{2\pi x} {\left(\frac{x}{e}\right)}^x, 
\\
\,\,\text{or} \,\,
\lim_{x \to \infty} \Gamma(x) &\sim \sqrt{2\pi} \, x^{x-1/2} e^{-x}.
\end{align}

\subsection{Transformation}\label{section:formula-transform}

Laplace transform of cosine is\footnote{
See \url{https://proofwiki.org/wiki/Laplace_Transform_of_Cosine}
}
\begin{align}\label{eq:int-cosine}
\int_{0}^{\infty} dt \, \cos(xt) e^{-t/\nu}
    &= 
    \frac{\nu^{-1}}{x^2 + \nu^{-2}}
    = 
    \frac{\nu}{{(\nu x)}^2 + 1}
\end{align}

Gaussian transform of cosine is\footnote{
See \url{https://www.wolframalpha.com/input?i=integrate+cos\%28a+x\%29+e\%5E\%28-x\%5E2\%2F2\%29+dx+from+0+to+infty}
}
\begin{align} \label{eq:gauss-cosine}
\int_{0}^{\infty} dt \, \cos(xt) \, e^{-t^2/2}
    &= \sqrt{\frac{\pi}{2}} \, e^{-x^2/2}
\\  \notag
\text{Hence} \,
\int_{0}^{\infty} dt \, \cos(xt) \, e^{-t^2/2s^2}
    &= \sqrt{\frac{\pi}{2}} s \, e^{-(sx)^2/2}
\end{align}

\subsection{Half-Normal Distribution}\label{section:formula-half-normal}

The moments of the half-normal distribution (HN)\footnote{
See \url{https://en.wikipedia.org/wiki/Half-normal_distribution}
} are used several times.
Its PDF is defined as
\begin{align} 
  p_{HN}(x;\sigma) &:= 
    \sqrt{\frac{2}{\pi}} \frac{1}{\sigma} 
    \, e^{ -x^2/{(2\sigma^2)} } 
    \, , \,\, x > 0
\end{align}
which is a special case of GG with \(d=1, p=2, a=\sqrt{2} \sigma\). 
Its moments are
\begin{align}\label{eq:HN-moments}
E_{HN}(T^n) &= \sigma^n
    \frac{ 2^{n/2} }{ \sqrt{\pi} }
    \, \Gamma\left(\frac{n+1}{2}\right)
\end{align}
which are the same as those of a normal distribution.